\documentclass[english,11pt, letter]{article}
\usepackage[T1]{fontenc}
\usepackage[utf8]{luainputenc}
\usepackage{geometry}
\usepackage{float}
\usepackage{amsthm}
\usepackage{amsmath}
\usepackage{amssymb}
\usepackage{esint}

\makeatletter

\providecommand{\tabularnewline}{\\}

\numberwithin{equation}{section}
\numberwithin{figure}{section}
\theoremstyle{plain}
\newtheorem{thm}{\protect\theoremname}
  \theoremstyle{definition}
  \newtheorem{defn}[thm]{\protect\definitionname}
  \theoremstyle{plain}
  \newtheorem{lem}[thm]{\protect\lemmaname}
  \theoremstyle{remark}
  \newtheorem{rem}[thm]{\protect\remarkname}

\pdfoutput=1

\usepackage{amsmath, amsthm, amssymb, amsfonts,bbold}
\usepackage{bbm}
\usepackage[basic]{complexity}
\usepackage[pdftex,bookmarks,colorlinks]{hyperref}
\usepackage{enumitem}
\usepackage{color}

\DeclareMathOperator*{\argmaxTex}{arg\,max}
\DeclareMathOperator*{\argminTex}{arg\,min}
\DeclareMathOperator*{\signTex}{sign}
\DeclareMathOperator*{\rankTex}{rank}
\DeclareMathOperator*{\diagTex}{diag}
\DeclareMathOperator*{\imTex}{im}

\hypersetup{colorlinks=true,citecolor=red}

\renewcommand{\varepsilon}{\epsilon}

\makeatother

\usepackage{babel}
  \providecommand{\definitionname}{Definition}
  \providecommand{\lemmaname}{Lemma}
  \providecommand{\remarkname}{Remark}
\providecommand{\theoremname}{Theorem}

\begin{document}

\global\long\def\R{\mathbb{R}}
 \global\long\def\Rn{\mathbb{R}^{n}}
 \global\long\def\Rm{\mathbb{R}^{m}}
 \global\long\def\Rmn{\mathbb{R}^{m \times n}}
 \global\long\def\Rnm{\mathbb{R}^{n \times m}}
 \global\long\def\Rmm{\mathbb{R}^{m \times m}}
 \global\long\def\Rnn{\mathbb{R}^{n \times n}}
 \global\long\def\Z{\mathbb{Z}}
 \global\long\def\rPos{\R_{> 0}}
 \global\long\def\dom{\mathrm{dom}}
\global\long\def\dInterior{\Omega^{0}}
\global\long\def\dFull{\{\dInterior\times\rPos^{m}\}}
\global\long\def\dWeight{\rPos^{m}}


\global\long\def\ellOne{\ell_{1}}
 \global\long\def\ellTwo{\ell_{2}}
 \global\long\def\ellInf{\ell_{\infty}}
 \global\long\def\ellP{\ell_{p}}

\global\long\def\otilde{\widetilde{O}}

\global\long\def\argmax{\argmaxTex}

\global\long\def\argmin{\argminTex}

\global\long\def\sign{\signTex}

\global\long\def\rank{\rankTex}

\global\long\def\diag{\diagTex}

\global\long\def\im{\imTex}

\global\long\def\enspace{\quad}

\global\long\def\boldVar#1{\mathbf{#1}}

\global\long\def\mvar#1{\boldVar{#1}}

\global\long\def\vvar#1{\vec{#1}}



\global\long\def\defeq{\stackrel{\mathrm{{\scriptscriptstyle def}}}{=}}

\global\long\def\diag{\mathrm{{diag}}}

\global\long\def\mDiag{\mvar{diag}}
 \global\long\def\ceil#1{\left\lceil #1 \right\rceil }

\global\long\def\E{\mathbb{E}}
 \global\long\def\abs#1{\left|#1\right|}

\global\long\def\gradient{\nabla}
\global\long\def\grad{\gradient}
 \global\long\def\hessian{\gradient^{2}}
 \global\long\def\hess{\hessian}
 \global\long\def\jacobian{\mvar J}
 \global\long\def\gradIvec#1{\vvar{f_{#1}}}
 \global\long\def\gradIval#1{f_{#1}}


\global\long\def\onesVec{\vec{\mathbb{1}}}
 \global\long\def\indicVec#1{\onesVec_{#1}}
\global\long\def\indic{1}


\global\long\def\specGeq{\succeq}
 \global\long\def\specLeq{\preceq}
 \global\long\def\specGt{\succ}
 \global\long\def\specLt{\prec}

\global\long\def\va{\vvar a}
 \global\long\def\vb{\vvar b}
 \global\long\def\vc{\vvar c}
 \global\long\def\vd{\vvar d}
 \global\long\def\ve{\vvar e}
 \global\long\def\vf{\vvar f}
 \global\long\def\vg{\vvar g}
 \global\long\def\vh{\vvar h}
 \global\long\def\vl{\vvar l}
 \global\long\def\vm{\vvar m}
 \global\long\def\vn{\vvar n}
 \global\long\def\vo{\vvar o}
 \global\long\def\vp{\vvar p}
 \global\long\def\vs{\vvar s}
 \global\long\def\vu{\vvar u}
 \global\long\def\vv{\vvar v}
 \global\long\def\vx{\vvar x}
 \global\long\def\vy{\vvar y}
 \global\long\def\vz{\vvar z}
 \global\long\def\vxi{\vvar{\xi}}
 \global\long\def\valpha{\vvar{\alpha}}
 \global\long\def\veta{\vvar{\eta}}
 \global\long\def\vphi{\vvar{\phi}}
\global\long\def\vpsi{\vvar{\psi}}
 \global\long\def\vsigma{\vvar{\sigma}}
 \global\long\def\vgamma{\vvar{\gamma}}
 \global\long\def\vphi{\vvar{\phi}}
\global\long\def\vDelta{\vvar{\Delta}}
\global\long\def\vzero{\vvar 0}
 \global\long\def\vones{\vvar 1}

\global\long\def\ma{\mvar A}
 \global\long\def\mb{\mvar B}
 \global\long\def\mc{\mvar C}
 \global\long\def\md{\mvar D}
 \global\long\def\mf{\mvar F}
 \global\long\def\mg{\mvar G}
 \global\long\def\mh{\mvar H}
 \global\long\def\mj{\mvar J}
 \global\long\def\mk{\mvar K}
 \global\long\def\mm{\mvar M}
 \global\long\def\mn{\mvar N}
 \global\long\def\mq{\mvar Q}
 \global\long\def\mr{\mvar R}
 \global\long\def\ms{\mvar S}
 \global\long\def\mt{\mvar T}
 \global\long\def\mU{\mvar U}
 \global\long\def\mv{\mvar V}
 \global\long\def\mx{\mvar X}
 \global\long\def\my{\mvar Y}
 \global\long\def\mz{\mvar Z}
 \global\long\def\mSigma{\mvar{\Sigma}}
 \global\long\def\mLambda{\mvar{\Lambda}}
\global\long\def\mPhi{\mvar{\Phi}}
 \global\long\def\mZero{\mvar 0}
 \global\long\def\iMatrix{\mvar I}

\global\long\def\oracle{\mathcal{O}}

\global\long\def\runtime{\mathcal{T}}

\global\long\def\dWeights{\rPos^{m}}


\global\long\def\shurProd{\circ}
 \global\long\def\shurSquared#1{{#1}^{(2)}}

\global\long\def\weight{w}
 \global\long\def\vWeight{\vvar{\weight}}
 \global\long\def\mWeight{\mvar W}

\global\long\def\mProj{\mvar P}

\global\long\def\vLever{\vsigma}
 \global\long\def\mLever{\mSigma}
 \global\long\def\mLapProj{\mvar{\Lambda}}

\global\long\def\penalizedObjective{f_{t}}
 \global\long\def\penalizedObjectiveWeight{\hat{f}}

\global\long\def\fvWeight{\vg}
 \global\long\def\fmWeight{\mg}

\global\long\def\gradW{\grad_{\vWeight}}
 \global\long\def\gradX{\grad_{\vx}}
 \global\long\def\hessWW{\hess_{\vWeight\vWeight}}
 \global\long\def\hessWX{\hess_{\vWeight\vx}}
 \global\long\def\hessXW{\hess_{\vx\vWeight}}
 \global\long\def\hessXX{\hess_{\vx\vx}}

\global\long\def\vNewtonStep{\vh}


\global\long\def\norm#1{\big\|#1\big\|}
 \global\long\def\normFull#1{\left\Vert #1\right\Vert }
 \global\long\def\normA#1{\norm{#1}_{\ma}}
 \global\long\def\normFullInf#1{\normFull{#1}_{\infty}}
 \global\long\def\normInf#1{\norm{#1}_{\infty}}
 \global\long\def\normOne#1{\norm{#1}_{1}}
 \global\long\def\normTwo#1{\norm{#1}_{2}}
 \global\long\def\normLeverage#1{\norm{#1}_{\mSigma}}
 \global\long\def\normWeight#1{\norm{#1}_{\fmWeight}}

\global\long\def\cWeightSize{c_{1}}
 \global\long\def\cWeightStab{c_{\gamma}}
 \global\long\def\cWeightCons{c_{\delta}}

\global\long\def\TODO#1{{\color{red}TODO:\text{#1}}}

\global\long\def\mixedNorm#1#2{\normFull{#1}_{#2+\infty}}
\global\long\def\CNorm{C_{\mathrm{norm}}}
\global\long\def\Pxw{\mProj_{\vx,\vWeight}}
\global\long\def\vq{\vec{q}}
\global\long\def\cnorm{C_{\text{norm}}}

\global\long\def\next#1{#1^{\text{(new)}}}

\global\long\def\trInit{\vx^{(0)}}
 \global\long\def\trCurr{\vx^{(k)}}
 \global\long\def\trNext{\vx^{(k + 1)}}
 \global\long\def\trAdve{\vy^{(k)}}
 \global\long\def\trAfterAdve{\vy}
 \global\long\def\trMeas{\vz^{(k)}}
 \global\long\def\trAfterMeas{\vz}
 \global\long\def\trGradCurr{\grad\Phi_{\alpha}(\trCurr)}
 \global\long\def\trGradAdve{\grad\Phi_{\alpha}(\trAdve)}
 \global\long\def\trGradMeas{\grad\Phi_{\alpha}(\trMeas)}
 \global\long\def\trGradAfterAdve{\grad\Phi_{\alpha}(\trAfterAdve)}
 \global\long\def\trGradAfterMeas{\grad\Phi_{\alpha}(\trAfterMeas)}
 \global\long\def\trSetCurr{U^{(k)}}
\global\long\def\vWeightError{\vvar{\Psi}}
\global\long\def\code#1{\texttt{#1}}

\global\long\def\nnz{\mathrm{nnz}}

\newcommand{\bracket}[1]{[#1]}

\title{Path Finding II :\\
An $\otilde(m\sqrt{n})$ Algorithm for the Minimum Cost Flow Problem}

\author{Yin Tat Lee\\
MIT\\
yintat@mit.edu\and  Aaron Sidford\\
MIT\\
sidford@mit.edu}

\date{}
\maketitle
\begin{abstract}
In this paper we present an $\otilde(m\sqrt{n}\log^{O(1)}U)$ time
algorithm for solving the maximum flow problem on directed graphs
with $m$ edges, $n$ vertices, and capacity ratio $U$. This improves
upon the previous fastest running time of $O(m\min\{m^{1/2},n^{2/3}\}\log\left(n^{2}/m\right)\log(U))$
achieved over 15 years ago by Goldberg and Rao \cite{GoldbergRao}.
In the special case of solving dense directed unit capacity graphs
our algorithm improves upon the previous fastest running times of
of $O(m\min\{m^{1/2},n^{2/3}\})$ achieved by Even and Tarjan \cite{even1975network}
and Karzanov \cite{k1973} over 35 years ago and of $\otilde(m^{10/7})$
achieved recently by Mądry \cite{madryFlow}. 

We achieve these results through the development and application of
a new general interior point method that we believe is of independent
interest. The number of iterations required by this algorithm is better
than that achieved by analyzing the best self-concordant barrier of
the feasible region. By applying this method to the linear programming
formulations of maximum flow, minimum cost flow, and lossy generalized
minimum cost flow analyzed by Daitch and Spielman \cite{daitch2008faster}
we achieve a running time of $\otilde(m\sqrt{n}\log^{O(1)}(U/\epsilon))$
for these problems as well. Furthermore, our algorithm is parallelizable
and using a recent nearly linear work polylogarithmic depth Laplacian
system solver of Spielman and Peng \cite{peng2013efficient} we achieve
a $\otilde(\sqrt{n}\log^{O(1)}(U/\epsilon))$ depth and $\otilde(m\sqrt{n}\log^{O(1)}(U/\epsilon))$
work algorithm for solving these problems.
\end{abstract}

\section{Introduction}

The maximum flow problem and its dual, the minimum $s$-$t$ cut problem,
are two of the most well studied problems in combinatorial optimization
\cite{schrijver2003combinatorial}. These problems are key algorithmic
primitives used extensively throughout both the theory and practice
of computer science \cite{ahuja1993network}. Numerous problems in
algorithm design efficiently reduce to the maximum flow problem \cite{arora2012multiplicative,ShermanMaxFlow}
and techniques developed in the study of this problem have had far
reaching implications \cite{benczur1996approximating,arora2012multiplicative}. 

Study of the maximum flow problem dates back to 1954 when the problem
was first posed by Harris \cite{schrijver2002history}. After decades
of work the current fastest running time for solving the maximum flow
problem is due to a celebrated result of Goldberg and Rao in 1998
in which they produced a $O(m\min\{m^{1/2},n^{2/3}\}\log(n^{2}/m)\log(U))$
time algorithm for weighted directed graphs with $n$ vertices, $m$
edges and integer capacities of maximum capacity $U$ \cite{GoldbergRao}.%
\footnote{Throughout this paper we restrict our attention to ``weakly'' polynomial
time algorithms, that is algorithms which may depend polylogarithmically
on $U$. The current fastest ``strongly polynomial'' running time
is $O(nm)$ \cite{orlin2013max}. %
} While there have been numerous improvements in the running time for
solving special cases of this problem (see Section~\ref{sub:previous_work}),
the running time for solving the maximum flow problem in full generality
has not been improved since 1998.

In this paper we provide an algorithm that solves the maximum flow
problem with a running time of $\otilde(m\sqrt{n}\log^{O(1)}(U))$,%
\footnote{Here and in the remainder of the paper we use $\otilde(\cdot)$ to
hide $\polylog(m)$ factors.%
} yielding the first improvement to the running time for maximum flow
in 15 years and the running time for solving dense unit capacity directed
graphs in 35 years. Furthermore, our algorithm is easily parallelizable
and using \cite{peng2013efficient}, we obtain a $\otilde(m\sqrt{n}\log^{O(1)}(U))$
work $\otilde(\sqrt{n}\log^{O(1)}(U))$ depth algorithm. Using the
same technique, we also solve the minimum cost flow problem in time
$\otilde(m\sqrt{n}\log^{O(1)}(U))$ time and produce $\epsilon$-approximate
solutions to the lossy generalized minimum cost flow problem in $\otilde(m\sqrt{n}\log^{O(1)}(U/\epsilon))$
time.

We achieve these running times through a novel extension of the work
in Part I \cite{lsInteriorPoint}. In particular, we show how to implement
and analyze an algorithm that is essentially ``dual'' to our approach
in \cite{lsInteriorPoint} and we generalize this algorithm to work
for a broader class of barrier functions. This extension is nontrivial
as it ultimately yields a path following algorithm that achieves a
convergence rate better than that of the best possible self-concordant
barrier for feasible region. To the best of the authors' knowledge
this is the first interior point method to break this long-standing
barrier to the convergence rate of general interior point methods
\cite{Nesterov1994}. Furthermore, by applying our algorithm to the
linear programming formulations of the maximum flow, minimum cost
flow, and lossy generalized minimum cost flow problems analyzed in
\cite{daitch2008faster}, and by using both the error analysis in
\cite{daitch2008faster} and nearly linear time algorithms for solving
Laplacian systems \cite{spielman2004nearly,KoutisMP10,KMP11,Kelner2013,lee2013ACDM,li2012iterative,peng2013efficient},
we achieve the desired running times. 

While our approach is general and the analysis is technical, for the
specific case of the maximum flow problem our linear programming algorithm
has a slightly more straightforward interpretation. The algorithm
simply alternates between re-weighting costs, solving electric flow
problems to send more flow, and approximately computing the effective
resistance of all edges in the graph to keep the effective resistance
of all edges in the graph fairly small and uniform. Hence, by following
the path of (almost) least (effective) resistance, we solve the maximum
flow problem in $\otilde(m\sqrt{n}\log^{O(1)}(U))$.

\subsection{Previous Work}

\label{sub:previous_work}

While the worst case asymptotic running time for solving the maximum
flow problem has remained unchanged over the past 15 years, there
have been significant breakthroughs on specific instances of the problem,
generalizations of the problem, and the technical machinery used to
solve the problem. Here we survey some of the key results that we
leverage to achieve our running times.

Although the running time for solving general directed instances of
maximum flow has remained relatively stagnant until recently \cite{madryFlow},
there have been significant improvements in the running time for computing
maximum flows on undirected graphs over the past few decades. A beautiful
line of work on faster algorithms for approximately solving the maximum
flow problem on undirected graphs began with a result of Benzcur and
Karger in which they showed how to reduce approximately computing
minimum cuts in arbitrary undirected graphs to the same problem on
sparse graphs, i.e. those with only a nearly linear number of vertices
\cite{benczur1996approximating}. In later work, Karger also showed
how reduce computing approximate maximum flow on dense undirected
graphs to computing approximate maximum flows on sparse undirected
graphs \cite{karger1998better}. Pushing this idea further, in a series
of results Karger and Levine showed how to compute the exact maximum
flow in an unweighted undirected graph in time $\otilde(m+nF)$ where
$F$ is the maximum flow value of the graph \cite{karger2002random}. 

In 2004 a breakthrough result of Spielman and Teng \cite{spielman2004nearly}
showed that a particular class of linear systems, Laplacians, can
be solved in nearly linear time and Christiano, Kelner, Mądry, and
Spielman \cite{christiano2011electrical} showed how to use these
fast Laplacian system solvers to approximately solve the maximum flow
problem on undirected graphs in time $\otilde(mn^{1/3}\epsilon^{-11/3})$.
Later Lee, Rao and Srivastava \cite{lee2013new} showed how to solve
the problem in $\otilde(mn^{1/3}\epsilon^{-2/3})$ for undirected
unweighted graphs. This exciting line of work culminated in recent
breakthrough results of Sherman \cite{ShermanMaxFlow} and Kelner,
Lee, Orecchia and Sidford \cite{lee2014linearmaxflow} who showed
how to solve the problem in time almost linear in the number of edges
in the graph, $\otilde(m^{1+o(1)}\epsilon^{-2})$, using congestion-approximators,
oblivious routings, efficient construction techniques developed by
Mądry \cite{DBLP:conf/focs/Madry10}.

In the exact and directed setting, over the past few years significant
progress has been made on solving the maximum flow problem and its
generalizations using interior point methods, a powerful and general
technique for convex optimization \cite{karmarkar1984new,Nesterov1994}.
In 2008, Daitch and Spielman \cite{daitch2008faster} showed that,
by careful application of interior point techniques, fast Laplacian
system solvers \cite{spielman2004nearly}, and a novel method for
solving M-matrices, they could match (up to polylogarithmic factors)
the running time of Goldberg Rao and achieve a running time of $\otilde(m^{3/2}\log^{O(1)}(U))$
not just for maximum flow but also for the minimum cost flow and lossy
generalized minimum cost flow problems. Furthermore, very recently
Mądry \cite{madryFlow} achieved an astounding running time of $\otilde(m^{10/7})$
for solving the maximum flow problem on un-capacitated directed graphs
by a novel application and modification of interior point methods.
This shattered numerous barriers providing the first general improvement
over the running time of $O(m\min\{m^{1/2},n^{2/3}\})$ for solving
unit capacity graphs proven over 35 years ago by Even and Tarjan \cite{even1975network}
and Karzanov \cite{k1973} in 1975.

While our algorithm for solving the maximum flow problem is new, we
make extensive use of these breakthroughs on the maximum flow problem.
We use sampling techniques first discovered in the context of graph
sparsification \cite{spielmanS08sparsRes}, but not to sparsify a
graph but rather to re-weight the graph so that we make progress at
a rate commensurate with the number of vertices and not the number
of edges. We use fast Laplacian system solvers as in \cite{christiano2011electrical,lee2013new},
but we use them to make the cost of interior point iterations cheap
as in \cite{daitch2008faster,madryFlow}. We then use reductions and
error analysis in Daitch and Spielman \cite{daitch2008faster} as
well as their solvers for M-matrices to apply our framework to flow
problems. Furthermore, as in Mądry we use weights to change the central
path (albeit for a slightly different purpose). We believe this further
emphasizes the power of these tools as general purpose techniques
for algorithm design.
\begin{figure}[H]
\begin{centering}
\begin{tabular}{|c|l|c|}
\hline 
Year  & Author  & Running Time\tabularnewline
\hline 
\hline 
1972  & Edmonds and Karp \cite{edmonds1972theoretical} & $\tilde{O}(m^{2}\log(U))$ \tabularnewline
\hline 
1984  & Tardos \cite{tardos1985strongly} & $O(m^{4})$\tabularnewline
\hline 
1984  & Orlin \cite{orlin1984genuinely} & $\tilde{O}(m^{3})$\tabularnewline
\hline 
1986  & Galil and Tardos \cite{galil1988n} & $\tilde{O}\left(mn^{2}\right)$\tabularnewline
\hline 
1987  & Goldberg and Tarjan \cite{goldberg1990finding} & $\tilde{O}(mn\log(U))$ \tabularnewline
\hline 
1988 & Orlin \cite{orlin1993faster} & $\tilde{O}(m^{2})$\tabularnewline
\hline 
2008 & Daitch and Spielman \cite{daitch2008faster} & $\tilde{O}(m^{3/2}\log^{2}(U))$\tabularnewline
\hline 
2013 & This paper & $\tilde{O}(m\sqrt{n}\log^{O(1)}(U))$\tabularnewline
\hline 
\end{tabular}
\par\end{centering}

\protect\caption{Here we summarize the running times of algorithms for the minimum
cost flow problem. $U$ denotes the maximum absolute value of capacities
and costs. For simplicity, we only list exact algorithms which yielded
polynomial improvements.}
\end{figure}

\subsection{Our Approach}

Our approach to the maximum flow problem is motivated by our work
in Part I \cite{lsInteriorPoint}. In Part I we provided a new method
for solving a general linear program written in the \emph{dual} of
standard form
\begin{equation}
\min_{\vy\in\Rn~:~\ma\vy\geq\vc}\vb^{T}\vy\label{eq:lpsimple}
\end{equation}
where $\ma\in\Rmn$, $\vb\in\Rn$, and $\vc\in\Rm$. We showed how
to solve \eqref{eq:lpsimple} in $\otilde(\sqrt{\rank(\ma)}\log(U/\varepsilon))$
iterations while only solving $\otilde(1)$ linear systems in each
iteration.%
\footnote{Throughout this paper we use $U$ to denote the width of a linear
program defined in Theorem \ref{thm:LPSolve_detailed}%
} Whereas previous comparable linear program solvers required $\sqrt{\max\{m,n\}}$
iterations when $\ma$ was full rank, ours only required $\sqrt{\min\{m,n\}}$
in a fairly general regime.

Unfortunately, this result was insufficient to produce faster algorithms
for the maximum flow problem and its generalizations. Given an arbitrary
minimum cost maximum flow instance there is a natural linear program
that one can use to express the problem:
\begin{equation}
\min_{\begin{array}{c}
\vx\in\Rm~:~\ma^{T}\vx=\vb\\
\forall i\in[m]~:~l_{i}\leq x_{i}\leq u_{i}
\end{array}}\vc^{T}\vx\label{eq:lpmaxflowillustration}
\end{equation}
where the variables $x_{i}$ denote the flow on an edge, the $l_{i}$
and $u_{i}$ denote lower and upper bounds on how much flow we can
put on the edge, and $\ma$ is the incidence matrix associated with
the graph \cite{daitch2008faster}. In this formulation, $\rank(\ma)$
is less than the number of vertices in the graph and using fast Laplacian
system solvers \cite{spielman2004nearly,KoutisMP10,KMP11,Kelner2013,lee2013ACDM,li2012iterative,peng2013efficient}
we can solve linear systems involving $\ma$ in time nearly linear
in the number of edges in the graph. Thus, if we could perform similar
error analysis as in Daitch and Spielman \cite{daitch2008faster}
and solve \eqref{eq:lpmaxflowillustration} in time comparable to
that we achieve for solving \eqref{eq:lpsimple} this would immediately
yield a $\otilde(m\sqrt{n}\log^{O(1)}(U))$ algorithm for the maximum
flow problem. Unfortunately, it is not clear how to apply our previous
results in this more general setting and naive attempts to write \eqref{eq:lpmaxflowillustration}
in the form of \eqref{eq:lpsimple} without increasing $\rank(\ma)$
fail. 

Even more troubling, achieving a faster than $\otilde(\sqrt{m}L)$
iterations interior point method for solving general linear programs
in this form would break a long-standing barrier for the convergence
rate of interior point methods. In a seminal result of Nesterov and
Nemirovski \cite{Nesterov1994}, they provided a unifying theory for
interior point methods and showed that given the ability to construct
a \emph{$v$-self concordant barrier }for a convex set, one can minimize
linear functions over that convex set with a convergence rate of $O(\sqrt{v})$.
Furthermore, they showed how to construct such barriers for a variety
of convex sets and thereby achieve fast running times. 

To the best of the authors knowledge, there is no general purpose
interior point method that achieves a convergence rate faster than
the self concordance of the best barrier of the feasible region. Furthermore,
using lower bounds results of Nesterov and Nemirovski, it is not hard
to see that any general barrier for \eqref{eq:lpmaxflowillustration}
must have self-concordance $\Omega(m)$. To be more precise, note
the following result of Nesterov and Nemirovski.
\begin{thm}
[{\cite[Proposition 2.3.6]{Nesterov1994}}]\label{thm:nesterov_result}
Let $\Omega$ be a convex polytope in $\mathbb{R}^{m}$. Suppose there
is a vertices of the polytope belongs exactly to $k$ linearly independent
$(m-1)$-dimensional facets. Then, the self-concordance of any barrier
on $\Omega$ is at least $k$.
\end{thm}
Consequently, even if our maximum flow instance just consisted of
$O(m)$ edges in parallel Theorem~\ref{thm:nesterov_result} implies
that a barrier for the polytope must have self-concordance at least
$\Omega(m)$. Note that this does not rule out a different reduction
of the problem to minimizing a linear function over a convex body
for which there is a $O(n)$ self-concordant barrier. However, it
does reflect the difficulty of using standard analysis of interior
point methods.

\subsection{Our Contributions}

\label{sub:contributions}

In this paper we provide an $\tilde{O}(\sqrt{\rank(\ma)}\log(U/\varepsilon))$
iteration algorithm for solving linear programs of the form \eqref{eq:lpmaxflowillustration}.
This is the first general interior point method we aware of that converges
at a faster rate than the self-concordance of the best barrier of
the feasible region. Each iteration of our algorithm involves solving
of $\otilde(1)$ linear systems of the form $\mbox{\ensuremath{\ma^{T}\md\ma\vx}=\ensuremath{\vd}}$.
By applying this method to the linear program formulation of lossy
generalized minimum cost flow analyzed in Daitch and Spielman \cite{daitch2008faster},
we achieve a running time of $\otilde(m\sqrt{n}\log^{O(1)}(U/\epsilon))$
for solving this problem.

We achieve this running time by a novel extension of the ideas in
\cite{lsInteriorPoint} to work with the primal linear program formulation
\eqref{eq:lpmaxflowillustration} directly. Using an idea from \cite{freund1995barrier},
we create a 1-self concordant barrier for each of the $l_{i}\leq x\leq u_{i}$
constraints and run a primal path following algorithm with the sum
of these barriers. While this would naively yield a $O(\sqrt{m}\log(U/\varepsilon))$
iteration method, we show how to use weights in a similar manner as
in \cite{lsInteriorPoint} to improve the convergence rate to $\tilde{O}(\sqrt{\rank(\ma)}\log(U/\varepsilon))$.

While there are similarities between this analysis and the analysis
in Part I, we cannot use that result directly. Changing from weighted
path following in the dual linear program formulation to this primal
formulation changes the behavior of the algorithm and in essentially
shifts degeneracies in maintaining weights to degeneracies to maintaining
feasibility. This simplifies some parts of the analysis and makes
others make some parts of the analysis simpler and some more complicated. 

On the positive side, the optimization problem we need to solve to
computes the weights becomes better conditioned. Furthermore, inverting
the behavior of the weights obviates the need for $r$-steps that
were key to our analysis in our previous work. 

On the negative side, we have to regularize the weight computation
so that weight changes do not undo newton steps on the feasible point
and we have to do further work to show that Newton steps on the current
feasible point are stable. In the dual formulation it was easy to
assert that small Newton steps on the current point do not change
the point multiplicatively. However, for this primal analysis this
is no longer the case and we need to explicitly bound the size of
the Newton step in both the $\ellInf$ norm and a weighted $\ellTwo$
norm. Hence, we measure the centrality of our points by the size of
the Newton step in a mixed norm of the form $\norm{\cdot}=\norm{\cdot}_{\infty}+\cnorm\norm{\cdot}_{\mWeight}$
to keep track of these two quantities simultaneously.

Measuring of Newton step size both with respect to the mixed norm
helps to explain how our method outperforms the self-concordance of
the best barrier for the space. Self-concordance is based on $\ell_{2}$
analysis and the lower bounds for self-concordance are precisely the
failure of the sphere to approximate a box. While ideally we would
just perform optimization over the $\ellInf$ box directly, $\ellInf$
is ripe with degeneracies that makes this analysis difficult. Nevertheless,
unconstrained minimization over a box is quite simple and by working
with this mixed norm and choosing weights to improve the conditioning,
we are taking advantage of the simplicity of minimizing $\ellInf$
over most of the domain and only paying for the $n$-self-concordance
of a barrier for the smaller subspace induce by the requirement that
$\ma^{T}\vx=\vb$. We hope that this analysis may find further applications.

\subsection{Paper Organization}

The rest of our paper is structured as follows. In Section~\ref{sec:Notation}
and Section~\ref{sec:Preliminaries} we cover preliminaries. In Section~\ref{sec:weighted_path_finding}
we introduce our path finding framework and in Section~\ref{sec:weighted-path}
we present the key lemmas used to analyze progress along paths and
in Section~\ref{sec:weight-function} we introduce the weight function
we use to find paths. In Section~\ref{sec:master_thm} provide our
linear programming algorithm and in Section~\ref{sec:master_thm_stable}
we discuss the requirements of the linear system solvers we use in
the algorithm. In Section~\ref{sec:Applications} we use these results
to achieve our desired running times for the maximum flow problem
and its generalizations.

Some of the analysis in this paper is similar to our previous work
in Part I \cite{lsInteriorPoint} and when the analysis is nearly
the same we often omit details. We encourage the reader to look at
Part I \cite{lsInteriorPoint} for more detailed analysis and longer
expositions of the machinery we use in this paper. Note that throughout
this paper we make no attempt to reduce polylogarithmic factors in
our running times.

\section{Notation\label{sec:Notation}}

Here we introduce various notation that we will use throughout the
paper. This section should be used primarily for reference as we reintroduce
notation as needed later in the paper. (For a summary of linear programming
specific notation we use, see Appendix~\ref{sec:glossary}.)

\medskip{}

\textbf{Variables:} We use the vector symbol, e.g. $\vx$, to denote
a vector and we omit the symbol when we denote the vectors entries,
e.g. $\vx=(x_{1},x_{2},\ldots)$. We use bold, e.g. $\ma$, to denote
a matrix. For integers $z\in\Z$ we use $[z]\subseteq\Z$ to denote
the set of integers from 1 to $z$. We let $\indicVec i$ denote the
vector that has value $1$ in coordinate $i$ and is $0$ elsewhere.

\medskip{}

\textbf{Vector Operations:} We frequently apply scalar operations
to vectors with the interpretation that these operations should be
applied coordinate-wise. For example, for vectors $\vx,\vy\in\Rn$
we let $\vx/\vy\in\Rn$ with $[\vx/\vy]_{i}\defeq(x_{i}/y_{i})$ and
$\log(\vx)\in\Rn$ with $[\log(\vx)]_{i}=\log(x_{i})$ for all $i\in[n]$
. 

\medskip{}

\textbf{Matrix Operations:} We call a symmetric matrix $\ma\in\Rnn$
positive semidefinite (PSD) if $\vx^{T}\ma\vx\geq0$ for all $\vx\in\Rn$
and we call $\ma$ positive definite (PD) if $\vx^{T}\ma\vx>0$ for
all $\vx\in\Rn$. For a positive definite matrix $\ma\in\Rnn$ we
denote let $\|\cdot\|_{\ma}\,:\,\R^{n}\rightarrow\R$ denote the norm
such that for all $\vx\in\Rn$ we have $\norm{\vx}_{\ma}\defeq\sqrt{\vx^{T}\ma\vx}$.
For symmetric matrices $\ma,\mb\in\Rnn$ we write $\ma\specLeq\mb$
to indicate that $\mb-\ma$ is PSD (i.e. $\vx^{T}\ma\vx\leq\vx^{T}\mb\vx$
for all $\vx\in\Rn$) and we write $\ma\specLt\mb$ to indicate that
$\mb-\ma$ is PD (i.e. that $\vx^{T}\ma\vx<\vx^{T}\mb\vx$ for all
$\vx\in\Rn$). We define $\specGt$ and $\specGeq$ analogously. For
$\ma,\mb\in\R^{n\times m}$, we let $\ma\shurProd\mb$ denote the
Schur product, i.e. $[\ma\shurProd\mb]_{ij}\defeq\ma_{ij}\cdot\mb_{ij}$
for all $i\in[n]$ and $j\in[m]$, and we let $\shurSquared{\ma}\defeq\ma\shurProd\ma$.
We use $\nnz(\ma)$ to denote the number of nonzero entries in $\ma$.
For any norm $\|\cdot\|$ and matrix $\mm$, the operator norm of
$\mm$ is defined by $\norm{\mm}=\sup_{\|\vx\|=1}\norm{\mm\vx}$.

\medskip{}

\textbf{Diagonal Matrices:} For $\ma\in\R^{n\times n}$ we let $\diag(\ma)\in\R^{n}$
denote the vector such that $\diag(\ma)_{i}=\ma_{ii}$ for all $i\in[n]$.
For $\vx\in\R^{n}$ we let $\mDiag(\vx)\in\R^{n\times n}$ be the
diagonal matrix such that $\diag(\mDiag(\vx))=\vx$. For $\ma\in\R^{n\times n}$
we let $\mDiag(\ma)$ be the diagonal matrix such that $\diag(\mDiag(\ma))=\diag(\ma)$.
For $\vx\in\Rn$ when the meaning is clear from context we let $\mx\in\Rnn$
denote $\mx\defeq\mDiag(\vx)$.

\medskip{}

\textbf{Multiplicative Approximations:} Frequently in this paper we
need to convey that two vectors $\vx$ and $\vy$ are close multiplicatively.
We often write $\|\mx^{-1}(\vy-\vx)\|_{\infty}\leq\epsilon$ to convey
the equivalent facts that $y_{i}\in[(1-\epsilon)x_{i},(1+\epsilon)x_{i}]$
for all $i$ or $(1-\epsilon)\mx\specLeq\my\specLeq(1+\epsilon)\mx$.
At times we find it more convenient to write $\|\log\vx-\log\vy\|_{\infty}\leq\epsilon$
which is approximately equivalent for small $\epsilon$. In Lemma~\ref{lem:appendix:log_helper},
we bound the quality of this approximation.

\medskip{}

\textbf{Matrices:} We use $\rPos^{m}$ to denote the vectors in $\Rm$
where each coordinate is positive and for a matrix $\ma\in\Rmn$ and
vector $\vx\in\rPos^{m}$ we define the following matrices and vectors 
\begin{itemize}
\item Projection matrix $\mProj_{\ma}(\vx)\in\Rmm$: $\mProj_{\ma}(\vx)\defeq\mx^{1/2}\ma(\ma^{T}\mx\ma)^{-1}\ma^{T}\mx^{1/2}$.
\item Leverage scores $\vLever_{\ma}(\vx)\in\Rm$: $\vLever_{\ma}(\vx)\defeq\diag(\mProj_{\ma}(\vx))$.
\item Leverage matrix $\mLever_{\ma}(\vx)\in\Rmm$: $\mLever_{\ma}(\vx)\defeq\mDiag(\mProj_{\ma}(\vx))$.
\item Projection Laplacian $\mLapProj_{\ma}(\vx)\in\Rmm$: $\mLapProj_{\ma}(\vx)\defeq\mLever_{\ma}(\vx)-\shurSquared{\mProj_{\ma}(\vx)}$. 
\end{itemize}
The definitions of projection matrix and leverage scores are standard
when the rows of $\ma$ are reweighed by the values in vector $\vx$.

\medskip{}

\textbf{Convex Sets:} We call a set $U\subseteq\R^{k}$ \emph{convex}
if for all $\vx,\vy\in\R^{k}$ and all $t\in[0,1]$ it holds that
$t\cdot\vx+(1-t)\cdot\vy\in U$. We call $U$ \emph{symmetric} if
$\vx\in\R^{k}\Leftrightarrow-\vx\in\R^{k}$. For any $\alpha>0$ and
convex set $U\subseteq\R^{k}$ we let $\alpha U\defeq\{\vx\in\R^{k}|\alpha^{-1}\vx\in U\}$.
For any $p\in[1,\infty]$ and $r>0$ we refer to the symmetric convex
set $\{\vx\in\R^{k}|\|\vx\|_{p}\leq r\}$ as \emph{the $\ellP$ ball
of radius $r$}.

\medskip{}

\textbf{Calculus:} For $f:\Rn\rightarrow\R$ differentiable at $x\in\Rn$,
we let $\grad f(\vx)\in\Rn$ denote the gradient of $f$ at $\vx$,
i.e. $[\grad f(\vx)]_{i}=\frac{\partial}{\partial x_{i}}f(\vx)$ for
all $i\in[n]$. For $f\in\mathbb{R}^{n}\rightarrow\R$ twice differentiable
at $x\in\Rn$, we let $\hess f(\vx)$ denote the hessian of $f$ at
$x$, i.e. $[\grad f(\vx)]_{ij}=\frac{\partial^{2}}{\partial x_{i}\partial x_{j}}f(\vx)$
for all $i,j\in[n]$. Often we will consider functions of two vectors,
$g:\R^{n_{1}\times n_{2}}\rightarrow\R$, and wish to compute the
gradient and Hessian of $g$ restricted to one of the two vectors.
For $\vx\in\R^{n_{1}}$ and $\vy\in\R^{n_{2}}$ we let $\grad_{\vx}\vg(\va,\vb)\in\R^{n_{1}}$
denote the gradient of $\vg$ for fixed $\vy$ at point $\{\va,\vb\}\in\R^{n_{1}\times n_{2}}$.
We define $\grad_{\vy}$, $\hess_{\vx\vx}$, and $\hess_{\vy\vy}$
similarly. For $h:\R^{n}\rightarrow\R^{m}$ differentiable at $\vx\in\Rn$
we let $\mj(\vh(\vx))\in\R^{m\times n}$ denote the Jacobian of $\vh$
at $\vx$ where for all $i\in[m]$ and $j\in[n]$ we let $[\mj(\vh(\vx))]_{ij}\defeq\frac{\partial}{\partial x_{j}}h(\vx)_{i}$.
For functions of multiple vectors we use subscripts, e.g. $\mj_{\vx}$,
to denote the Jacobian of the function restricted to the $\vx$ variable.

\section{Preliminaries\label{sec:Preliminaries}}

\subsection{The Problem\label{sub:Preliminaries:The-Problem}}

The central goal of this paper is to efficiently solve the following
linear program

\begin{equation}
\min_{\begin{array}{c}
\vx\in\Rm~:~\ma^{T}\vx=\vb\\
\forall i\in[m]~:~l_{i}\leq x_{i}\leq u_{i}
\end{array}}\vc^{T}\vx\label{eq:lp}
\end{equation}
where $\ma\in\Rmn$, $\vb\in\Rn$, $\vc\in\Rm$, $l_{i}\in\R\cup\{-\infty\}$,
and $u_{i}\in\R\cup\{+\infty\}$.%
\footnote{Typically \eqref{eq:lp} is written as $\ma\vx=\vb$ rather than $\ma^{T}\vx=\vb$.
We chose this formulation to be consistent with the dual formulation
in \cite{lsInteriorPoint} and to be consistent with the standard
use of $n$ to denote the number of vertices and $m$ to denote the
number of edges in a graph in the linear program formulation of flow
problems.%
} We assume that that for all $i\in[m]$ the domain of variable $x_{i}$,
$\dom(x_{i})\defeq\{x\,:\, l_{i}\leq x\leq u_{i}\}$, is non-degenerate.
In particular we assume that $\dom(x_{i})$ is not the empty set,
a singleton, or the entire real line, i.e. $l_{i}<u_{i}$ and either
$l_{i}\neq-\infty$ or $u_{i}\neq+\infty$. Furthermore we make the
standard assumptions that $\ma$ has full column rank, and therefore
$m\geq n$, and we assume that the interior of the polytope, $\dInterior\defeq\{\vx\in\Rm~:~\ma^{T}\vx=\vb,l_{i}<x_{i}<u_{i}\}$,
is non-empty.%
\footnote{For techniques to relax these assumptions see Appendix~E of Part~I
\cite{lsInteriorPoint}.%
}

The linear program \eqref{eq:lp} is a generalization of standard
form, the case where for all $i\in[m]$ we have $l_{i}=0$ and $u_{i}=+\infty$.
While it is well known that all linear programs can be written in
standard form, the transformations to rewrite \eqref{eq:lp} in standard
form may increase the rank of $\ma$ and therefore we solve \eqref{eq:lp}
directly.

\subsection{Coordinate Barrier Functions\label{sub:Preliminaries:coordinate_barriers}}

Rather than working directly with the different domain of the $x_{i}$
we take a slightly more general approach and for the remainder of
the paper assume that for all $i\in[m]$ we have a \emph{barrier function},
$\phi_{i}:\dom(x_{i})\rightarrow\R$, such that 
\[
\lim_{x\rightarrow l_{i}}\phi_{i}(x)=\lim_{x\rightarrow u_{i}}\phi_{i}(x)=+\infty.
\]
More precisely, we assume that each $\phi_{i}$ is a \emph{1-self-concordant
barrier function.}
\begin{defn}[1-Self-Concordant Barrier Function \cite{Nesterov2003}]
\label{def:1-self-concordant-barrier}A thrice differentiable real
valued barrier function $\phi$ on a convex subset of $\mathbb{R}$
is called a \emph{1-self-concordant barrier function} if
\begin{equation}
\left|\phi'''(x)\right|\leq2(\phi''(x))^{3/2}\text{ for all }x\in\dom(\phi)\label{eq:gen:barrier_assumption_concordance}
\end{equation}
and
\begin{equation}
\left|\phi'(x)\right|\leq\sqrt{\phi''(x)}\text{ for all }x\in\dom(\phi).\label{eq:gen:barrier_assumption_size}
\end{equation}

\end{defn}
The first condition \eqref{eq:gen:barrier_assumption_concordance}
bounds how quickly the second order approximation to the function
can change and the second condition \eqref{eq:gen:barrier_assumption_size}
bounds how much force the barrier can exert.

The existence of a self-concordant barrier for the domain is a standard
assumption for interior point methods \cite{Nesterov1994}. However,
for completeness, here we show how for each possible setting of the
$l_{i}$ and $u_{i}$ there is an explicit 1-self-concordant barrier
function we can use:
\begin{itemize}
\item \emph{Case (1):}\textbf{\emph{ }}\emph{$l_{i}$ finite and $u_{i}=\infty$}:
Here we use a\emph{ log barrier }defined as $\phi_{i}(x)\defeq-\log(x-l_{i})$.
For this barrier we have
\[
\phi_{i}'(x)=-\frac{1}{x-l_{i}}\enspace\text{,}\enspace\phi_{i}''(x)=\frac{1}{(x-l_{i})^{2}}\enspace\text{, and}\enspace\phi_{i}'''(x)=-\frac{2}{(x-l_{i})^{3}}
\]
and therefore clearly $|\phi_{i}'''(x)|=2(\phi_{i}''(x))^{3/2}$ ,
$|\phi_{i}'(x)|=\sqrt{\phi''_{i}(x)}$, and $\lim_{x\rightarrow l_{i}}\phi_{i}(x)=\infty.$ 
\item \emph{Case (2): $l_{i}=-\infty$ and $u_{i}$ finite}: Here we use
a\emph{ log barrier }defined as $\phi_{i}(x)\defeq-\log(u_{i}-x)$.
For this barrier we have
\[
\phi_{i}'(x)=\frac{1}{u_{i}-x}\enspace\text{,}\enspace\phi_{i}''(x)=\frac{1}{(u_{i}-x)^{2}}\enspace\text{, and}\enspace\phi_{i}'''(x)=-\frac{2}{(u_{i}-x)^{3}}
\]
and therefore clearly $|\phi_{i}'''(x)|=2(\phi_{i}''(x))^{3/2}$,
$|\phi_{i}'(x)|=\sqrt{\phi''_{i}(x)}$, and $\lim_{x\rightarrow u_{i}}\phi_{i}(x)=\infty.$ 
\item \emph{Case (3): $l_{i}$ finite and $u_{i}$ finite}: Here we use
a\emph{ trigonometric barrier}%
\footnote{The authors are unaware of this barrier being used previously. In
\cite{freund1995barrier} they considered a similar setting of 0,
1, or 2 sided constraints in \eqref{eq:lp} however for the finite
$l_{i}$ and $u_{i}$ case they considered either the the barrier
$-\log(u_{i}-x_{i})-\log(x_{i}-l_{i})$, for which the proof of condition
\eqref{eq:gen:barrier_assumption_size} in Definition~\ref{def:1-self-concordant-barrier}
is more subtle or the barrier $-\log(\min\{u_{i}-x,x-l_{i}\})+\min\{u_{i}-x_{i},x_{i}-l_{i}\}/((u_{i}-l_{i})/2$
which is not thrice differentiable. The ``trigonometric barrier''
we use arises as the (unique) solution of the ODE $\phi'''=2\left(\phi''\right)^{3/2}$
such that the function value goes to infinity up at $u_{i}$ and $l_{i}$. %
}\emph{ }defined as $\phi_{i}(x)\defeq-\log\cos(a_{i}x+b_{i})$ for
$a_{i}=\frac{\pi}{u_{i}-l_{i}}$ and $b_{i}=-\frac{\pi}{2}\frac{u_{i}+l_{i}}{u_{i}-l_{i}}$.
Note for this choice as $x\rightarrow u_{i}$ we have $a_{i}x+b_{i}\rightarrow\frac{\pi}{2}$
and as $x\rightarrow l_{i}$ we have $a_{i}x+b_{i}\rightarrow\frac{-\pi}{2}$
and in both cases $\phi_{i}(x)\rightarrow\infty.$ Furthermor\emph{e,
\[
\phi_{i}'(x)=a_{i}\tan\left(a_{i}x+b_{i}\right)\enspace,\enspace\phi_{i}''(x)=\frac{a_{i}^{2}}{\cos^{2}(a_{i}x+b_{i})}\enspace\text{, and}\enspace\phi_{i}'''=\frac{2a_{i}^{3}\sin(a_{i}x+b_{i})}{\cos^{3}(a_{i}x+b_{i})}.
\]
}Therefore, we have
\[
\left|\phi_{i}'''(x)\right|=\left|\frac{2a_{i}^{3}\sin(a_{i}x+b_{i})}{\cos^{3}(a_{i}x+b_{i})}\right|\leq\frac{2a_{i}^{3}}{|\cos(a_{i}x+b_{i})|^{3}}=2(\phi''(x))^{3/2}
\]
and $|\phi_{i}'(x)|\leq\frac{a_{i}}{\left|\cos\left(a_{i}x+b_{i}\right)\right|}=\sqrt{\phi_{i}''(x)}$.
\end{itemize}
For the remainder of this paper we will simply assume that we have
a 1 self-concordant barrier $\phi_{i}$ for each of the $\dom(\phi_{i})$
and not use any more structure about the barriers.

While there is much theory regarding properties of self-concordant
barrier functions we will primarily use two common properties about
self-concordant barriers functions. The first property, Lemma~\ref{lem:gen:phi_properties_sim},
shows that the Hessian of the barrier cannot change to quickly, and
the second property, Lemma~\ref{lem:gen:phi_properties_force} we
use to reason about how the force exerted by the barrier changes over
the domain.
\begin{lem}
[{\cite[Theorem 4.1.6]{Nesterov2003}}]\label{lem:gen:phi_properties_sim}
Suppose $\phi$ is a 1-self-concordant barrier function. For all $s\in\dom(\phi)$
if $r\defeq\sqrt{\phi''(s)}\left|s-t\right|<1$ then $t\in\dom(\phi)$
and
\[
(1-r)\sqrt{\phi''(s)}\leq\sqrt{\phi''(t)}\leq\frac{\sqrt{\phi''(s)}}{1-r}.
\]

\begin{lem}
[{\cite[Theorem 4.2.4]{Nesterov2003}}]\label{lem:gen:phi_properties_force}
Suppose $\phi$ is a 1-self-concordant barrier function. For all $x,y\in\dom(\phi)$
, we have
\[
\phi'(x)\cdot(y-x)\leq1.
\]

\end{lem}
\end{lem}

\section{Weighted Path Finding\label{sec:weighted_path_finding}}

In this paper we show how \eqref{eq:lp} can be solved using \emph{weighted
path finding.}%
\footnote{See Part I \cite{lsInteriorPoint} for more motivation regarding weighted
paths.%
} Our algorithm is essentially ``dual'' to the algorithm in Part
I \cite{lsInteriorPoint} and our analysis holds in a more general
setting. In this section we formally introduce this weighted central
path (Section~\ref{sub:weighted_central_path}) and define key properties
of the path (Section~\ref{sub:centrality}) and the weights (Section~\ref{sub:weight_function})
that we will use to produce an efficient path finding scheme.

\subsection{The Weighted Central Path\label{sub:weighted_central_path}}

Our linear programming algorithm maintains a feasible point $\vx\in\dInterior$,
weights $\vWeight\in\rPos^{m}$, and minimizes the following \emph{penalized
objective function}
\begin{equation}
\min_{\ma^{T}\vx=\vb}f_{t}(\vx,\vWeight)\enspace\text{where }\enspace f_{t}(\vx,\vWeight)\defeq t\cdot\vc^{T}\vx+\sum_{i\in[m]}w_{i}\phi_{i}(\vx_{i})\label{eq:penalized_objective}
\end{equation}
for increasing $t$ and small $\vWeight$. For every fixed set of
\emph{weights}, $\vWeight\in\rPos^{m}$ the set of points $\vx_{\vWeight}(t)=\argmin_{\vx\in\dInterior}f_{t}(\vx,\vWeight)$
for $t\in[0,\infty)$ form a path through the interior of the polytope
that we call the \emph{weighted central path}. We call $\vx_{\vWeight}(0)$
a \emph{weighted center} of the polytope and note that $\lim_{t\rightarrow\infty}\vx_{\vWeight}(0)$
is a solution to the linear program. 

While all weighted central paths converge to a solution of the linear
program, different paths may have different algebraic properties either
increasing or decreasing the difficult of a path following scheme
(see Part 1 \cite{lsInteriorPoint}). Consequently, our algorithm
alternates between advancing down a central path (i.e. increasing
$t$), moving closer to the weighted central path (i.e. updating $\vx$),
and picking a better path (i.e. updating the weights $\vWeight$). 

Ultimately, our weighted path finding algorithm follows a simple iterative
scheme. We assume we have a feasible point $\{\vx,\vWeight\}\in\dFull$
and a weight function $\vg(\vx):\dInterior\rightarrow\rPos^{m}$,
such that for any point $\vx\in\rPos^{m}$ the function $\vg(\vx)$
returns a good set of weights that suggest a possibly better weighted
path. Our algorithm then repeats the following.
\begin{enumerate}
\item If $\vx$ close to $\argmin_{\vy\in\Omega}f_{t}\left(\vy,\vWeight\right)$,
then increase $t$.
\item Otherwise, use projected Newton step to update $\vx$ and move $\vWeight$
closer to $\vg(\vx)$.
\item Repeat.
\end{enumerate}
In the remainder of this section we present how we measure both the
quality of a current feasible point $\{\vx,\vWeight\}\in\dFull$ and
the quality of the weight function. In Section~\ref{sub:centrality}
we derive and present both how we measure how close $\{\vx,\vWeight\}$
is to the weighted central path and the step we take to improve this
\emph{centrality}. Then in Section~\ref{sub:weight_function} we
present how we measure the quality of a weight function, i.e. how
good the weighted paths it finds are.

\subsection{Measuring Centrality.\label{sub:centrality}}

Here we explain how we measure the distance from $\vx$ to the minimum
of $f_{t}\left(\vx,\vWeight\right)$ for fixed $\vWeight$. This distance
is a measure of how close $\vx$ is to the weighted central path and
we refer to it as the \emph{centrality }of $\vx$, denote $\delta_{t}(\vx,\vWeight)$.
Whereas in Part I \cite{lsInteriorPoint} we simply measured centrality
by the size of the Newton step in the Hessian norm, here we use a
slightly more complicated definition in order to reason about multiplicative
changes in the Hessian (See Section~\ref{sub:contributions}).

To motivate our centrality measure we first compute a projected Newton
step for $\vx$. For all $\vx\in\dInterior$, we define $\vphi(\vx)\in\Rm$
by $\vphi(\vx)_{i}=\phi_{i}(\vx_{i})$ for $i\in[m]$. We define $\vphi'(\vx)$,
$\vphi''(\vx)$, and $\vphi'''(\vx)$ similarly and let $\mPhi,\mPhi',\mPhi'',\mPhi'''$
denote the diagonal matrices corresponding to these matrices. Using
this, we have%
\footnote{Recall that $\vWeight\vphi'(\vx)$ denotes the entry-wise multiplication
of the vectors $\vWeight$ and $\vphi'(\vx)$.%
}
\[
\grad_{x}f_{t}(\vx,\vWeight)=t\cdot\vc+\vWeight\vphi'(\vx)\enspace\text{ and }\enspace\grad_{xx}f_{t}(\vx,\vWeight)=\mWeight\mPhi''(\vx)\,.
\]
Therefore, a Newton step for $\vx$ is given by
\begin{align}
\vh_{t}(\vx,\vWeight) & =-\left(\mWeight\mPhi''(\vx)\right)^{-1/2}\mProj_{\ma^{T}(\mWeight\mPhi''(\vx))^{-1/2}}\left(\mWeight\mPhi''(\vx)\right)^{-1/2}\grad_{x}f_{t}(\vx,\vWeight)\nonumber \\
 & =-\mPhi''(\vx)^{-1/2}\Pxw\mWeight^{-1}\mPhi''(\vx)^{-1/2}\grad_{x}f_{t}(\vx,\vWeight)\label{eq:newton_step}
\end{align}
where $\mProj_{\ma^{T}(\mWeight\mPhi''(\vx))^{-1/2}}$ is the orthogonal
projection onto the kernel of $\ma^{T}(\mWeight\mPhi''(\vx))^{-1/2}$
and $\Pxw$ is the orthogonal projection onto the kernel of $\ma^{T}(\mPhi''(\vx))^{-1/2}$
with respect to the norm $\|\cdot\|_{\mWeight}$, i.e.
\begin{equation}
\mProj_{\vx,\vWeight}\defeq\iMatrix-\mWeight^{-1}\ma_{x}\left(\ma_{x}^{T}\mWeight^{-1}\ma_{x}\right)^{-1}\ma_{x}^{T}\enspace\text{ for }\enspace\ma_{x}\defeq\mPhi''(\vx)^{-1/2}\ma\,.\label{eq:def_Pxw}
\end{equation}
 As with standard convergence analysis of Newton's method, we wish
to keep the Newton step size in the Hessian norm, i.e. $\norm{\vh_{t}(\vx,\vWeight)}_{\mWeight\mPhi''(\vx)}=\norm{\sqrt{\vphi''(\vx)}\vh_{t}(\vx,\vWeight)}_{\mWeight}$,
small and the multiplicative change in the Hessian, $\norm{\sqrt{\vphi''(\vx)}\vh_{t}(\vx,\vWeight)}_{\infty}$,
small (See Lemma \ref{lem:gen:phi_properties_force}). While in the
unweighted case we can bound the multiplicative change by the change
in the hessian norm (since $\|\cdot\|_{\infty}\leq\|\cdot\|_{2}$),
here we would like to use small weights and this comparison would
be insufficient.

To track both these quantities simultaneously, we define the \emph{mixed
norm} for all $\vy\in\Rm$ by
\begin{equation}
\mixedNorm{\vy}{\vWeight}\defeq\norm{\vy}_{\infty}+\cnorm\norm{\vy}_{\mWeight}\label{eq:mixed_norm}
\end{equation}
for some $\cnorm>0$ that we define later. Note that $\mixedNorm{\cdot}{\vWeight}$
is indeed a norm for $\vWeight\in\rPos^{m}$ as in this case both
$\normInf{\cdot}$and $\norm{\cdot}_{\mWeight}$ are norms. However,
rather than measuring centrality by the quantity $\norm{\sqrt{\vphi''(\vx)}\vh_{t}(\vx,\vWeight)}_{\vWeight+\infty}=\mixedNorm{\mProj_{\vx,\vWeight}\left(\frac{\grad_{x}f_{t}(\vx,\vWeight)}{\vWeight\sqrt{\vec{\phi}''}}\right)}{\vWeight}$,
we instead find it more convenient to use the following idealized
form
\[
\delta_{t}(\vx,\vWeight)\defeq\min_{\veta\in\Rn}\normFull{\frac{\grad_{x}f_{t}(\vx,\vWeight)-\ma\veta}{\vWeight\sqrt{\vphi''(\vx)}}}_{\vWeight+\infty}.
\]
We justify this definition by showing these two quantities differ
by at most a multiplicative factor of $\mixedNorm{\mProj_{\vx,\vWeight}}{\vWeight}$
as follows
\begin{equation}
\delta_{t}(\vx,\vWeight)\leq\mixedNorm{\sqrt{\vphi''(\vx)}\vh_{t}(\vx,\vWeight)}{\vWeight}\leq\mixedNorm{\mProj_{\vx,\vWeight}}{\vWeight}\cdot\delta_{t}(\vx,\vWeight).\label{eq:centrality_equivalence}
\end{equation}
This a direct consequence of the more general Lemma \ref{lem:max_flow:projection_lemma}
that we prove in the appendix. 

We summarize this section with the following definition.
\begin{defn}[Centrality Measure]
\label{Def:centrality_measure} For $\{\vx,\vWeight\}\in\dFull$
and $t\geq0$, we let $\vh_{t}(\vx,\vWeight)$ denote the \emph{projected
newton step} for $\vx$ on the penalized objective $f_{t}$ given
by
\[
\vh_{t}(\vx,\vWeight)\defeq-\frac{1}{\sqrt{\vec{\phi}''(\vx)}}\mProj_{\vx,\vWeight}\left(\frac{\grad_{x}f_{t}(\vx,\vWeight)}{\vWeight\sqrt{\vec{\phi}''(\vx)}}\right)
\]
where $\mProj_{\vx,\vWeight}$ is the orthogonal projection onto the
kernel of $\ma^{T}(\mPhi'')^{-1/2}$ with respect to the norm $\norm{\cdot}_{\mWeight}$
(see \ref{eq:mixed_norm}). We measure the \emph{centrality} of $\{\vx,\vWeight\}$
by
\begin{equation}
\delta_{t}(\vx,\vWeight)\defeq\min_{\veta\in\Rn}\normFull{\frac{\grad_{x}f_{t}(\vx,\vWeight)-\ma\veta}{\vWeight\sqrt{\vphi''(\vx)}}}_{\vWeight+\infty}\label{eq:centrality_definition}
\end{equation}
where for all $\vy\in\Rm$ we let $\mixedNorm{\vy}{\vWeight}\defeq\norm{\vy}_{\infty}+\cnorm\norm{\vy}_{\mWeight}$
for some $\cnorm>0$ we define later.
\end{defn}

\subsection{The Weight Function\label{sub:weight_function}}

With the Newton step and centrality conditions defined, the specification
of our algorithm becomes more clear. Our algorithm is as follows
\begin{enumerate}
\item If $\delta_{t}(\vx,\vWeight)$ is small, then increase $t$.
\item Set $\next{\vx}\leftarrow\vx+\vh_{t}(\vx,\vWeight)$ and move $\next{\vWeight}$
towards $\vg(\next{\vx})$.
\item Repeat.
\end{enumerate}
To prove this algorithm converges, we need to show what happens to
$\delta_{t}\left(\vx,\vWeight\right)$ when we change $t$, $\vx$,
$\vWeight$. At the heart of this paper is understanding what conditions
we need to impose on the weight function $\vg(\vx):\dInterior\rightarrow\rPos^{m}$
so that we can bound this change in $\delta_{t}(\vx,\vWeight)$ and
hence achieve fast converge rates. In Lemma~\ref{lem:gen:t_step}
we show that the effect of changing $t$ on $\delta_{t}$ is bounded
by $\cnorm$ and $\norm{\vg(\vx)}_{1}$, in Lemma~\ref{lem:gen:x_progress}
we show that the effect that a Newton Step on $\vx$ has on $\delta_{t}$
is bounded by $\mixedNorm{\mProj_{\vx,\vg(\vx)}}{\vg(\vx)}$, and
in Lemma~\ref{lem:gen:w_step} and \ref{lem:gen:w_change} we show
the change of $\vWeight$ as $\vg(\vx)$ changes is bounded by $\mixedNorm{\fmWeight(\vx)^{-1}\fmWeight'(\vx)(\mPhi''(\vx))^{-1/2}}{\vg(\vx)}$. 

Hence for the remainder of the paper we assume we have a weight function
$\vg(\vx):\dInterior\rightarrow\rPos^{m}$ and make the following
assumptions regarding our weight function. In Section~\ref{sec:weight-function}
we prove that such weight function exists. 
\begin{defn}
[Weight Function]\label{def:gen:weight_function} A \emph{weight
function} is a differentiable function from $\fvWeight:\dInterior\rightarrow\rPos^{m}$
such that for constants $\cWeightSize(\vg)$, $\cWeightStab(\vg)$,
and $\cWeightCons(\vg)$, we have the following for all $\vx\in\dInterior$:
\begin{itemize}
\item \emph{Size }: The size \emph{$\cWeightSize(\fvWeight)=\normOne{\fvWeight(\vx)}$}.
\item \emph{Slack Sensitivity}: The slack sensitivity $\cWeightStab(\fvWeight)$\emph{
}satisfies \emph{$1\leq\cWeightStab(\fvWeight)\leq\frac{5}{4}$} and
$\mixedNorm{\mProj_{\vx,\vWeight}}{\vWeight}\leq\cWeightStab(\fvWeight)$
for any $\vWeight$ such that $\frac{4}{5}\vg\left(\vx\right)\leq\vWeight\leq\frac{5}{4}\vg\left(\vx\right)$.
\item \emph{Step Consistency }:\emph{ }The step consistency\emph{ $\cWeightCons(\fvWeight)$}
satisfies $\cWeightCons(\fvWeight)\cdot\cWeightStab(\fvWeight)<1$
and 
\[
\mixedNorm{\fmWeight(\vx)^{-1}\fmWeight'(\vx)(\mPhi''(\vx))^{-1/2}}{\vg(\vx)}\leq\cWeightCons\leq1.
\]
 
\item \emph{Uniformity }: The weight function satisfies $\normInf{\fvWeight(\vx)}\leq2$.\end{itemize}
\end{defn}

\section{Progressing Along Weighted Paths\label{sec:weighted-path}}

In this section, we provide the main lemmas we need for an $\tilde{O}(\sqrt{\rank(\ma)}\log(U/\varepsilon))$
iterations weighted path following algorithm for (\ref{eq:lp}) assuming
a weight function satisfying Definition~\ref{sub:weight_function}.
In Section \ref{sub:weighted-path:t-delta}, \ref{sub:weighted-path:x-delta},
and \ref{sub:weighted-path:w-delta} we show how centrality, $\delta_{t}(\vx,\vWeight),$
is affected by changing $t$, $\vx\in\dInterior,$ and $\vWeight\in\rPos^{m}$
respectively. In Section \ref{sub:weighted-path:centering} we then
show how to use these Lemmas to improve centrality using approximate
computations of the weight function, $\vg:\dInterior\rightarrow\rPos^{m}$.

\subsection{Changing $t$ \label{sub:weighted-path:t-delta}}

Here we bound how much centrality increases as we increase $t$. We
show that this rate of increase is governed by $\cnorm$ and $\norm{\vWeight}_{1}.$
\begin{lem}
\label{lem:gen:t_step}For all $\{\vx,\vWeight\}\in\dFull$, $t>0$
and $\alpha\geq0$, we have
\begin{eqnarray*}
\delta_{(1+\alpha)t}(\vx,\vWeight) & \leq & (1+\alpha)\delta_{t}(\vx,\vWeight)+\alpha\left(1+\cnorm\sqrt{\norm{\vWeight}_{1}}\right).
\end{eqnarray*}
\end{lem}
\begin{proof}
Let $\veta_{t}\in\Rn$ be such that
\[
\delta_{t}(\vx,\vWeight)=\mixedNorm{\frac{\grad_{x}f_{t}(\vx,\vWeight)+\ma\veta_{t}}{\vWeight\sqrt{\vphi''(\vx)}}}{\vWeight}=\mixedNorm{\frac{t\cdot\vc+\vWeight\vphi'(\vx)+\ma\veta_{t}}{\vWeight\sqrt{\vphi''(\vx)}}}{\vWeight}.
\]
Applying this to the definition of $\delta_{(1+\alpha)t}$ and using
that $\ensuremath{\mixedNorm{\cdot}{\vWeight}}$ is a norm then yields
\begin{align*}
\delta_{(1+\alpha)t}(\vx,\vWeight) & =\min_{\veta\in\Rn}\mixedNorm{\frac{(1+\alpha)t\cdot\vc+\vWeight\vphi'(\vx)+\ma\veta}{\vWeight\sqrt{\vphi''(\vx)}}}{\vWeight}\\
 & \leq\mixedNorm{\frac{(1+\alpha)t\cdot\vc+\vWeight\vphi'(\vx)+\ma(1+\alpha)\veta_{t}}{\vWeight\sqrt{\vphi''(\vx)}}}{\vWeight}\\
 & \leq(1+\alpha)\mixedNorm{\frac{t\cdot\vc+\vWeight\vphi'(\vx)+\ma\veta_{t}}{\vWeight\sqrt{\vphi''(\vx)}}}{\vWeight}+\alpha\mixedNorm{\frac{\vWeight\vphi'(\vx)}{\vWeight\sqrt{\vphi''(\vx)}}}{\vWeight}\\
 & =(1+\alpha)\delta_{t}(\vx,\vWeight)+\alpha\left(\normFullInf{\frac{\vphi'(\vx)}{\sqrt{\vphi''(\vx)}}}+\cnorm\normFull{\frac{\vphi'(\vx)}{\sqrt{\vphi''(\vx)}}}_{\mWeight}\right)
\end{align*}
Using that $|\phi_{i}'(\vx)|\leq\sqrt{\phi_{i}''(\vx)}$ for all $i\in[m]$
and $\vx\in\Rm$ by Definition \ref{def:1-self-concordant-barrier}
yields the result.
\end{proof}

\subsection{Changing $\protect\vx$ \label{sub:weighted-path:x-delta}}

Here we analyze the effect of a Newton step on $\vx$ on centrality.
We show for sufficiently central $\{\vx,\vWeight\}\in\dFull$ and
$\vWeight$ sufficiently close to $\vg(\vx)$ Newton steps converge
quadratically.
\begin{lem}
\label{lem:gen:x_progress} Let $\{\vx_{0},\vWeight\}\in\dFull$ such
that $\delta_{t}(\vx_{0},\vWeight)\leq\frac{1}{10}$ and $\frac{4}{5}\vg\left(\vx\right)\leq\vWeight\leq\frac{5}{4}\vg\left(\vx\right)$
and consider a Newton step $\vx_{1}=\vx_{0}+\vh_{t}(\vx,\vWeight)$.
Then, $\delta_{t}(\vx_{1},\vWeight)\leq4\left(\delta_{t}(\vx_{0},\vWeight)\right)^{2}.$\end{lem}
\begin{proof}
Let $\vphi_{0}\defeq\vphi(\vx_{0})$ and let $\vphi_{1}\defeq\vphi(\vx_{1})$.
By the definition of $\vh_{t}(\vx_{0},\vWeight)$ and the formula
of $\mProj_{\vx_{0},\vWeight}$ we know that there is some $\veta_{0}\in\Rn$
such that
\[
-\sqrt{\vec{\phi}_{0}''}\vh_{t}(\vx_{0},\vWeight)=\frac{t\cdot\vc+\vWeight\vphi_{0}'-\ma\veta_{0}}{\vWeight\sqrt{\vphi_{0}''}}.
\]
Therefore, $\ma\veta_{0}=\vc+\vWeight\phi_{0}'+\vWeight\phi_{0}''h_{t}(\vx_{0},\vWeight)$.
Recalling the definition of $\delta_{t}$ this implies that
\begin{eqnarray*}
\delta_{t}(\vx_{1},\vWeight) & = & \min_{\veta\in\Rn}\mixedNorm{\frac{t\cdot\vc+\vWeight\vec{\phi}_{1}'-\ma\veta}{\vWeight\sqrt{\vec{\phi}_{1}''}}}{\vWeight}\leq\mixedNorm{\frac{t\cdot\vc+\vWeight\vec{\phi}_{1}'-\ma\veta_{0}}{\vWeight\sqrt{\vec{\phi}_{1}''}}}{\vWeight}\\
 & \leq & \mixedNorm{\frac{\vWeight(\vec{\phi}_{1}'-\vphi_{0}')-\vWeight\vphi_{0}''\vh_{t}(\vx_{0},\vWeight)}{\vWeight\sqrt{\vec{\phi}_{1}''}}}{\vWeight}=\mixedNorm{\frac{(\vec{\phi}_{1}'-\vphi_{0}')-\vphi_{0}''\vh_{t}(\vx_{0},\vWeight)}{\sqrt{\vec{\phi}_{1}''}}}{\vWeight}
\end{eqnarray*}
By the mean value theorem, we have $\vec{\phi}_{1}'-\vphi_{0}'=\vec{\phi}''(\vec{\theta})\vh_{t}(\vx_{0},\vWeight)$
for some $\vec{\theta}$ between $\vx_{0}$ and $\vx_{1}$ coordinate-wise.
Hence,
\begin{eqnarray*}
\delta_{t}(\vx_{1},\vWeight) & \leq & \mixedNorm{\frac{\vec{\phi}''(\vec{\theta})\vh_{t}(\vx_{0},\vWeight)-\vphi_{0}''\vh_{t}(\vx_{0},\vWeight)}{\sqrt{\vec{\phi}_{1}''}}}{\vWeight}=\mixedNorm{\frac{\left(\vec{\phi}''(\vec{\theta})-\vec{\phi}_{0}''\right)}{\sqrt{\vec{\phi}_{1}''}\sqrt{\vphi_{0}''}}(\sqrt{\phi_{0}''}\vh_{t}(\vx_{0},\vWeight))}{\vWeight}\\
 & \leq & \normFull{\frac{\vec{\phi}''(\vec{\theta})-\vec{\phi}_{0}''}{\sqrt{\vec{\phi}_{1}''}\sqrt{\vphi_{0}''}}}_{\infty}\cdot\mixedNorm{\sqrt{\phi_{0}''}\vh_{t}(\vx_{0},\vWeight)}{\vWeight}.
\end{eqnarray*}
To bound the first term, we use Lemma \ref{lem:gen:phi_properties_sim}
as follows
\begin{eqnarray*}
\normFull{\frac{\left(\vec{\phi}''(\vec{\theta})-\vec{\phi}_{0}''\right)}{\sqrt{\vec{\phi}_{1}''}\sqrt{\vec{\phi}_{0}''}}}_{\infty} & \leq & \normFull{\frac{\vphi''(\vec{\theta})}{\vphi_{0}''}-\onesVec}_{\infty}\cdot\normFull{\frac{\sqrt{\vec{\phi}_{0}''}}{\sqrt{\vec{\phi}_{1}''}}}_{\infty}\\
 & \leq & \left|\left(1-\normFull{\sqrt{\vec{\phi}_{0}''}\vh_{t}(\vx_{0},\vWeight)}_{\infty}\right)^{-2}-1\right|\cdot\left(1-\normFull{\sqrt{\vec{\phi}_{0}''}\vh_{t}(\vx_{0},\vWeight)}_{\infty}\right)^{-1}.
\end{eqnarray*}
Using (\ref{eq:centrality_equivalence}), i.e. Lemma \ref{lem:max_flow:projection_lemma},
the bound $c_{\gamma}\leq2$, and the assumption on $\delta_{t}(\vx_{0},\vWeight)$,
we have
\[
\normFull{\sqrt{\vec{\phi}_{0}''}\vh_{t}(\vx_{0},\vWeight)}_{\infty}\leq\mixedNorm{\sqrt{\vec{\phi}_{0}''}\vh_{t}(\vx_{0},\vWeight)}{\vWeight}\leq c_{\gamma}\cdot\delta_{t}(\vx_{0},\vWeight)\leq\frac{1}{5}.
\]
Using $\left((1-t)^{-2}-1\right)\cdot(1-t)^{-1}\leq4t$ for $t\leq1/5$,
we have
\[
\normFull{\frac{\left(\vec{\phi}''(\vec{\theta})-\vec{\phi}_{0}''\right)}{\sqrt{\vec{\phi}_{1}''}\sqrt{\vec{\phi}_{0}''}}}_{\infty}\leq4\normFull{\sqrt{\vec{\phi}_{0}''}h_{t}(\vx_{0},\vWeight)}_{\infty}.
\]
Combining the above formulas yields that $\delta_{t}(\vx_{1},\vWeight)\leq4\left(\delta_{t}(\vx_{0},\vWeight)\right)^{2}$
as desired.
\end{proof}

\subsection{Changing $\protect\vWeight$ \label{sub:weighted-path:w-delta}}

In the previous subsection we used the assumption that the weights,
$\vWeight$, were multiplicatively close to the output of the weight
function, $\vg(\vx)$, for the current point $\vx\in\dInterior.$
In order to maintain this invariant when we change $\vx$ we will
need to change $\vWeight$ to move it closer to $\vg(\vx).$ Here
we bound how much $\vg(\vx)$ can move as we move $\vx$ (Lemma \ref{lem:gen:w_step})
and we bound how much changing $\vWeight$ can hurt centrality (Lemma
\ref{lem:gen:w_change}). Together these lemmas will allow us to show
that we can keep $\vWeight$ close to $\vg(\vx)$ while still improving
centrality (Section~\ref{sub:weighted-path:centering}).
\begin{lem}
\label{lem:gen:w_step} For all $t\in[0,1]$, let $\vx_{t}\defeq\vx_{0}+t\vDelta_{x}$
for $\vDelta_{x}\in\Rm$, $\vx_{t}\in\dInterior$, $\vg_{t}=\vg(\vx_{t})$
and $\epsilon=\mixedNorm{\sqrt{\vec{\phi}_{0}''}\vDelta_{x}}{\vg_{0}}\leq0.1$.
Then\textup{
\[
\mixedNorm{\log\left(\vg_{1}\right)-\log\left(\vg_{0}\right)}{\vg_{0}}\leq c_{\delta}\epsilon(1+4\epsilon)\leq0.2
\]
and for all $s,t\in[0,1]$ and for all $\vy\in\Rm$ we hav}e
\begin{equation}
\mixedNorm{\vy}{\vg_{s}}\leq(1+2\epsilon)\mixedNorm{\vy}{\vg_{t}}.\label{eq:change_of_g_norm_est-1}
\end{equation}
\end{lem}
\begin{proof}
Let $\vq:[0,1]\rightarrow\R^{m}$ be given by $\vq(t)\defeq\log\left(\vg_{t}\right)$
for all $t\in[0,1]$. Then, we have
\[
\vq'(t)=\mg_{t}^{-1}\mg_{t}'\vDelta_{x}.
\]
Let $Q(t)\defeq\mixedNorm{\vec{q}(t)-\vec{q}(0)}{\vg_{0}}.$ Using
Jensen's inequality we have that for all $u\in[0,1]$,
\begin{eqnarray*}
Q(u) & \leq & \overline{Q}(u)\defeq\int_{0}^{u}\mixedNorm{\mg_{t}^{-1}\mg_{t}'\left(\vec{\phi}_{t}''\right)^{-1/2}}{\vg_{0}}\mixedNorm{\sqrt{\vec{\phi}_{t}''}\vDelta_{x}}{\vg_{0}}dt.
\end{eqnarray*}
Using Lemma \ref{lem:gen:phi_properties_sim} and $\epsilon\leq\frac{1}{10}$,
we have for all $t\in[0,1]$,
\begin{eqnarray*}
\mixedNorm{\sqrt{\vec{\phi}_{t}''}\vDelta_{x}}{\vg_{0}} & \leq & \normFull{\sqrt{\vec{\phi}_{t}''}/\sqrt{\vec{\phi}_{0}''}}_{\infty}\mixedNorm{\sqrt{\vec{\phi}_{0}''}\vDelta_{x}}{\vg_{0}}\\
 & \leq & \left(1-\normFull{\sqrt{\vec{\phi}_{0}''}\vDelta_{x}}_{\infty}\right)^{-1}\mixedNorm{\sqrt{\vec{\phi}_{0}''}\vDelta_{x}}{\vg_{0}}\\
 & \leq & \frac{\epsilon}{1-\epsilon}.
\end{eqnarray*}
Thus, we have
\begin{eqnarray}
\overline{Q}(u) & \leq & \frac{\epsilon}{1-\epsilon}\int_{0}^{u}\mixedNorm{\mg_{t}^{-1}\mg_{t}'\left(\vec{\phi}_{t}''\right)^{-1/2}}{\vg_{0}}dt.\label{eq:change_of_g_norm_est_int-1}
\end{eqnarray}
Note that $\overline{Q}$ is monotonically increasing. Let $\theta=\sup_{u\in[0,1]}\left\{ \overline{Q}(u)\leq c_{\delta}\epsilon(1+4\epsilon)\right\} $.
Since $\overline{Q}(\theta)\leq\frac{1}{2}$, we know that for all
$s,t\in[0,\theta]$, we have
\[
\normFull{\frac{\vg(\vx_{s})-\vg(\vx_{t})}{\vg(\vx_{t})}}_{\infty}\leq\norm{\vec{q}(s)-\vec{q}(t)}_{\infty}+\norm{\vec{q}(s)-\vec{q}(t)}_{\infty}^{2}
\]
and therefore
\[
\normFull{\vg_{s}/\vg_{t}}_{\infty}\leq\left(1+\norm{\vec{q}(s)-\vec{q}(t)}_{\infty}+\norm{\vec{q}(s)-\vec{q}(t)}_{\infty}^{2}\right)^{2}\leq(1+c_{\delta}\epsilon(1+4\epsilon))^{2}
\]
Consequently,
\begin{eqnarray*}
\mixedNorm{\vy}{\vg_{s}} & \leq & \left(1+c_{\delta}\epsilon(1+4\epsilon)\right)\mixedNorm{\vy}{\vg_{t}}\leq\left(1+2\epsilon\right)\mixedNorm{\vy}{\vg_{t}}.
\end{eqnarray*}
Using (\ref{eq:change_of_g_norm_est_int-1}), we have for all $u\in[0,\theta]$,
\begin{eqnarray*}
Q(u)\leq\overline{Q}(u) & \leq & \frac{\epsilon}{1-\epsilon}\int_{0}^{u}\mixedNorm{\mg_{t}^{-1}\mg_{t}'\left(\vec{\phi}_{t}''\right)^{-1/2}}{\vg_{0}}dt\\
 & \leq & \frac{\epsilon}{1-\epsilon}\int_{0}^{u}\left(1+2\epsilon\right)\mixedNorm{\mg_{t}^{-1}\mg_{t}'\left(\vec{\phi}_{t}''\right)^{-1/2}}{\vg_{t}}dt\\
 & \leq & \frac{\epsilon}{1-\epsilon}\left(1+2\epsilon\right)c_{\delta}\theta\\
 & < & c_{\delta}\epsilon(1+4\epsilon).
\end{eqnarray*}
Consequently, we have that $\theta=1$ and we have the desired result
with $Q(1)\leq c_{\delta}\epsilon(1+4\epsilon)<\frac{1}{5}$.\end{proof}
\begin{lem}
\label{lem:gen:w_change}Let $\vv,\vWeight\in\rPos^{m}$ such that
$\mbox{\ensuremath{\epsilon=\mixedNorm{\log(\vWeight)-\log(\vv)}{\vWeight}}}\le0.1$.
Then for $\vx\in\dInterior$ we have
\[
\delta_{t}(\vx,\vv)\leq(1+4\epsilon)(\delta_{t}(\vx,\vWeight)+\epsilon).
\]
\end{lem}
\begin{proof}
Let $\veta_{w}$ be such that
\begin{equation}
\delta_{t}(\vx,\vWeight)=\mixedNorm{\frac{\vc+\vWeight\vec{\phi}'(\vx)-\ma\veta_{w}}{\vWeight\sqrt{\vec{\phi}''(\vx)}}}{\vWeight}\label{eq:gen:w_change:1}
\end{equation}
Furthermore, the assumption shows that $(1+\epsilon)^{-2}\vWeight_{i}\leq\vv_{i}\leq(1+\epsilon)^{2}\vWeight_{i}$
for all $i$. Using these, we bound the energy with the new weights
as follows
\begin{align*}
\delta_{t}(\vx,\vv) & =\min_{\eta}\mixedNorm{\frac{\vc+\vv\vec{\phi}'(\vx)-\ma\veta}{\vv\sqrt{\vec{\phi}''(\vx)}}}{\vv}\leq\mixedNorm{\frac{\vc+\vv\vec{\phi}'(\vx)-\ma\veta_{\weight}}{\vv\sqrt{\vec{\phi}''(\vx)}}}{\vv}\\
 & \leq(1+\epsilon)\mixedNorm{\frac{\vc+\vv\vec{\phi}'(\vx)-\ma\veta_{\weight}}{\vv\sqrt{\vec{\phi}''(\vx)}}}{\vWeight}\\
 & \leq\left(1+\epsilon\right)\cdot\left(\mixedNorm{\frac{\vc+\vWeight\vec{\phi}'(\vx)-\ma\veta_{\weight}}{\vv\sqrt{\vec{\phi}''(\vx)}}}{\vWeight}+\mixedNorm{\frac{(\vv-\vWeight)\vec{\phi}'(\vx)}{\vv\sqrt{\vec{\phi}''(\vx)}}}{\vWeight}\right)\\
 & \leq\left(1+\epsilon\right)^{3}\delta_{t}(\vx,\vWeight)+(1+\epsilon)\cdot\normFullInf{\frac{\vphi'(\vx)}{\sqrt{\vphi''(\vx)}}}\cdot\mixedNorm{\frac{(\vv-\vWeight)}{\vv}}{\vWeight}
\end{align*}
Using that $\left|\phi'_{i}(\vx)\right|\leq\sqrt{\phi_{i}''(\vx)}$
for all $i\in[m]$ by Definition \ref{def:1-self-concordant-barrier}
and using Lemma \ref{lem:appendix:log_helper} we have that
\begin{eqnarray*}
\delta_{t}(\vx,\vv) & \leq & (1+\epsilon)^{3}\delta_{t}(\vx,\vWeight)+(1+\epsilon)^{2}\epsilon\\
 & \leq & \left(1+4\epsilon\right)\left(\delta_{t}(\vx,\vWeight)+\epsilon\right).
\end{eqnarray*}

\end{proof}

\subsection{Centering \label{sub:weighted-path:centering}}

In the previous subsection, we saw how much the weight function, $\vg(\vx)$,
can change after a Newton step on $\vx$ and we bounded how much we
can move the weights without affecting centrality too much. Here we
study how to correctly move the weights even when we cannot compute
the weight function exactly. Our solution is based on ``the chasing
$\vzero$ game'' defined in Part I \cite{lsInteriorPoint}. We restate
the main result from Part I \cite{lsInteriorPoint} on this game and
in Theorem~\ref{thm:smoothing:centering_inexact_weight} show how
to use this result to improve centrality and maintain the weights
even when we can only compute the weight function approximately.
\begin{thm}
[\cite{lsInteriorPoint}]\label{thm:zero_game} For $\vx_{0}\in\Rm$
and $0<\epsilon<\frac{1}{5}$, consider the two player game consisting
of repeating the following for $k=1,2,\ldots$
\begin{enumerate}
\item The adversary chooses $\trSetCurr\subseteq\R^{k}$, $\vu^{(k)}\in U^{(k)}$,
and sets $\trAdve=\trCurr+\vu^{(k)}$. 
\item The adversary chooses $\trMeas$ such that $\norm{\trMeas-\trAdve}_{\infty}\leq R$ 
\item The adversary reveals $\vz^{(k)}$ and $U^{(k)}$ to the player.
\item The player chooses $\vec{\Delta}^{(k)}\in\left(1+\epsilon\right)\trSetCurr$
and sets $\trNext=\trAdve+\vec{\Delta}^{(k)}.$
\end{enumerate}
Suppose that each $\trSetCurr$ is a symmetric convex set that contains
an $\ellInf$ ball of radius $r_{k}$ and is contained in a $\ellInf$
ball of radius $R_{k}\leq R$ and consider the strategy
\[
\vec{\Delta}^{(k)}=\left(1+\epsilon\right)\argmin_{\vec{\Delta}\in U^{(k)}}\left\langle \nabla\Phi_{\mu}(\vz^{(k)}),\vec{\Delta}\right\rangle 
\]
where $\mu=\frac{\epsilon}{12R}$ and $\Phi_{\mu}(\vx)=\sum_{i}\left(e^{\mu x_{i}}+e^{-\mu x_{i}}\right)$.
Let $\tau=\max_{k}\frac{R_{k}}{r_{k}}$ and suppose $\Phi_{\mu}(\trInit)\leq\frac{12m\tau}{\epsilon}$.
This strategy guarantees that for all $k$ we have
\[
\Phi_{\mu}(\trNext)\leq\left(1-\frac{\epsilon^{2}r_{k}}{24R}\right)\Phi_{\mu}(\trCurr)+\epsilon m\frac{R_{k}}{2R}\leq\frac{12m\tau}{\epsilon}.
\]
In particular, we have \textup{$\norm{\trCurr}_{\infty}\leq\frac{12R}{\epsilon}\log\left(\frac{12m\tau}{\epsilon}\right)$.} 
\end{thm}
We can think updating weight is playing this game, we want to make
sure the error between $\vWeight$ and $\vg(\vx)$ is close to $0$
while the adversary control the next point $\vg(\vx)$ and the noise
in the approximate $\vg(\vx)$. Theorem \ref{thm:zero_game} shows
that we can control the error to be small in $\ell_{\infty}$ if we
can approximate $\vg(\vx)$ with small $\ell_{\infty}$ error.

Formally, we will measure its distance from the optimal weights in
log scale by
\begin{equation}
\vWeightError(\vx,\vWeight)\defeq\log(\vg\left(\vx\right))-\log(\vWeight).\label{eq:def_psi}
\end{equation}
Our goal will be to keep $\mixedNorm{\vWeightError(\vx,\vWeight)}{\vWeight}\leq K$
for some error $K$ that is just small enough to not impair our ability
to decrease $\delta_{t}$ linearly and not to impair our ability to
approximate $\vg$. We will attempt to do this without moving $\vWeight$
too much in $\mixedNorm{\cdot}{\vWeight}$.

\begin{center}
\begin{tabular}{|l|}
\hline 
$(\next{\vx},\next{\vWeight})=\code{centeringInexact}(\vx,\vWeight,K)$\tabularnewline
\hline 
\hline 
1. $c_{k}=\frac{1}{1-c_{\delta}c_{\gamma}}$,$R=\frac{K}{48c_{k}\log\left(400m\right)}$,
$\delta_{t}=\delta_{t}(\vx,\vWeight)$ and $\epsilon=\frac{1}{2c_{k}}$.\tabularnewline
\hline 
2. $\next{\vx}=\vx-\frac{1}{\sqrt{\vec{\phi}''(\vx)}}\mProj_{\vx,\vWeight}\left(\frac{t\vc-\vWeight\vec{\phi}'(\vx)}{\vWeight\sqrt{\vec{\phi}''(\vx)}}\right).$\tabularnewline
\hline 
3. Let $U=\{\vx\in\Rm~|~\mixedNorm{\vx}{\vWeight}\leq\left(1-\frac{7}{8c_{k}}\right)\delta_{t}\}$ \tabularnewline
\hline 
4. Find $\vz$ such that $\normInf{\vz-\log(\vg(\next{\vx}))}\leq R$.\tabularnewline
\hline 
5. $\next{\vWeight}=\exp\left(\log(\vWeight)+\left(1+\epsilon\right)\argmin_{\vu\in U}\left\langle \nabla\Phi_{\frac{\epsilon}{12R}}(\vz-\log\left(\vWeight\right)),\vu\right\rangle \right)$ \tabularnewline
\hline 
\end{tabular}
\par\end{center}

The minimization problem in step 5 is simply a projection onto the
convex set $U$ and it can be done in $\tilde{O}(1)$ depth and $\tilde{O}(m)$
work. See section \ref{sub: projection_mixed_ball} for details. 
\begin{thm}
\label{thm:smoothing:centering_inexact_weight} Assume that $24m^{1/4}\geq c_{k}\defeq\frac{1}{1-c_{\delta}c_{\gamma}}\geq5$,
$1\leq\cnorm\leq2c_{k}$ and $K\leq\frac{1}{20c_{k}}.$ Let $\vWeightError(\vx,\vWeight)\defeq\log(\vg\left(\vx\right))-\log(\vWeight)$.
Suppose that
\[
\delta\defeq\delta_{t}(\vx,\vWeight)\leq\frac{K}{48c_{k}\log\left(400m\right)}\enspace\text{ and }\enspace\Phi_{\mu}(\vWeightError(\vx,\vWeight))\leq\left(400m\right)^{2}
\]
where $\mu=\frac{\epsilon}{12R}=2\log\left(400m\right)/K$. Let $(\next{\vx},\next{\vWeight})=\code{centeringInexact}(\vx,\vWeight,K)$,
then
\[
\delta_{t}(\next{\vx},\next{\vWeight})\leq\left(1-\frac{1}{4c_{k}}\right)\delta\enspace\text{ and }\enspace\Phi_{\mu}(\vWeightError(\next{\vx},\next{\vWeight}))\leq\left(400m\right)^{2}.
\]
Also, we have $\norm{\log(\vg(\next{\vx}))-\log(\vWeight)}_{\infty}\leq K$.\end{thm}
\begin{proof}
By Lemma \ref{lem:gen:w_step}, inequality (\ref{eq:centrality_equivalence}),
$c_{\delta}c_{\gamma}\leq1$ and $c_{\gamma}\leq\frac{5}{4}$ (see
Def \ref{def:gen:weight_function}), we have
\begin{eqnarray*}
\mixedNorm{\log\left(\vg(\next{\vx})\right)-\log\left(\vg(\vx)\right)}{\vg(\vx)} & \leq & c_{\delta}c_{\gamma}\delta(1+4c_{\gamma}\delta)\\
 & \leq & c_{\delta}c_{\gamma}\delta+5\delta^{2}\\
 & \leq & \left(1-\frac{15}{16c_{k}}\right)\delta.
\end{eqnarray*}
Using $K\leq\frac{1}{20c_{k}}$, we have
\[
\normFull{\frac{\vWeight-\vg\left(\vx\right)}{\vg\left(\vx\right)}}_{\infty}\leq\norm{\log\left(\vWeight\right)-\log\left(\vg\left(\vx\right)\right)}_{\infty}+\norm{\log\left(\vWeight\right)-\log\left(\vg\left(\vx\right)\right)}_{\infty}^{2}\leq\frac{21}{20}K\leq\frac{1}{16c_{k}}.
\]
Hence, we have
\begin{eqnarray*}
\mixedNorm{\log\left(\vg(\next{\vx})\right)-\log\left(\vg(\vx)\right)}{\vWeight} & \leq & \left(1+\frac{1}{16c_{k}}\right)\left(1-\frac{15}{16c_{k}}\right)\delta\\
 & \leq & \left(1-\frac{7}{8c_{k}}\right)\delta.
\end{eqnarray*}
Therefore, we know that for the Newton step, we have $\vWeightError(\next{\vx},\vWeight)-\vWeightError(\vx,\vWeight)\in U$
where $U$ is the symmetric convex set given by
\[
U\defeq\{\vx\in\R^{n}~|~\mixedNorm{\vx}{\vWeight}\leq C\}
\]
where $C=\left(1-\frac{7}{8c_{k}}\right)\delta.$ Note that from our
assumption on $\delta$, we have
\[
C\leq\delta\leq\frac{K}{48c_{k}\log\left(400m\right)}=R.
\]
It ensures that $U$ are contained in some $\ellInf$ ball of radius
$R$. Therefore, we can play the chasing 0 game on $\vWeightError(\vx,\vWeight)$
attempting to maintain the invariant that $\normInf{\vWeightError(\vx,\vWeight)}\leq K$
without taking steps that are more than $1+\epsilon$ times the size
of $U$ where we pick $\epsilon=\frac{1}{2c_{k}}$ so to not interfere
with our ability to decrease $\delta_{t}$ linearly.

However, to do this with the chasing 0 game, we need to ensure that
$R$ satisfying the following
\[
\frac{12R}{\epsilon}\log\left(\frac{12m\tau}{\epsilon}\right)\leq K
\]
where here $\tau$ is as defined in Theorem \ref{thm:zero_game}.

To bound $\tau$, we need to lower bound the radius of $\ellInf$
ball it contains. Since by assumption $\normInf{\fvWeight(\vx)}\leq2$
and $\normInf{\vWeightError(\vx,\vWeight)}\leq\frac{1}{8}$, we have
that $\normInf{\vWeight}\leq3$. Hence, we have
\[
\forall u\in\Rm\enspace:\enspace\norm{\vu}_{\infty}^{2}\geq\frac{1}{3m}\norm{\vu}_{\vWeight}^{2}.
\]
Consequently, if $\norm{\vu}_{\infty}\leq\frac{\delta}{5\cnorm\sqrt{m}}$,
then $\vu\in U$. So, we have that $U$ contains a box of radius $\frac{\delta}{5\cnorm\sqrt{m}}$
and since $U$ is contained in a box of radius $\delta$, we have
that
\begin{eqnarray*}
\tau & \leq & 5\cnorm\sqrt{m}\leq10c_{k}\sqrt{m}.
\end{eqnarray*}
Using $c_{k}\leq24m^{1/4}$, we have
\begin{eqnarray*}
\frac{12R}{\epsilon}\log\left(\frac{12m\tau}{\epsilon}\right) & \leq & 24c_{k}R\log\left(240m^{3/2}c_{k}^{2}\right)\\
 & \leq & 48c_{k}R\log\left(400m\right)=K.
\end{eqnarray*}
and
\[
\frac{12m\tau}{\epsilon}\leq240m^{3/2}c_{k}^{2}\leq(400m)^{2}.
\]
This proves that we meet the conditions of Theorem \ref{thm:zero_game}.
Consequently, $\normInf{\vWeightError(\next{\vx},\next{\vWeight})}\leq K$
and $\Phi_{\alpha}(\vWeightError(\next{\vx},\next{\vWeight}))\leq(400m)^{2}$. 

Since $K\leq\frac{1}{4}$, Lemma \ref{lem:gen:x_progress} shows that
\begin{eqnarray*}
\delta_{t}(\next{\vx},\vWeight) & \leq & 4\left(\delta_{t}(\vx,\vWeight)\right)^{2}.
\end{eqnarray*}
The step 5 shows that
\begin{eqnarray*}
\mixedNorm{\log(\vWeight)-\log(\next{\vWeight})}{\vWeight} & \leq & \left(1+\frac{1}{2c_{k}}\right)\left(1-\frac{7}{8c_{k}}\right)\delta\\
 & \leq & \left(1-\frac{3}{8c_{k}}\right)\delta.
\end{eqnarray*}
Using $\delta\leq\frac{1}{80c_{k}}$, the Lemma \ref{lem:gen:w_change}
shows that
\begin{eqnarray*}
\delta_{t}(\next{\vx},\next{\vWeight}) & \leq & \left(1+4\left(1-\frac{3}{8c_{k}}\right)\delta\right)\left(\delta_{t}(\next{\vx},\vWeight)+\left(1-\frac{3}{8c_{k}}\right)\delta\right)\\
 & \leq & \left(1+4\delta\right)\left(4\delta^{2}+\left(1-\frac{3}{8c_{k}}\right)\delta\right)\\
 & \leq & \left(1-\frac{3}{8c_{k}}\right)\delta+4\delta^{2}+16\delta^{3}+4\delta^{2}\\
 & \leq & \left(1-\frac{1}{4c_{k}}\right)\delta.
\end{eqnarray*}
\end{proof}

\section{Weight Function\label{sec:weight-function}}

In this section we present the weight function that we use to achieve
our $\otilde(\sqrt{\rank(\ma)}\log(U/\varepsilon))$ iteration linear
program solver. This weight function is similar to the one we used
in Part I \cite{lsInteriorPoint}. Due to subtle differences in the
analysis we provide many of proofs of properties of the weight function
in full. For further intuition on the weight function or proof details
see Part I \cite{lsInteriorPoint}. 

We define the weight function $\vg:\dInterior\rightarrow\dWeight$
for all $\vx\in\dWeight$ as follows 
\begin{equation}
\vg(\vx)\defeq\argmin_{\vWeight\in\rPos^{m}}\penalizedObjectiveWeight(\vx,\vWeight)\enspace\text{ where }\enspace\penalizedObjectiveWeight(\vx,\vWeight)\defeq\onesVec^{T}\vWeight+\frac{1}{\alpha}\log\det\left(\ma_{x}^{T}\mWeight^{-\alpha}\ma_{x}\right)-\beta\sum_{i\in[m]}\log\weight_{i}.\label{eq:sec:weights:weight_function}
\end{equation}
where here and in the remainder of the subsection we let $\ma_{x}\defeq(\mPhi''(\vx))^{-1/2}\ma$
and the parameters $\alpha,\beta$ are chosen later such that the
following hold 
\begin{equation}
\alpha\in[1,2)\,\text{, }\,\beta\in(0,1)\,\text{, and \,}\beta^{1-\alpha}\leq2\,\text{.}\label{eq:weights:constants_assumptions}
\end{equation}
Here we choose $\beta$ small and $\alpha$ just slightly larger than
$1$.%
\footnote{Note that this formula is different than the formula we used in \cite{lsInteriorPoint}.%
}

We start by computing the gradient and Hessian of $\penalizedObjectiveWeight(\vx,\vWeight)$
with respect to $\vWeight$. 
\begin{lem}
\label{lem:maxflow:derivatives_of_weight} For all $\vx\in\dInterior$
and $\vWeight\in\dWeights$, we have 
\[
\grad_{w}\penalizedObjectiveWeight(\vx,\vWeight)=\left(\iMatrix-\mSigma\mWeight^{-1}-\beta\mWeight^{-1}\right)\onesVec\enspace\text{ and }\enspace\hessian_{ww}\penalizedObjectiveWeight(\vx,\vWeight)=\mWeight^{-1}\left(\mSigma+\beta\iMatrix+\alpha\mLambda\right)\mWeight^{-1}
\]
where $\mSigma\defeq\mSigma_{\ma_{x}}\left(\vWeight^{-\alpha}\right)$
and $\mLapProj\defeq\mLapProj_{\ma_{x}}\left(\vWeight^{-\alpha}\right)$.\end{lem}
\begin{proof}
Using Lemma~\ref{lem:appendix:projection_matrices} and the chain
rule we compute the gradient of $\grad_{w}\penalizedObjectiveWeight(\vx,\vWeight)$
as follows 
\begin{eqnarray*}
\grad_{w}\penalizedObjectiveWeight(\vx,\vWeight) & = & \onesVec+\frac{1}{\alpha}\mSigma\mWeight^{\alpha}\left(-\alpha\mWeight^{-\alpha-1}\right)-\beta\mWeight^{-1}\onesVec\\
 & = & \left(\iMatrix-\mSigma\mWeight^{-1}-\beta\mWeight^{-1}\right)\onesVec.
\end{eqnarray*}
Next, using Lemma~\ref{lem:appendix:projection_matrices} and chain
rule, we compute the following for all $i,j\in[m]$
\begin{align*}
\frac{\partial(\grad_{w}\penalizedObjectiveWeight(\vx,\vWeight))_{i}}{\partial\vWeight_{j}} & =-\frac{\vWeight_{i}\mLambda_{ij}\vWeight_{j}^{\alpha}\left(-\alpha\vWeight_{j}^{-\alpha-1}\right)-\mSigma_{ij}\indicVec{i=j}}{\vWeight_{i}^{2}}+\beta\indicVec{i=j}\left\{ \vWeight_{i}^{-2}\right\} \\
 & =\frac{\mSigma_{ij}}{\vWeight_{i}\vWeight_{j}}+\alpha\frac{\mLambda_{ij}}{\vWeight_{i}\vWeight_{j}}+\frac{\beta\indicVec{i=j}}{\vWeight_{i}^{2}}\enspace.\tag{Using that \ensuremath{\mSigma}is diagonal}
\end{align*}
Consequently, $\hessian_{ww}\penalizedObjectiveWeight(\vx,\vWeight)=\mWeight^{-1}\left(\mSigma+\beta\iMatrix+\alpha\mLambda\right)\mWeight^{-1}$
as desired. \end{proof}
\begin{lem}
\label{lem:max_flow:existence_and_size} For all $\vx\in\dInterior$,
the weight function $\vg(\vx)$ is a well defined with 
\[
\beta\leq g_{i}(\vs)\leq1+\beta\enspace\text{ and }\enspace\norm{\vg(\vx)}_{1}=\rank(\ma)+\beta\cdot m.
\]
Furthermore, for all $\vx\in\dInterior$, the weight function obeys
the following equations 
\[
\mg(\vx)=\left(\mSigma+\beta\iMatrix\right)\onesVec\enspace\text{, and }\enspace\mg'(\vx)=-\mg(\vx)\left(\mg(\vx)+\alpha\mLapProj\right)^{-1}\mLapProj\left(\mPhi''(\vx)\right)^{-1}\mPhi'''(\vx)
\]
where $\mSigma\defeq\mSigma_{\ma_{x}}\left(\vg^{-\alpha}(\vx)\right)$,
$\mLapProj\defeq\mLapProj_{\ma_{x}}\left(\vg^{-\alpha}(\vx)\right)$,
and $\mg'(\vx)$ is the Jacobian matrix of $\vg$ at $\vx$.\end{lem}
\begin{proof}
By Lemma~\ref{lem:appendix:projection_matrices} we have that $\mSigma\specGeq\mLapProj\specGeq\mZero.$
Therefore, by Lemma~\ref{lem:maxflow:derivatives_of_weight}, we
have that $\hessian_{ww}\penalizedObjectiveWeight(\vx,\vWeight)\specGeq\beta\mWeight^{-2}$
and $\penalizedObjectiveWeight(\vx,\vWeight)$ is convex. Using the
formula for the gradient in Lemma~\ref{lem:maxflow:derivatives_of_weight},
we see that that for all $i\in[m]$ it is the case that 
\[
\left[\gradient_{w}\penalizedObjectiveWeight(\vx,\vWeight)\right]_{i}=\frac{1}{w_{i}}\left(w_{i}-\mSigma_{ii}-\beta\right).
\]
Using that $0\leq\sigma_{i}\leq1$ for all $i$ by Lemma~\ref{lem:appendix:projection_matrices}
and $\beta\in(0,1)$ by (\ref{eq:weights:constants_assumptions}),
we see that if $\vWeight_{i}\in(0,\beta)$ then $\left[\gradient_{w}\penalizedObjectiveWeight(\vx,\vWeight)\right]_{i}$
is strictly negative and if $\vWeight_{i}\in(1+\beta,\infty)$ then
$\left[\gradient_{w}\penalizedObjectiveWeight(\vx,\vWeight)\right]_{i}$
is strictly positive. Therefore, for any $\vx\in\dInterior$, the
$\vWeight$ that minimizes this convex function $\penalizedObjectiveWeight(\vx,\vWeight)$
lies between the box between $\beta$ to $1+\beta$. Since $\penalizedObjectiveWeight$
is strongly convex in this region, the minimizer is unique.

The formula for $\mg(\vx)$ follows by setting $\grad_{w}\penalizedObjectiveWeight(\vx,\vWeight)=\vzero$
and the size of $g(\vx)$ follows from the fact that $\norm{\vsigma}_{1}=\mathrm{tr}\left(\mProj_{\ma_{x}}\left(\vg^{-\alpha}(\vx)\right)\right)$.
Since $\mProj_{\ma_{x}}\left(\vg^{-\alpha}(\vx)\right)$ is a projection
onto the image of $\mg(\vx)^{-\alpha/2}\ma_{x}$ and since $\vg(\vx)>\vzero$
and $\vec{\phi}''(\vx)>\vzero$, we have that the dimension of the
image of $\mg(\vx)^{-\alpha/2}\ma_{x}$ is the rank of $\ma$. Hence,
we have that $\normOne{\vg(\vx)}=\rank(\ma)+\beta\cdot m$. 

By Lemma~\ref{lem:appendix:projection_matrices} and chain rule,
we get the following for all $i,j\in[m]$ 
\begin{align*}
\frac{\partial(\grad_{\vWeight}\penalizedObjectiveWeight(\vx,\vWeight))_{i}}{\partial\vx_{j}} & =-\vWeight_{i}^{-1}\mLambda_{ij}\vec{\phi}''_{j}(\vx)\left(-(\vec{\phi}''_{j}(\vx))^{-2}\vec{\phi}'''_{j}(\vx)\right)=\vWeight_{i}^{-1}\mLambda_{ij}(\vec{\phi}''_{j}(\vx))^{-1}\vec{\phi}'''_{j}(\vx).
\end{align*}
Consequently, $\jacobian_{\vx}(\grad_{\vWeight}\penalizedObjectiveWeight(\vx,\vWeight))=\mWeight^{-1}\mLapProj\left(\mPhi''(\vx)\right)^{-1}\mPhi'''(\vx)$
where $\jacobian_{\vx}$ denotes the Jacobian matrix of the function
$\grad_{\vWeight}\penalizedObjectiveWeight(\vx,\vWeight)$ with respect
to $\vx$. Since we have already know that $\jacobian_{\vWeight}(\grad_{\vWeight}\penalizedObjectiveWeight(\vx,\vWeight))=\hessian_{\vWeight\vWeight}f_{t}(\vx,\vWeight)=\mWeight^{-1}\left(\mSigma+\beta\iMatrix+\alpha\mLambda\right)\mWeight^{-1}$
is positive definite (and hence invertible), by applying the implicit
function theorem to the specification of $\vg(\vx)$ as the solution
to $\grad_{\vWeight}\penalizedObjectiveWeight(\vx,\vWeight)=\vzero$,
we have
\[
\mg'(\vx)=-\left(\jacobian_{\vWeight}(\grad_{\weight}\penalizedObjectiveWeight(\vx,\vWeight))\right)^{-1}\left(\jacobian_{\vx}(\grad_{\vWeight}\penalizedObjectiveWeight(\vx,\vWeight))\right)=-\mg(\vx)\left(\mg(\vx)+\alpha\mLapProj\right)^{-1}\mLapProj\left(\mPhi''(\vx)\right)^{-1}\mPhi'''(\vx).
\]

\end{proof}
Now we show the step consistency of $\vg$.
\begin{lem}
[Step Consistency]\label{lem:weights_full:consistency-dual} For
all $\vx\in\dInterior$ and $\vy\in\Rm$, and

\[
\mb\defeq\mg(\vx)^{-1}\mg'(\vx)(\vphi(\vx)'')^{-1/2},
\]
we have 
\[
\normFull{\mb\vy}_{\mg(\vx)}\leq\frac{2}{1+\alpha}\norm{\vy}_{\mg(\vx)}\enspace\text{ and }\enspace\norm{\mb\vy}_{\infty}\leq\frac{2}{1+\alpha}\left(\norm{\vy}_{\infty}+\frac{1+2\alpha}{1+\alpha}\norm{\vy}_{\mg(\vx)}\right).
\]
Therefore
\[
\mixedNorm{\mb}{\vg}\leq\frac{2}{1+\alpha}\left(1+\frac{2}{\cnorm}\right).
\]
\end{lem}
\begin{proof}
Fix an arbitrary $\vx\in\dInterior$ and let $\fvWeight\defeq\vg(\vx)$,
$\vsigma\defeq\vLever_{\ma_{x}}\left(\vg^{-\alpha}(\vx)\right)$,
$\mSigma\defeq\mSigma_{\ma_{x}}\left(\vg^{-\alpha}(\vx)\right)$,
$\mProj\defeq\mProj_{\ma_{x}}\left(\vg^{-\alpha}(\vx)\right)$, $\mLapProj\defeq\mLapProj_{\ma_{x}}\left(\vg^{-\alpha}(\vx)\right)$.
Also, fix an arbitrary $\vy\in\Rm$ and let $\vz\defeq\mb\vy$.

By Lemma \ref{lem:max_flow:existence_and_size}, $\mg'=-\mg\left(\mg+\alpha\mLapProj\right)^{-1}\mLapProj\left(\mPhi''\right)^{-1}\mPhi'''$
and therefore 
\begin{align*}
\mb & =-\mg^{-1}\left(\fmWeight\left(\fmWeight+\alpha\mLambda\right)^{-1}\mLambda\left(\mPhi''\right)^{-1}\mPhi'''\right)\left(\mPhi''\right)^{-1/2}\\
 & =\left(\mg+\alpha\mLambda\right)^{-1}\left(2\mLambda\right)\mDiag\left(\frac{-\vphi'''}{2(\vphi'')^{3/2}}\right).
\end{align*}
Let $\mc\defeq\left(\mg+\alpha\mLambda\right)^{-1}\left(2\mLambda\right)$
and let $\vy'\defeq\mDiag\left(\frac{-\vphi'''}{2(\vphi'')^{3/2}}\right)\vy$.
By the self concordance of $\vphi$ (Definition \ref{def:1-self-concordant-barrier})
we know that $\norm{\vy'}\leq\norm{\vy}$ for both $\norm{\cdot}_{\mg}$
and $\norm{\cdot}_{\infty}$ . Since $\vz=\mb\vy=\mc\vy'$, it suffices
to bound $\norm{\mc\vy'}$ in terms of $\norm{\vy'}$ for the necessary
norms.

Letting $\bar{\mLambda}\defeq\mg^{-1/2}\mLambda\mg^{-1/2}$, we simplify
the equation further and note that
\[
\norm{\mc}_{\mg}=\norm{\mg^{1/2}\left(\mg+\alpha\mLambda\right)^{-1}\left(2\mLambda\right)\mg^{-1/2}}_{2}=\norm{\left(\iMatrix+\alpha\bar{\mLambda}\right)^{-1}\left(2\bar{\mLambda}\right)}_{2}.
\]
Now, for any eigenvector, $\vv$, of $\bar{\mLambda}$ with eigenvalue
$\lambda$, we see that $\vv$ is an eigenvector of $(\iMatrix+\alpha\bar{\mLambda})^{-1}(2\bar{\mLambda})$
with eigenvalue $2\lambda/(1+\alpha\lambda)$. Furthermore, since
$\mZero\specLeq\bar{\mLambda}\specLeq\iMatrix$, we have that $\norm{\mc}_{\mg}\leq2/(1+\alpha)$
and hence $\norm{\vz}_{\mg}\leq2(1+\alpha)^{-1}\norm{\vy'}_{\mg}\leq2(1+\alpha)^{-1}\norm{\vy}_{\mg}$
as desired.

To bound $\norm{\vz}_{\infty}$, we use that $\left(\mg+\alpha\mLambda\right)\vz=2\mLambda\vy'$,
$\mLapProj=\mLever-\shurSquared{\mProj}$, and $\fmWeight=\mSigma+\beta\iMatrix$
to derive 
\[
\left(1+\alpha\right)\mSigma\vz+\beta\vz-\alpha\mProj^{(2)}\vz=2\mSigma\vy'-2\mProj^{(2)}\vy'.
\]
Looking at the $i^{th}$ coordinate of both sides and using that $\vsigma_{i}\geq0$,
we have 
\begin{align*}
 & \left((1+\alpha)\vsigma_{i}+\beta\right)\left|\vz_{i}\right|\\
\leq & \alpha\left|[\mProj^{(2)}\vz]_{i}\right|+2\vsigma_{i}\norm{\vy'}_{\infty}+2\left|[\mProj^{(2)}\vy']_{i}\right|\\
\leq & \alpha\vsigma_{i}\norm{\vz}_{\mSigma}+2\vsigma_{i}\norm{\vy'}_{\infty}+2\vsigma_{i}\norm{\vy'}_{\mSigma}\tag{Lemma \ref{lem:appendix:projection_matrices} }\\
\leq & 2\vsigma_{i}\norm{\vy'}_{\infty}+\vsigma_{i}\left(\frac{2\alpha}{1+\alpha}+2\right)\norm{\vy'}_{\mg}\tag{\ensuremath{\mSigma\specLeq\mg}and \ensuremath{\norm{\vz}_{\mg}\leq2(1+\alpha)^{-1}\norm{\vy'}_{\mg}}}
\end{align*}
Hence, we have
\begin{eqnarray*}
|\vz_{i}| & \leq & \frac{2}{1+\alpha}\norm{\vy'}_{\infty}+\frac{1}{1+\alpha}\left(\frac{2\alpha}{1+\alpha}+2\right)\norm{\vy'}_{\mg}\\
 & \leq & \frac{2}{1+\alpha}\left[\norm{\vy'}_{\infty}+2\norm{\vy'}_{\mg}\right].
\end{eqnarray*}
Therefore, $\norm{\mb\vy}_{\infty}=\norm{\vz}_{\infty}\leq2(1+\alpha)^{-1}(\norm{\vy'}_{\infty}+2\norm{\vy'}_{\fmWeight})$.
Finally, we note that
\begin{align*}
\mixedNorm{\mb\vy}{\vg} & =\norm{\mb\vy}_{\infty}+\cnorm\norm{\mb\vy}_{\mg}\tag{Definition}\\
 & \leq\frac{2}{1+\alpha}\norm{\vy}_{\infty}+\frac{2}{1+\alpha}\cdot2\norm{\vy}_{\mg}+\frac{2}{1+\alpha}\cnorm\norm{\vy}_{\mg}\\
 & \leq\frac{2}{1+\alpha}\left(1+\frac{2}{\cnorm}\right)\mixedNorm{\vy}{\vg}.
\end{align*}
\end{proof}
\begin{thm}
\label{thm:max_flow:weight_properties} Choosing parameters
\[
\alpha=1+\frac{1}{\log_{2}\left(\frac{2m}{\rank(\ma)}\right)}\enspace,\enspace\beta=\frac{\rank(\ma)}{2m}\enspace\text{, and }\enspace\cnorm=18\log_{2}\left(\frac{2m}{\rank(\ma)}\right)
\]
yields 
\[
\cWeightSize(\vg)=2\rank(\ma)\enspace,\enspace\cWeightStab(\vg)=1+\frac{1}{9\log_{2}\left(\frac{2m}{\rank(\ma)}\right)}\enspace\text{, and }\enspace\cWeightCons(\fvWeight)=1-\frac{2}{9\log_{2}\left(\frac{2m}{\rank(\ma)}\right)}.
\]
In particular, we have\textup{
\[
\cWeightStab(\vg)\cWeightCons(\fvWeight)\leq1-\frac{1}{9\log_{2}\left(\frac{2m}{\rank(\ma)}\right)}.
\]
}\end{thm}
\begin{proof}
The bounds on $c_{1}(\vg)$ and $c_{\delta}(\vg)$ follow immediately
from Lemma \ref{lem:max_flow:existence_and_size} and Lemma \ref{lem:weights_full:consistency-dual}.
Now, we estimate the $c_{\gamma}(\vg)$ and let $\frac{4}{5}\vg\leq\vWeight\leq\frac{5}{4}\vg$.
Fix an arbitrary $\vx\in\dInterior$ and let $\fvWeight\defeq\vg(\vx)$.
Recall that by Lemma~\ref{lem:max_flow:existence_and_size}, we have
$\vg\geq\beta$. Furthermore, since $\vg^{-1}=\vg^{\alpha-1}\vg^{-\alpha}$
and $\beta^{\alpha-1}\geq\frac{1}{2}$, the following holds 
\begin{equation}
\frac{4}{10}\vg_{i}^{-\alpha}\leq\frac{4}{5}\beta^{\alpha-1}\vg_{i}^{-\alpha}\leq\frac{4}{10}\vg_{i}^{-1}\leq\vWeight_{i}^{-1}\label{eq:weights_full:cond1}
\end{equation}
for all $i$. Applying this and using the definition of $\mProj_{\ma_{x}}$
yields 
\begin{equation}
\ma_{x}(\ma_{x}^{T}\mWeight^{-1}\ma_{x})^{-1}\ma_{x}^{T}\specLeq\frac{10}{4}\ma_{x}(\ma_{x}^{T}\mg^{-\alpha}\ma_{x})^{-1}\ma_{x}^{T}=\frac{10}{4}\mg^{\alpha/2}\mProj_{\ma_{x}}(\vg^{-\alpha})\mg^{\alpha/2}\enspace.\label{eq:weights_full:cond2}
\end{equation}
Hence, we have
\begin{eqnarray*}
\frac{\vsigma_{i}\left(\frac{1}{\vWeight\vec{\phi}''}\right)}{\vWeight_{i}} & = & \frac{\onesVec_{i}^{T}\ma_{x}(\ma_{x}^{T}\mWeight^{-1}\ma_{x})^{-1}\ma_{x}^{T}\onesVec_{i}}{\vWeight_{i}^{2}}\\
 & \leq & \frac{10}{4}\frac{\onesVec_{i}^{T}\mg^{\alpha/2}\mProj_{\ma_{x}}(\vg^{-\alpha})\mg^{\alpha/2}\onesVec_{i}}{\vWeight_{i}^{2}}\\
 & \leq & \frac{10}{4}\left(\frac{5}{4}\right)^{2}\frac{\vsigma_{i}\left(\frac{1}{\vg^{\alpha}\vec{\phi}''}\right)}{\vg_{i}^{-2\alpha}}<4.
\end{eqnarray*}
Since $\mProj_{\vx,\vWeight}$ is an orthogonal projection in $\norm{\cdot}_{\vWeight}$,
we have $\normFull{\mProj_{\vx,\vWeight}}_{\vWeight\rightarrow\vWeight}=1.$
Let $\overline{\mProj}_{\vx,\vWeight}\defeq\iMatrix-\mProj_{\vx,\vWeight}$,
we have
\begin{align*}
\normFull{\overline{\mProj}_{\vx,\vWeight}}_{\vWeight\rightarrow\infty} & =\max_{i\in[m]}\max_{\norm{\vy}_{\vWeight}\leq1}\onesVec_{i}^{T}\overline{\mProj}_{\vx,\vWeight}\vy\\
 & \leq\max_{i\in[m]}\norm{\left(\vWeight\right)^{-1/2}\overline{\mProj}_{\vx,\vWeight}^{T}\onesVec_{i}}^{2}\\
 & =\max_{i\in[m]}\sqrt{\onesVec_{i}\mWeight^{-1}\ma_{x}\left(\ma_{x}^{T}\mWeight^{-1}\ma_{x}\right)^{-1}\ma_{x}^{T}\mWeight^{-1}\onesVec_{i}}.\\
 & =\max_{i\in[m]}\sqrt{\frac{\sigma_{i}\left(\frac{1}{\vWeight\vec{\phi}''}\right)}{w_{i}}}\leq2.
\end{align*}
For any $\vy$, we have
\begin{eqnarray*}
\mixedNorm{\mProj_{\vx,\vWeight}\vy}{\vWeight} & \leq & \normFull{\mProj_{\vx,\vWeight}\vy}_{\infty}+\cnorm\normFull{\mProj_{\vx,\vWeight}\vy}_{\vWeight}\\
 & \leq & \normFull{\vy}_{\infty}+\normFull{\overline{\mProj}_{\vx,\vWeight}\vy}_{\infty}+\cnorm\normFull{\vy}_{\vWeight}\\
 & \leq & \normFull{\vy}_{\infty}+(2+\cnorm)\normFull{\vy}_{\vWeight}\\
 & \leq & \frac{\cnorm+2}{\cnorm}\mixedNorm{\vy}{\vWeight}.
\end{eqnarray*}
Hence, we have $c_{\gamma}\leq\frac{\cnorm+2}{\cnorm}.$ Thus, we
have picked $\cnorm=\frac{18}{\alpha-1}$ and have $c_{\gamma}\leq1+\frac{\alpha-1}{9}.$
\end{proof}

\subsection{Computing and Correcting Weight Function }

Here we discuss how to compute the weight function using gradient
descent and dimension reduction techniques as in \cite{spielmanS08sparsRes}
for approximately computing leverage scores. The algorithm and the
proof is essentially the same as in Part I \cite{lsInteriorPoint},
modified to the subtle changes in the weight function.
\begin{thm}
[Weight Computation and Correction]\label{thm:weights_full:approximate_weight_withoutproof}
There is an algorithm, $\code{computeWeight}(\vx,\vWeight^{(0)},K)$,
that given $K<1$ and $\{\vx^{(0)},\vWeight^{(0)}\}\in\dFull$ such
that $\normInf{\mWeight_{(0)}^{-1}(\vg(\vx)-\vWeight^{(0)})}\leq\frac{1}{48}$
the algorithm returns $\vWeight\in\rPos^{m}$ such that
\[
\normInf{\mg(\vx)^{-1}(\vg(\vx)-\vWeight)}\leq K
\]
with probability $(1-\frac{1}{m})^{O(\log^{2}(m/K))}$ using only
$\tilde{O}(\log^{3}(1/K)/K^{2})$ linear system solves. 

Without the initial weight $\vWeight^{(0)}$, there is an algorithm,
$\code{computeInitialWeight}(\vx,K)$, that returns a weight with
same guarantee with constant probability using only $\tilde{O}(\sqrt{\rank\left(\ma\right)}\log^{3}(1/K)/K^{2})$
times linear system solves.\end{thm}
\begin{proof}
Let $Q=\{\vWeight:\normInf{\mWeight_{(0)}^{-1}(\vWeight-\vWeight^{(0)})}\leq\frac{1}{48}\}.$
From our assumption, $\vg(\vx)\in Q$. For any $\vWeight\in Q$, it
is easy to see that
\[
\frac{4}{5}\mWeight^{-1}\specLeq\hessWW\hat{f}(\vx,\vWeight)\specLeq4\mWeight^{-1}.
\]
Therefore, in this region $Q$, the function is well conditioned and
gradient descent converges to the minimizer of $\hat{f}$ quickly.
Note that a gradient descent step projected on $Q$ can be written
as
\[
\vWeight^{(j)}=\code{median}\left(\left(1-\frac{1}{48}\right)\vWeight^{(0)},\frac{3}{4}\vWeight^{(j-1)}+\frac{1}{4}\vsigma_{\ma_{x}}\left(\left(\vWeight^{(k)}\right)^{-\alpha}\right)+\frac{\beta}{4},\left(1+\frac{1}{48}\right)\vWeight^{(0)}\right).
\]
Similarly to \cite{lsInteriorPoint}, one can show that the iteration
is stable under noise induced by approximate leverage score computation
and therefore yields the desired approximation of $\vg(\vx)$ assuming
we can compute $\vsigma_{\ma_{x}}$ with small multiplicative $\ell_{\infty}$
error. Since such leverage scores can be computed with high probability
by solving $\tilde{O}(1)$ linear systems \cite{spielmanS08sparsRes}
we have that there is an algorithm $\code{computeWeight}$ as desired.

To compute the initial weight, we follow the approach in Part I \cite{lsInteriorPoint}.
Note that if $\beta=100$, then $\vWeight=100$ is a good approximation
of $\vg(\vx)$. Consequently, we can repeatedly use $\code{computeWeight}$
to compute the $\vg(\vx)$ for a certain $\beta$ and then decrease
$\beta$ by a factor of $1-\sqrt{\rank\ma}$. This algorithm converges
in $\tilde{O}\left(\sqrt{\rank\ma}\right)$ iterations yielding the
desired result.\end{proof}

\section{The Algorithm}

\label{sec:master_thm}

Here we show how to use the results of previous sections to solve
\eqref{eq:lp} using exact linear system solver. In the next section
we will discuss how to relax this assumption. The central goal of
this section is to develop an algorithm, $\code{LPSolve}$, for which
we can prove the following theorem
\begin{thm}
\label{thm:LPSolve_detailed}Suppose we have an interior point $\vx_{0}\in\dInterior$
for the linear program \eqref{eq:lp}.Then, the algorithm $\code{LPSolve}$
outputs $\vx$ such that $\vc^{T}\vx\leq\text{OPT}+\epsilon$ in $\tilde{O}\left(\sqrt{\rank(\ma)}\left(\mathcal{T}_{w}+\nnz(\ma)\right)\log\left(U/\epsilon\right)\right)$work
and $\tilde{O}\left(\sqrt{\rank(\ma)}\mathcal{T}_{d}\log\left(U/\epsilon\right)\right)$
depth where $U=\max\left(\normFull{\frac{\vu-\vl}{\vu-\vx_{0}}}_{\infty},\normFull{\frac{\vu-\vl}{\vx_{0}-\vl}}_{\infty},\norm{\vu-\vl}_{\infty},\norm{\vc}_{\infty}\right)$
and $\mathcal{T}_{w}$ and $\mathcal{T}_{d}$ is the work and depth
needed to compute $\left(\ma^{T}\md\ma\right)^{-1}\vq$ for input
positive definite diagonal matrix $\md$ and vector $\vq$.
\end{thm}
We break this proof into several parts. First we provide Lemma~\ref{lem:weighted_path:duality_gap},
and adaptation of a proof from \cite[Thm 4.2.7]{Nesterov2003} that
allows us to reason about the effects of making progress along the
weighted central path. Then we provide Lemma~\ref{lem: distance gurantee}
that we use to bound the distance to the weighted central path in
terms of centrality. After that in Lemma~ \ref{lem: distance gurantee},
we analyze a subroutine, $\code{pathFollowing}$, for following the
weighted central path. Using these lemmas we conclude by describing
our $\code{LPSolve}$ algorithm and proving Theorem~\ref{thm:LPSolve_detailed}.
\begin{lem}
[{\cite[Theorem 4.2.7]{Nesterov2003}}]\label{lem:weighted_path:duality_gap}
Let $x^{*}\in\Rm$ denote an optimal solution to \eqref{eq:lp} and
$\vx_{t}=\arg\min f_{t}\left(\vx,\vWeight\right)$ for some $t>0$
and $\vWeight\in\dWeight$. Then the following holds
\[
\vc^{T}\vx_{t}(\vWeight)-\vc^{T}\vx^{*}\leq\frac{\norm{\vWeight}_{1}}{t}.
\]
\end{lem}
\begin{proof}
By the optimality conditions of \eqref{eq:lp} we know that $\grad_{x}f_{t}(\vx_{t}(\vWeight))=t\cdot\vc+\vWeight\vec{\phi}'(\vx_{t}(\vWeight))$
is orthogonal to the kernel of $\ma^{T}$. Furthermore since $\vx_{t}(\vWeight)-\vx^{*}\in\ker(\ma^{T})$
we have
\[
\left(t\cdot\vc+\vWeight\vec{\phi}'(\vx_{t}(\vWeight))\right)^{T}(\vx_{t}(\vWeight)-\vx^{*})=0.
\]
Using that $\phi_{i}'(x_{t}(\vWeight)_{i})\cdot(x_{i}^{*}-x_{t}(\vWeight)_{i})\leq1$
by Lemma \ref{lem:gen:phi_properties_force} then yields
\begin{align*}
\vc^{T}(\vx_{t}(\vWeight)-\vx^{*}) & =\frac{1}{t}\sum_{i\in[m]}w_{i}\cdot\phi_{i}'(x_{t}(\vWeight)_{i})\cdot(x_{i}^{*}-x_{t}(\vWeight)_{i})\leq\frac{\norm{\vWeight}_{1}}{t}.
\end{align*}
\end{proof}
\begin{lem}
\label{lem: distance gurantee}For $\delta_{t}(\vx^{(1)},\vg(\vx^{(1)}))\leq\frac{1}{960c_{k}^{2}\log\left(400m\right)}$
and $\vx_{t}\defeq\arg\min f_{t}\left(\vx,\vWeight\right)$ we have
\[
\normFull{\sqrt{\vec{\phi}''(\vx_{t})}\left(\vx^{(1)}-\vx_{t}\right)}_{\infty}\leq16c_{\gamma}c_{k}\delta_{t}(\vx^{(1)},\vg(\vx^{(1)})).
\]
\end{lem}
\begin{proof}
We use Theorem \ref{thm:smoothing:centering_inexact_weight} with
exact weight computation and start with $\vx^{(1)}$ and $\vWeight^{(1)}=\vg(\vx^{(1)})$.
In each iteration, $\delta_{t}$ is decreased by a factor of $\left(1-\frac{1}{4c_{k}}\right)$.
\eqref{eq:centrality_equivalence} shows that 
\[
\norm{\sqrt{\vec{\phi}''(\vx^{(k)})}\left(\vx^{(k+1)}-\vx^{(k)}\right)}_{\infty}\leq c_{\gamma}\delta_{t}(\vx^{(k)},\vWeight^{(k)}).
\]
The Lemma \ref{lem:gen:phi_properties_sim} shows that
\begin{eqnarray*}
\normFull{\log\left(\vphi''(\vx^{(k)})\right)-\log\left(\vphi''(\vx^{(k+1)})\right)}_{\infty} & \leq & \left(1-2c_{\gamma}\delta_{t}(\vx^{(k)},\vWeight^{(k)})\right)^{-1}\\
 & \leq & e^{4c_{\gamma}\delta_{t}(\vx^{(k)},\vWeight^{(k)})}.
\end{eqnarray*}
Therefore, for any $k$, we have
\begin{eqnarray*}
\normFull{\log\left(\vphi''(\vx^{(1)})\right)-\log\left(\vphi''(\vx^{(k)})\right)}_{\infty} & \leq & e^{4c_{\gamma}\sum\delta_{t}(\vx^{(k)},\vWeight^{(k)})}\\
 & \leq & e^{32c_{k}c_{\gamma}\delta_{t}(\vx^{(1)},\vg(\vx^{(1)}))}\\
 & \leq & 2.
\end{eqnarray*}
Hence, for any $k$, we have
\begin{eqnarray*}
\normFull{\sqrt{\vec{\phi}''(\vx_{t})}\left(\vx^{(1)}-\vx^{(k)}\right)}_{\infty} & \leq & \sum2c_{\gamma}\delta_{t}(\vx^{(k)},\vWeight^{(k)})\\
 & \leq & 16c_{\gamma}c_{k}\delta_{t}(\vx^{(1)},\vg(\vx^{(1)})).
\end{eqnarray*}
It is clear now $\vx^{(k)}$ forms a Cauchy sequence and converges
to $\vx_{t}$ because $\delta_{t}$ continuous and $\vx_{t}$ is the
unique point such that $\delta_{t}=0$.
\end{proof}
Next, we put together the results of Section~\ref{sec:weighted-path}
and analyze the following algorithm for following the weighted central
path.

\begin{center}
\begin{tabular}{|l|}
\hline 
\textbf{$\ensuremath{(\vx^{\text{(final)}},\next{\vWeight})=\code{pathFollowing}(\vx,\vWeight,t_{\text{start}},t_{\text{end}},\epsilon)}$}\tabularnewline
\hline 
\hline 
1. $c_{k}=9\log_{2}\left(\frac{2m}{\rank(\ma)}\right),t=t_{\text{start}},K=\frac{1}{20c_{k}}.$\tabularnewline
\hline 
2. While ($t<t_{end}$ if $t_{\text{start}}<t_{\text{end}}$) or ($t>t_{end}$
if $t_{\text{start}}>t_{\text{end}}$)\tabularnewline
\hline 
2a. $(\next{\vx},\next{\vWeight})=\code{centeringInexact}(\vx,\vWeight,K)$\tabularnewline
where it use the function $\code{computeWeight}$ to find the approximation
of $\vg(\vx)$.\tabularnewline
\hline 
2b. $t\leftarrow t\left(1\pm\frac{1}{10^{5}c_{k}^{4}\log\left(400m\right)\sqrt{\rank\left(\ma\right)}}\right)$
where the sign of $\pm$ is the sign of $t_{\text{end}}-t_{\text{start}}$\tabularnewline
\hline 
2c. $\vx\leftarrow\vx^{\text{(final)}}$, $\vWeight\leftarrow\next{\vWeight}$.\tabularnewline
\hline 
3. Repeat $4c_{k}\log\left(1/\epsilon\right)$ times\tabularnewline
\hline 
3a. $(\vx,\vWeight)=\code{centeringInexact}(\vx,\vWeight,K)$\tabularnewline
where it use the function $\code{computeWeight}$ to find the approximation
of $\vg(\vx)$.\tabularnewline
\hline 
4. Output $(\vx,\vWeight)$.\tabularnewline
\hline 
\end{tabular}
\par\end{center}
\begin{thm}
\label{thm:LPSolve}Suppose that
\[
\delta_{t_{\text{start}}}(\vx,\vWeight)\leq\frac{1}{960c_{k}^{2}\log\left(400m\right)}\enspace\text{ and }\enspace\Phi_{\mu}(\vWeightError(\vx,\vWeight))\leq\left(400m\right)^{2}.
\]
where $\mu=2\log\left(400m\right)/K$. Let $\ensuremath{(\vx^{\text{(final)}},\next{\vWeight})=\code{pathFollowing}(\vx,\vWeight,t_{\text{start}},t_{\text{end}})}$,
then
\[
\delta_{t_{\text{end}}}(\vx^{\text{(final)}},\next{\vWeight})\leq\epsilon\enspace\text{ and }\enspace\Phi_{\mu}(\vWeightError(\vx^{\text{(final)}},\next{\vWeight}))\leq\left(400m\right)^{2}.
\]
Furthermore, $\code{pathFollowing}(\vx,\vWeight,t_{\text{start}},t_{\text{end}})$
takes time $\tilde{O}\left(\sqrt{\rank(\ma)}\left(\left|\log\left(\frac{t_{\text{end}}}{t_{\text{start}}}\right)\right|+\log\left(1/\epsilon\right)\right)\left(\mathcal{T}+m\right)\right)$
where $\mathcal{T}$ is the time needed to solve on linear system.\end{thm}
\begin{proof}
This algorithm maintains the invariant that $\delta_{t}(\vx,\vWeight)\leq\frac{1}{960c_{k}^{2}\log\left(400m\right)}$
and $\Phi_{\alpha}(\vWeightError(\vx,\vWeight))\leq\left(400m\right)^{2}$
on each iteration in the beginning of the step (2a). Theorem~\ref{thm:smoothing:centering_inexact_weight}
shows that
\begin{equation}
\norm{\log(\vg(\next{\vx}))-\log(\vWeight)}_{\infty}\leq K\leq\frac{1}{20c_{k}}.\label{eq:gw_bound}
\end{equation}
Thus, the weight satisfies the condition of Theorem \ref{thm:weights_full:approximate_weight_withoutproof}
and the algorithm $\code{centeringInexact}$ can use the function
$\code{computeWeight}$ to find the approximation of $\vg(\next{\vx})$.
Consequently, 
\[
\delta_{t}(\vx^{\text{(final)}},\next{\vWeight})\leq\left(1-\frac{1}{4c_{k}}\right)\delta_{t}\enspace\text{ and }\enspace\Phi_{\alpha}(\vWeightError(\vx^{\text{(final)}},\next{\vWeight}))\leq\left(400m\right)^{2}\,.
\]
Using Lemma \ref{lem:gen:t_step}, \eqref{eq:gw_bound} and Theorem
\ref{thm:max_flow:weight_properties},we have
\[
\delta_{t}(\vx^{\text{(final)}},\next{\vWeight})\leq\frac{1}{960c_{k}^{2}\log\left(400m\right)}.
\]
Hence, we proved that for every step (2c), we have the invariant.
The $\delta_{t}<\epsilon$ bounds follows from the last loop.
\end{proof}
\begin{center}
\begin{tabular}{|l|}
\hline 
\textbf{$\ensuremath{\vx^{\text{(final)}}=\code{LPSolve}(\vx,\epsilon)}$}\tabularnewline
\hline 
\hline 
Input: an initial point $\vx$.\tabularnewline
\hline 
1. $\beta=\frac{\rank(\ma)}{2m}$, $\vWeight=\code{computeInitialWeight}(\vx,\frac{1}{10^{5}\log^{5}\left(400m\right)}),$
$d=-\vWeight_{i}\phi_{i}'(\vx)$.\tabularnewline
\hline 
2. $t_{1}=(10^{10}U^{2}m^{3})^{-1}$, $t_{2}=3m/\epsilon$, $\epsilon_{1}=\frac{1}{2000c_{k}^{2}\log\left(400m\right)}$,
$\epsilon_{2}=\frac{\epsilon}{100^{3}m^{3}U^{2}}$.\tabularnewline
\hline 
3. $(\next{\vx},\next{\vWeight})=\ensuremath{\code{pathFollowing}(\vx,\vWeight,1,t_{1},\epsilon_{1})}$
with cost vector $\vd$.\tabularnewline
\hline 
4.$(\vx^{\text{(final)}},\vWeight^{\text{(final)}})=\ensuremath{\code{pathFollowing}(\next{\vx},\next{\vWeight},t_{1},t_{2},\epsilon_{2})}$
with cost vector $\vc$.\tabularnewline
\hline 
5. Output $\vx^{\text{(final)}}$.\tabularnewline
\hline 
\end{tabular}
\par\end{center}
\begin{proof}[Proof of Theorem~\ref{thm:LPSolve_detailed}]
By Theorem \ref{thm:weights_full:approximate_weight_withoutproof},
we know step 1 gives an weight 
\[
\normInf{\mg(\vx)^{-1}(\vg(\vx)-\vWeight)}\leq\frac{1}{10^{5}\log^{5}\left(400m\right)}.
\]
By the definition of $\vd$, we have $\vx$ is the minimum of
\[
\min\vd^{T}\vx-\sum\vWeight_{i}\phi_{i}(\vx)\text{ given }\ma^{T}\vx=\vb.
\]
Therefore, $(\vx,\vWeight)$ satisfies the assumption of theorem \ref{thm:LPSolve}
because $\delta_{t}=0$ and $\Phi_{\alpha}$ is small enough. Hence,
we have
\[
\delta_{t_{1}}(\next{\vx},\next{\vWeight})\leq\frac{1}{2000c_{k}^{2}\log\left(400m\right)}\enspace\text{ and }\enspace\Phi_{\alpha}(\vWeightError(\next{\vx},\next{\vWeight}))\leq(400m)^{2}.
\]
Lemma~\ref{lem:gen:phi_properties_force} shows that $\norm{\phi_{i}'(\vx)}_{\infty}\le U$
and hence $\norm{\vc-\vd}_{\infty}\leq2U.$ Also, Lemma \ref{lem:gen:phi_properties_sim}
shows that $\min_{\vy}\sqrt{\vec{\phi}''(\vy)}\geq\frac{1}{U}$. Therefore,
we have
\begin{eqnarray*}
\delta_{t_{1}}^{\vc}(\next{\vx},\next{\vWeight}) & = & \min_{\veta\in\Rn}\normFull{\frac{t_{1}\vc+\vWeight\vec{\phi}'(\next{\vx})-\ma\veta}{\next{\vWeight}\sqrt{\vec{\phi}''(\next{\vx})}}}_{\next{\vWeight}+\infty}\\
 & \leq & \min_{\veta\in\Rn}\normFull{\frac{t_{1}\vd+\vWeight\vec{\phi}'(\next{\vx})-\ma\veta}{\next{\vWeight}\sqrt{\vec{\phi}''(\next{\vx})}}}_{\next{\vWeight}+\infty}+t_{1}\mixedNorm{\frac{\vc-\vd}{\next{\vWeight}\sqrt{\vec{\phi}''(\next{\vx})}}}{\next{\vWeight}}\\
 & \leq & \delta_{t_{1}}^{\vd}(\next{\vx},\next{\vWeight})+4U^{2}t_{1}\mixedNorm{\onesVec}{\vWeight}\\
 & = & \delta_{t_{1}}^{\vd}(\next{\vx},\next{\vWeight})+100mU^{2}t_{1}.
\end{eqnarray*}

Since we have chosen $t_{1}$ small enough, we have $\delta_{t_{1}}^{\vc}(\next{\vx},\next{\vWeight})$
is small enough to satisfy the assumption of Theorem~\ref{thm:LPSolve}.
So, we only need to prove how large $t_{2}$ should be and how small
$\epsilon_{2}$ should be in order to get $\vx$ such that $\vc^{T}\vx\leq\text{OPT}+\epsilon$.
By Lemma \ref{lem:weighted_path:duality_gap} and $\norm{\vWeight^{\text{(final)}}}\leq3m$,
we have
\[
\vc^{T}\vx_{t_{2}}\leq\text{OPT}+\frac{3m}{t_{2}}.
\]
Also, Lemma \ref{lem: distance gurantee} shows that we have
\[
\normFull{\sqrt{\vec{\phi}''(\vx_{t_{2}})}\left(\vx^{\text{(final)}}-\vx_{t_{2}}\right)}_{\infty}\leq32\epsilon_{2}c_{k}.
\]
Using $\min_{\vy}\sqrt{\vec{\phi}''(\vy)}\geq\frac{1}{U}$, we have
$\normFull{\vx^{\text{(final)}}-\vx_{t_{2}}}_{\infty}\leq32\epsilon_{2}c_{k}U$
and hence our choice of $t_{2}$ and $\epsilon_{2}$ gives the result
\[
\vc^{T}\vx^{\text{(final)}}\leq\text{OPT}+\frac{3m}{t_{2}}+32\epsilon_{2}c_{k}U^{2}\leq\text{OPT}+\epsilon.
\]
\end{proof}

\section{Linear System Solver Requirements}

\label{sec:master_thm_stable}

Throughout our preceding analysis of weighted path finding we assumed
that linear systems related to $\ma$ could be solved exactly. In
this section, we relax this assumption and discuss the effect of using
inexact linear algebra in our algorithms. 

Proving stability of the algorithms in this paper is more difficult
than the ``dual'' algorithms in Part I \cite{lsInteriorPoint} for
two reasons. First, naively each iteration of interior point requires
a linear system to be solved to to $\tilde{O}(\poly(\varepsilon/U))$
accuracy and if we need to solve each linear system independently
then the overall running time of our algorithm would depends on $\log^{2}(U/\varepsilon)$
and improving this requires further insight. Second, here we need
to maintain equality constraints which further complicates the analysis.

For the remainder of this section we assume that we have an algorithm
$\mathtt{S}_{x,w}(\vq)$ such that for any vector $\vq$ the algorithm
$\mathtt{S}_{x,w}(\vq)$ outputs a vector in $\R^{n}$ such that 
\[
\normFull{\mathtt{S}_{x,w}(\vq)-\left(\ma_{x}^{T}\mWeight^{-1}\ma_{x}\right)^{-1}\vq}_{\ma_{x}^{T}\mWeight^{-1}\ma_{x}}\leq\varepsilon_{\mathtt{S}}\normFull{\left(\ma_{x}^{T}\mWeight^{-1}\ma_{x}\right)^{-1}\vq}_{\ma_{x}^{T}\mWeight^{-1}\ma_{x}}
\]
where $\varepsilon_{\mathtt{S}}=1/m^{d}$ for some sufficiently large,
but fixed, $d$. Our goal in this section is is to show that implementing
such a $\mathtt{S}_{x,w}(\vq)$ suffices for our algorithms (Section~\ref{sub:normal_force},
\ref{sub:efficient_x_step}, \ref{sub:efficient_w_step}, \ref{sub:stable_algorithm}).
In Section \ref{sub:well_conditioned}, we show that the vector $\vq$
satisfies some stability properties that allows us to construct efficient
solver $\mathtt{S}_{x,w}(\vq)$ in later section.

\subsection{The normal force $\protect\ma\protect\veta$.\label{sub:normal_force}}

To see the problem of using inexact linear system solvers more concretely,
recall that we defined a Newton steps on $\vx\in\dInterior$ in Section~\ref{sub:weighted-path:centering}
by 
\begin{eqnarray*}
\next{\vx} & := & \vx-\frac{1}{\sqrt{\vec{\phi}''(\vx)}}\mProj_{\vx,\vWeight}\left(\frac{t\vc+\vWeight\vec{\phi}'(\vx)}{\vWeight\sqrt{\vec{\phi}''(\vx)}}\right)\\
 & = & \vx-\frac{1}{\sqrt{\vec{\phi}''(\vx)}}\left(\iMatrix-\mWeight^{-1}\ma_{x}\left(\ma_{x}^{T}\mWeight^{-1}\ma_{x}\right)^{-1}\ma_{x}^{T}\right)\left(\frac{t\vc+\vWeight\vec{\phi}'(\vx)}{\vWeight\sqrt{\vec{\phi}''(\vx)}}\right).
\end{eqnarray*}
One naive way to implement this step is to replace $\left(\ma_{x}^{T}\mWeight^{-1}\ma_{x}\right)^{-1}$
with the algorithm $\mathtt{S}_{x,w}$ . Unfortunately, this does
not necessarily work well as the norm of the vector $(t\vc-\vWeight\vec{\phi}'(\vx))/\vWeight\sqrt{\vec{\phi}''(\vx)}$
can be as large as $\Omega(\log(U/\varepsilon))$ because the current
point $\vx$ can be very close to boundary. For certain linear programs,
the parameter $t$ need to be exponentially large and therefore for
this approach to work we would need to use exponentially small $\varepsilon_{K}$.
The dual problem does not has this problem because the optimality
conditions enforce $\nabla f$ is small. However, for the primal problem
we are solving, the equality constraints puts a normal force into
the systems. Therefore, even when we are very close to the optimal
point, $\nabla f$ can be very large due to the normal force. 

To circumvent this issue, we note that if we approximately know the
normal force, then we can subtract it off from the system and only
deal with a vector of reasonable size. In this section, we try to
maintain such normal force $\ma\veta$. Recall that our algorithm
measures the quality of $\vx$ by 
\[
\delta_{t}(\vx,\vWeight)\defeq\min_{\veta\in\Rn}\normFull{\frac{\grad_{x}f_{t}(\vx,\vWeight)-\ma\veta}{\vWeight\sqrt{\vphi''(\vx)}}}_{\vWeight+\infty}.
\]
We can think $\delta_{t}$ is the size of net force of the system,
i.e. the result of subtracting the normal force $\ma\veta$ from the
total force $\nabla f$. If $\delta_{t}$ is small, we know the contact
force $\grad_{x}f_{t}(\vx,\vWeight)-\ma\veta$ is small. Therefore,
the following formula gives a more stable way to compute $\next{\vx}$:
\[
\next{\vx}:=\vx-\left(\frac{t\vc+\vWeight\vec{\phi}'(\vx)-\ma\veta}{\vWeight\vec{\phi}''(\vx)}\right)+\frac{1}{\vWeight\sqrt{\vec{\phi}''(\vx)}}\ma_{x}\mathtt{S}_{x,w}\left(\ma_{x}^{T}\left(\frac{t\vc+\vWeight\vec{\phi}'(\vx)-\ma\veta}{\vWeight\sqrt{\vec{\phi}''(\vx)}}\right)\right).
\]
Furthermore, since $\mProj_{\vx,\vWeight}\mWeight^{-1}\ma_{x}=\mZero$,
subtracting $\ma\veta$ from $\nabla_{x}f_{t}$ does not affect the
step. Therefore, if we can find that $\veta$, then we have a more
stable algorithm. 

First , we show that there is an explicit $\veta^{*}$ that can be
computed in polynomial time. 
\begin{lem}
[$\veta^*$ is good]\label{def:num:eta_star}For all $(\vx,\vWeight)$
in the algorithm and $t>0$, we define the normal force
\[
\veta_{t}^{*}(\vx,\vWeight)=(\ma_{x}^{T}\mWeight^{-1}\ma_{x})^{-1}\ma_{x}^{T}\mWeight^{-1}\sqrt{\mPhi''(\vx)}^{-1}\grad_{x}f_{t}(\vx,\vWeight).
\]
Then, we have
\[
\normFull{\frac{\grad_{x}f_{t}(\vx,\vWeight)-\ma\veta_{t}^{*}(\vx,\vWeight)}{\vWeight\sqrt{\vphi''(\vx)}}}_{\vWeight+\infty}\leq2\delta_{t}(\vx,\vWeight).
\]
\end{lem}
\begin{proof}
Using \eqref{eq:centrality_equivalence} we have
\[
\mixedNorm{\sqrt{\vphi''(\vx)}\vh_{t}(\vx,\vWeight)}{\vWeight}\leq\mixedNorm{\mProj_{\vx,\vWeight}}{\vWeight}\cdot\delta_{t}(\vx,\vWeight).
\]
The result follows from the definition of $\vh_{t}(\vx,\vWeight)$,
i.e. 
\[
-\sqrt{\vphi''(\vx)}\vh_{t}(\vx,\vWeight)=\frac{\grad_{x}f_{t}(\vx,\vWeight)-\ma\veta_{t}^{*}(\vx,\vWeight)}{\vWeight\sqrt{\vphi''(\vx)}}
\]
and the fact the $\mixedNorm{\mProj_{\vx,\vWeight}}{\vWeight}\leq2$
during the algorithm.
\end{proof}
The following lemma shows that we can improve $\veta$ effectively
using $\mathtt{S}_{x,w}$.
\begin{lem}
[$\veta$ maintenance]\label{def:num:eta_maintenance}For all $(\vx,\vWeight)$
appears in the algorithm and $t>0$, we define
\[
\next{\veta}=\veta+\mathtt{S}_{x,w}\left(\ma_{x}^{T}\mWeight^{-1}\sqrt{\mPhi''(\vx)}^{-1}\left(\grad_{x}f_{t}(\vx,\vWeight)-\ma\veta\right)\right)
\]
and
\[
\veta_{t}^{*}(\vx,\vWeight)=(\ma_{x}^{T}\mWeight^{-1}\ma_{x})^{-1}\ma_{x}^{T}\mWeight^{-1}\sqrt{\mPhi''(\vx)}^{-1}\grad_{x}f_{t}(\vx,\vWeight).
\]
If $\varepsilon_{\mathtt{S}}\leq\frac{1}{2}$, we have
\[
\normFull{\frac{\ma\left(\next{\veta}-\veta_{t}^{*}(\vx,\vWeight)\right)}{\vWeight\sqrt{\vphi''(\vx)}}}_{\mWeight}\leq\varepsilon_{\mathtt{S}}\normFull{\frac{\ma\left(\veta-\veta_{t}^{*}(\vx,\vWeight)\right)}{\vWeight\sqrt{\vphi''(\vx)}}}_{\mWeight}\enspace.
\]
\end{lem}
\begin{proof}
By the definition of $\next{\veta}$ and $\veta^{*}(\vx,\vWeight)$,
we have
\begin{eqnarray*}
 &  & \normFull{\frac{\ma\left(\next{\veta}-\veta_{t}^{*}(\vx,\vWeight)\right)}{\vWeight\sqrt{\vphi''(\vx)}}}_{\mWeight}\\
 & = & \normFull{\next{\veta}-\veta_{t}^{*}(\vx,\vWeight)}_{\ma_{x}^{T}\mWeight^{-1}\ma_{x}}\\
 & = & \normFull{\veta-\veta_{t}^{*}(\vx,\vWeight)+\mathtt{S}_{x,w}\left(\ma_{x}^{T}\mWeight^{-1}\sqrt{\mPhi''(\vx)}^{-1}\left(\grad_{x}f_{t}(\vx,\vWeight)-\ma\veta\right)\right)}_{\ma_{x}^{T}\mWeight^{-1}\ma_{x}}\\
 & = & \normFull{\left(\veta^{*}(\vx,\vWeight)-\veta\right)-\mathtt{S}_{x,w}\left(\ma_{x}^{T}\mWeight^{-1}\ma_{x}\left(\veta_{t}^{*}(\vx,\vWeight)-\veta\right)\right)}_{\ma_{x}^{T}\mWeight^{-1}\ma_{x}}\\
 & \leq & \epsilon_{\mathtt{S}}\normFull{\veta-\veta_{t}^{*}(\vx,\vWeight)}_{\ma_{x}^{T}\mWeight^{-1}\ma_{x}}=\epsilon_{\mathtt{S}}\normFull{\frac{\ma\left(\veta-\veta_{t}^{*}(\vx,\vWeight)\right)}{\vWeight\sqrt{\vphi''(\vx)}}}_{\mWeight}.
\end{eqnarray*}

\end{proof}
Using the Lemma~\ref{def:num:eta_maintenance} we show how to maintain
a good $\veta$ throughout our algorithm $\code{LPSolve}$.
\begin{lem}
[Finding $\veta$]\label{lem:good_eta}Assume $\varepsilon_{\mathtt{S}}=1/m^{d}$
for a sufficiently large constant $d$. Throughout the algorithm we
can maintain $\veta$ such that
\[
\normFull{\frac{\ma\left(\veta-\veta_{t}^{*}(\vx,\vWeight)\right)}{\vWeight\sqrt{\vphi''(\vx)}}}_{\mWeight}^{2}\leq1\enspace.
\]
by calling $\mathtt{S}_{x,w}$ an amortized constant number of times
per iteration.\end{lem}
\begin{proof}
We use $\veta_{t}^{*}(\vx,\vWeight)$, defined in \eqref{def:num:eta_star},
as the initial $\veta$. Since we compute this only once, we can compute
a very precise $\veta_{t}^{*}(\vx,\vWeight)$ for the initial points
by gradient descent and preconditioning by $\mathtt{S}_{x,w}$.

Lemma~\ref{def:num:eta_maintenance} shows that we can move $\veta$
closer to $\veta_{t}^{*}(\vx,\vWeight)$ using $\mathtt{S}_{x,w}$.
Therefore, it suffices to show that during each step of the algorithm,
$\veta$ does not move far from $\veta_{t}^{*}(\vx,\vWeight)$ by
$O(\poly(n))$, or if it does, we can find $\next{\veta}$ that does
not.

We prove this by considering the three cases of changing of $t$,
changing of $\vWeight$ and changing of $\vx$ separately.

For the changes of $t$, the proof of Lemma \ref{lem:gen:t_step}
shows that
\begin{align*}
\normFull{\frac{\grad_{x}f_{t(1+\alpha)}(\vx,\vWeight)-(1+\alpha)\ma\veta}{\vWeight\sqrt{\vphi''(\vx)}}}_{\vWeight+\infty} & \leq(1+\alpha)\normFull{\frac{\grad_{x}f_{t}(\vx,\vWeight)-\ma\veta}{\vWeight\sqrt{\vphi''(\vx)}}}_{\vWeight+\infty}+\alpha\left(1+\cnorm\sqrt{\norm{\vWeight}_{1}}\right)\\
 & \leq2\normFull{\frac{\grad_{x}f_{t}(\vx,\vWeight)-\ma\veta}{\vWeight\sqrt{\vphi''(\vx)}}}_{\vWeight+\infty}+O(\poly(m)).
\end{align*}
 Using the induction hypothesis 
\[
\normFull{\frac{\ma\left(\veta-\veta_{t}^{*}(\vx,\vWeight)\right)}{\vWeight\sqrt{\vphi''(\vx)}}}_{\mWeight}^{2}\leq1\enspace,
\]
we have
\[
\normFull{\frac{\grad_{x}f_{t(1+\alpha)}(\vx,\vWeight)-(1+\alpha)\ma\veta}{\vWeight\sqrt{\vphi''(\vx)}}}_{\vWeight+\infty}\leq2\normFull{\frac{\grad_{x}f_{t}(\vx,\vWeight)-\ma\veta_{t}^{*}(\vx,\vWeight)}{\vWeight\sqrt{\vphi''(\vx)}}}_{\vWeight+\infty}+O(\poly(m))\enspace.
\]
Now, using Lemma~\ref{def:num:eta_star}, $\delta_{t}\leq1$, we
have
\[
\normFull{\frac{\grad_{x}f_{t(1+\alpha)}(\vx,\vWeight)-(1+\alpha)\ma\veta}{\vWeight\sqrt{\vphi''(\vx)}}}_{\vWeight+\infty}=O(\poly(m)).
\]
Using Lemma \ref{def:num:eta_star} again, we have
\begin{eqnarray*}
 &  & \normFull{\frac{\ma\left(\veta_{\next t}^{*}(\vx,\vWeight)-(1+\alpha)\veta\right)}{\vWeight\sqrt{\vphi''(\vx)}}}_{\vWeight+\infty}\\
 & \leq & \normFull{\frac{\grad_{x}f_{t(1+\alpha)}(\vx,\vWeight)-(1+\alpha)\ma\veta}{\vWeight\sqrt{\vphi''(\vx)}}}_{\vWeight+\infty}+\normFull{\frac{\grad_{x}f_{t(1+\alpha)}(\vx,\vWeight)-\veta_{\next t}^{*}(\vx,\vWeight)}{\vWeight\sqrt{\vphi''(\vx)}}}_{\vWeight+\infty}\\
 & = & O(\poly(m)).
\end{eqnarray*}
Therefore, we can set $\next{\veta}=(1+\alpha)\veta$ and this yields
$\veta_{\next t}^{*}(\vx,\vWeight)$ is polynomial close to $\next{\veta}$.

For the changes of $\vWeight$, the proof of Lemma \ref{lem:gen:w_change}
shows that 
\begin{eqnarray*}
\normFull{\frac{\grad_{x}f_{t}(\vx,\next{\vWeight})-\ma\veta}{\next{\vWeight}\sqrt{\vphi''(\vx)}}}_{\vWeight+\infty} & \leq & (1+4\epsilon)\left(\normFull{\frac{\grad_{x}f_{t}(\vx,\vWeight)-\ma\veta}{\vWeight\sqrt{\vphi''(\vx)}}}_{\vWeight+\infty}+\epsilon\right).
\end{eqnarray*}
Hence, by similar argument above, we have
\[
\normFull{\frac{\ma\left(\veta^{*}(\vx,\next{\vWeight})-\veta\right)}{\next{\vWeight}\sqrt{\vphi''(\vx)}}}_{\vWeight+\infty}=O(\poly(m)).
\]
Therefore, we can set $\next{\veta}=\veta$ and this gives $\veta^{*}(\vx,\next{\vWeight})$
is polynomial close to $\next{\veta}$.

For the changes of $\vx$, the proof of Lemma \ref{lem:gen:x_progress}
shows that 
\[
\normFull{\frac{\grad_{x}f_{t}(\next{\vx},\vWeight)-\ma\veta^{*}(\next{\vx},\vWeight)}{\vWeight\sqrt{\vphi''(\next{\vx})}}}_{\vWeight+\infty}\leq4\left(\delta_{t}(\vx,\vWeight)\right)^{2}=O(\poly(m)).
\]
It is easy to show that
\begin{eqnarray*}
 &  & \normFull{\frac{\ma\left(\veta^{*}(\next{\vx},\vWeight)-\veta\right)}{\vWeight\sqrt{\vphi''(\next{\vx})}}}_{\vWeight+\infty}\\
 & \leq & \normFull{\frac{\grad_{x}f_{t}(\next{\vx},\vWeight)-\ma\veta^{*}(\next{\vx},\vWeight)}{\vWeight\sqrt{\vphi''(\next{\vx})}}}_{\vWeight+\infty}+\normFull{\frac{\grad_{x}f_{t}(\next{\vx},\vWeight)-\ma\veta^{*}(\next{\vx},\vWeight)}{\vWeight\sqrt{\vphi''(\next{\vx})}}}_{\vWeight+\infty}\\
 &  & +\normFull{\frac{\ma\left(\veta^{*}(\next{\vx},\vWeight)-\veta\right)}{\vWeight\sqrt{\vphi''(\next{\vx})}}}_{\vWeight+\infty}\\
 & \leq & \poly(m).
\end{eqnarray*}
Therefore, we can set $\next{\veta}=\veta$ and this gives $\veta^{*}(\next{\vx},\vWeight)$
is polynomial close to $\next{\veta}$.

Consequently, in all cases, we can find $\next{\veta}$ such that
gives $\veta^{\text{*(new)}}$ is polynomial close to $\next{\veta}$.
Applying Lemma~\ref{def:num:eta_maintenance}, we can then obtain
a $\veta$ such that it is close to $\next{\veta}$ in $\norm{\cdot}_{w}$
norm. Note that in the first iteration we need to call $\mathtt{S}_{x,w}$
$O(\log(U/\varepsilon))$ time. Therefore, in average, we only call
$\mathtt{S}_{x,w}$ constant many times in average per iteration.
\end{proof}

\subsection{An efficient $\protect\vx$ step\label{sub:efficient_x_step}}

Having such ``normal vector'' $\veta$, we can implement $\vx$
step efficiently. Note that here we crucially use the assumption $\varepsilon_{\mathtt{S}}<C/m^{2}$.
\begin{lem}
[Efficient $\vx$ step]\label{lem:x_stable_update}For all $(\vx,\vWeight)$
in the algorithm and $t>0$ let 
\[
\next{\vx}=\vx-\frac{1}{\sqrt{\vec{\phi}''(\vx)}}\left(\iMatrix-\mWeight^{-1}\ma_{x}\left(\ma_{x}^{T}\mWeight^{-1}\ma_{x}\right)^{-1}\ma_{x}^{T}\right)\left(\frac{t\vc+\vWeight\vec{\phi}'(\vx)}{\vWeight\sqrt{\vec{\phi}''(\vx)}}\right)
\]
and
\[
\vx^{\text{(apx)}}=\vx-\left(\frac{t\vc+\vWeight\vec{\phi}'(\vx)-\ma\veta}{\vWeight\vec{\phi}''(\vx)}\right)+\frac{1}{\vWeight\sqrt{\vec{\phi}''(\vx)}}\ma_{x}\mathtt{S}_{x,w}\left(\ma_{x}^{T}\left(\frac{t\vc+\vWeight\vec{\phi}'(\vx)-\ma\veta}{\vWeight\sqrt{\vec{\phi}''(\vx)}}\right)\right).
\]
We have
\[
\normFull{\left(\mPhi''(\vx)\right)^{1/2}\left(\next{\vx}-\vx^{\text{(apx)}}\right)}_{\vWeight+\infty}\leq\tilde{O}\left(m\varepsilon_{\mathtt{S}}\right)
\]
and
\[
\delta_{t}(\vx^{\text{(apx)}},\vWeight)\leq\left(1+\tilde{O}\left(m\varepsilon_{\mathtt{S}}\right)\right)\delta_{t}(\next{\vx},\vWeight)+\tilde{O}\left(m\varepsilon_{\mathtt{S}}\right).
\]
\end{lem}
\begin{proof}
Note that
\begin{align*}
\mWeight^{1/2}\left(\mPhi''(\vx)\right)^{1/2}\left(\next{\vx}-\vx^{\text{(apx)}}\right) & =\mWeight^{-1/2}\ma_{x}\left(\ma_{x}^{T}\mWeight^{-1}\ma_{x}\right)^{-1}\ma_{x}^{T}\left(\frac{t\vc+\vWeight\vec{\phi}'(\vx)-\ma\veta}{\vWeight\sqrt{\vec{\phi}''(\vx)}}\right)\\
 & \enspace\enspace-\mWeight^{-1/2}\ma_{x}\mathtt{S}_{x,w}\left(\ma_{x}^{T}\left(\frac{t\vc+\vWeight\vec{\phi}'(\vx)-\ma\veta}{\vWeight\sqrt{\vec{\phi}''(\vx)}}\right)\right)
\end{align*}
Therefore, we have{\footnotesize{}
\begin{eqnarray*}
 &  & \normFull{\left(\mPhi''(\vx)\right)^{1/2}\left(\next{\vx}-\vx^{\text{(apx)}}\right)}_{\mWeight}\\
 & \leq & \normFull{\left(\ma_{x}^{T}\mWeight^{-1}\ma_{x}\right)^{-1}\ma_{x}^{T}\left(\frac{t\vc+\vWeight\vec{\phi}'(\vx)-\ma\veta}{\vWeight\sqrt{\vec{\phi}''(\vx)}}\right)-\mathtt{S}_{x,w}\left(\ma_{x}^{T}\left(\frac{t\vc+\vWeight\vec{\phi}'(\vx)-\ma\veta}{\vWeight\sqrt{\vec{\phi}''(\vx)}}\right)\right)}_{\ma_{x}^{T}\mWeight^{-1}\ma_{x}}\\
 & \leq & \varepsilon_{\mathtt{S}}\normFull{\left(\ma_{x}^{T}\mWeight^{-1}\ma_{x}\right)^{-1}\ma_{x}^{T}\mWeight^{-1/2}\left(\frac{t\vc+\vWeight\vec{\phi}'(\vx)-\ma\veta}{\sqrt{\vWeight\vec{\phi}''(\vx)}}\right)}_{\ma_{x}^{T}\mWeight^{-1}\ma_{x}}\\
 & = & \varepsilon_{\mathtt{S}}\normFull{\frac{t\vc+\vWeight\vec{\phi}'(\vx)-\ma\veta}{\sqrt{\vWeight\vec{\phi}''(\vx)}}}_{\mWeight^{-1/2}\ma_{x}\left(\ma_{x}^{T}\mWeight^{-1}\ma_{x}\right)^{-1}\ma_{x}^{T}\mWeight^{-1/2}}\leq\varepsilon_{\mathtt{S}}\normFull{\frac{t\vc+\vWeight\vec{\phi}'(\vx)-\ma\veta}{\vWeight\sqrt{\vec{\phi}''(\vx)}}}_{\mWeight}\\
 & \leq & \varepsilon_{\mathtt{S}}\normFull{\frac{t\vc+\vWeight\vec{\phi}'(\vx)-\ma\veta^{*}}{\vWeight\sqrt{\vec{\phi}''(\vx)}}}_{\mWeight}+\varepsilon_{\mathtt{S}}\normFull{\frac{\ma(\veta^{*}-\veta)}{\vWeight\sqrt{\vec{\phi}''(\vx)}}}_{\mWeight}\leq2\varepsilon_{\mathtt{S}}\delta_{t}(\vx,\vWeight)+\varepsilon_{\mathtt{S}}\enspace.
\end{eqnarray*}
}where the last line comes from Lemma~\ref{def:num:eta_star} and
Lemma~\ref{lem:good_eta}. Hence, we have
\[
\normFull{\left(\mPhi''(\vx)\right)^{1/2}\left(\next{\vx}-\vx^{\text{(apx)}}\right)}_{\mWeight}^{2}\leq3\varepsilon_{\mathtt{S}}.
\]
Therefore, we have
\[
\normFull{\left(\mPhi''(\vx)\right)^{1/2}\left(\next{\vx}-\vx^{\text{(apx)}}\right)}_{\vWeight+\infty}^{2}\leq\tilde{O}\left(m\varepsilon_{\mathtt{S}}\right).
\]

For the last assertion, take $\vq$ such that 
\[
\delta_{t}(\next{\vx},\vWeight)=\normFull{\frac{\grad_{x}f_{t}(\next{\vx},\vWeight)-\ma\vq}{\vWeight\sqrt{\vphi''(\next{\vx})}}}_{\vWeight+\infty}.
\]
Following similar analysis as in Lemma~\ref{lem:gen:x_progress},
we have
\begin{eqnarray*}
\delta_{t}(\vx^{\text{(apx)}},\vWeight) & \leq & \normFull{\frac{\grad_{x}f_{t}(\vx^{\text{(apx)}},\vWeight)-\ma\vq}{\vWeight\sqrt{\vphi''(\vx^{\text{(apx)}})}}}_{\vWeight+\infty}\\
 & \leq & \normFull{\frac{\sqrt{\vphi''(\next{\vx})}}{\sqrt{\vphi''(\vx^{\text{(apx)}})}}}_{\infty}\normFull{\frac{\grad_{x}f_{t}(\next{\vx},\vWeight)-\ma\vq}{\vWeight\sqrt{\vphi''(\next{\vx})}}}_{\vWeight+\infty}\\
 &  & +\normFull{\frac{\grad_{x}f_{t}(\next{\vx},\vWeight)-\grad_{x}f_{t}(\vx^{\text{(apx)}},\vWeight)}{\vWeight\sqrt{\vphi''(\vx^{\text{(apx)}})}}}_{\vWeight+\infty}\\
 & = & \left(1+\tilde{O}\left(m\varepsilon_{\mathtt{S}}\right)\right)\delta_{t}(\next{\vx},\vWeight)+\tilde{O}\left(m\varepsilon_{\mathtt{S}}\right).
\end{eqnarray*}

\end{proof}
The above lemma shows that $\vx^{\text{(apx)}}$ can be used to replace
$\next{\vx}$ without hurting $\delta_{t}$ too much. Also, the step
size $\vx^{\text{(apx)}}-\vx$ is almost the same as the step size
of $\next{\vx}-\vx$. Thus, we can implement the $\vx$ step without
using $\left(\ma_{x}^{T}\mWeight^{-1}\ma_{x}\right)^{-1}$ and using
$\mathtt{S}_{x,w}$ instead.

Unfortunately, there is one additional problem with the this algorithm,
it does not ensure $\ma^{T}\vx^{\text{(apx)}}=\vb$. Therefore, we
need to ensure $\ma^{T}\vx^{\text{(apx)}}\approx\vb$ during the algorithm.
Note that we cannot make $\ma^{T}\vx=\vb$ exactly using this approach
and consequently we need measure the infeasibility. We define
\[
I(\vx,\vWeight)\defeq\normFull{\ma^{T}\vx-\vb}_{\left(\ma_{x}^{T}\mWeight^{-1}\ma_{x}\right)^{-1}}.
\]

\begin{lem}
\label{lem:x_infeasiblity_changes} For all $(\vx,\vWeight)$ in the
algorithm define $\vx^{\text{(apx)}}$ as in Lemma~\ref{lem:x_stable_update}.
Then, we have
\[
I(\vx^{\text{(apx)}},\vWeight)\leq2I(\vx,\vWeight)+3\varepsilon_{\mathtt{S}}.
\]
\end{lem}
\begin{proof}
Since $\normFull{\left(\mPhi''(\vx)\right)^{1/2}(\vx^{\text{(apx)}}-\vx)}_{\infty}$
is small, it is easy to show 
\[
\normFull{\ma^{T}\vx^{\text{(apx)}}-\vb}_{\left(\ma_{\vx^{\text{(apx)}}}^{T}\mWeight^{-1}\ma_{\vx^{\text{(apx)}}}\right)^{-1}}\leq2\normFull{\ma\vx^{\text{(apx)}}-\vb}_{\left(\ma_{x}^{T}\mWeight^{-1}\ma_{x}\right)^{-1}}.
\]
Then, note that{\footnotesize{}
\begin{eqnarray*}
 &  & \normFull{\ma^{T}\vx^{\text{(apx)}}-\vb}_{\left(\ma_{x}^{T}\mWeight^{-1}\ma_{x}\right)^{-1}}\\
 & = & \normFull{\ma^{T}\vx-\ma^{T}\left(\frac{t\vc+\vWeight\vec{\phi}'(\vx)-\ma\veta}{\vWeight\vec{\phi}''(\vx)}\right)+\ma_{x}^{T}\mWeight\ma_{x}\mathtt{S}_{x,w}\left(\ma_{x}^{T}\left(\frac{t\vc+\vWeight\vec{\phi}'(\vx)-\ma\veta}{\vWeight\sqrt{\vec{\phi}''(\vx)}}\right)\right)-\vb}_{\left(\ma_{x}^{T}\mWeight^{-1}\ma_{x}\right)^{-1}}\\
 & \leq & I(\vx,\vWeight)+\normFull{(\ma_{x}^{T}\mWeight\ma_{x})^{-1}\ma^{T}\left(\frac{t\vc+\vWeight\vec{\phi}'(\vx)-\ma\veta}{\vWeight\vec{\phi}''(\vx)}\right)-\mathtt{S}_{x,w}\left(\ma_{x}^{T}\left(\frac{t\vc+\vWeight\vec{\phi}'(\vx)-\ma\veta}{\vWeight\sqrt{\vec{\phi}''(\vx)}}\right)\right)}_{\ma_{x}^{T}\mWeight\ma_{x}}\\
 & \leq & I(\vx,\vWeight)+\epsilon_{\mathtt{S}}\normFull{(\ma_{x}^{T}\mWeight\ma_{x})^{-1}\ma^{T}\left(\frac{t\vc+\vWeight\vec{\phi}'(\vx)-\ma\veta}{\vWeight\vec{\phi}''(\vx)}\right)}_{\ma_{x}^{T}\mWeight\ma_{x}}\\
 & \leq & I(\vx,\vWeight)+\epsilon_{\mathtt{S}}\normFull{\frac{t\vc+\vWeight\vec{\phi}'(\vx)-\ma\veta}{\vWeight\sqrt{\vec{\phi}''(\vx)}}}_{\mWeight}.
\end{eqnarray*}
}Now, we bound the last term using Lemma \ref{def:num:eta_star} and
\ref{lem:good_eta} as follows
\begin{eqnarray*}
\normFull{\frac{t\vc+\vWeight\vec{\phi}'(\vx)-\ma\veta}{\vWeight\sqrt{\vec{\phi}''(\vx)}}}_{\mWeight} & \leq & \normFull{\frac{t\vc+\vWeight\vec{\phi}'(\vx)-\ma\veta^{*}(\vx,\vWeight)}{\vWeight\sqrt{\vec{\phi}''(\vx)}}}_{\mWeight}+\normFull{\frac{\ma(\veta^{*}(\vx,\vWeight)-\veta)}{\vWeight\sqrt{\vec{\phi}''(\vx)}}}_{\mWeight}\\
 & \leq & 2\delta_{t}(\vx,\vWeight)+1\leq3.
\end{eqnarray*}

\end{proof}
Note that when we change $\vWeight$ so long as no coordinate changes
by more than a multiplicative constant then $I(\vx,\vWeight)$ changes
by at most a multiplicative constant and thus no further proof on
the stability of $I(\vx,\vWeight)$ with respect to $\vWeight$ is
needed. 

Now, we show how to improve $I(\vx,\vWeight)$.
\begin{lem}
[Improve Feasibility]\label{lem:x_feasiblity_ensure}Given $(\vx,\vWeight)$
appears in the algorithm. Define
\[
\vx^{\text{(fixed})}=\vx-\sqrt{\mPhi''(\vx)}^{-1}\mWeight^{-1}\ma_{x}\mathtt{S}_{x,w}\left(\ma\vx-\vb\right).
\]
Assume that $I(\vx,\vWeight)\leq0.01m^{-1}$, we have
\[
I(\vx^{\text{(fixed})},\vWeight)\leq2\varepsilon_{\mathtt{S}}I(\vx,\vWeight).
\]
Furthermore, $\normFull{\left(\mPhi''(\vx)\right)^{1/2}(\vx^{\text{(fixed})}-\vx)}_{\vWeight+\infty}\leq O\left(mI(\vx,\vWeight)\right)$.\end{lem}
\begin{proof}
Note that
\begin{eqnarray*}
\normFull{\left(\mPhi''(\vx)\right)^{1/2}(\vx^{\text{(fixed})}-\vx)}_{\mWeight} & = & \normFull{\mathtt{S}_{x,w}\left(\ma\vx-\vb\right)}_{\ma_{x}^{T}\mWeight^{-1}\ma_{x}}\\
 & \leq & (1+\varepsilon_{\mathtt{S}})\normFull{\left(\ma_{x}^{T}\mWeight^{-1}\ma_{x}\right)^{-1}\left(\ma\vx-\vb\right)}_{\ma_{x}^{T}\mWeight^{-1}\ma_{x}}\\
 & \leq & 2I(\vx,\vWeight).
\end{eqnarray*}
Hence, $\normFull{\left(\mPhi''(\vx)\right)^{1/2}(\vx^{\text{(fixed})}-\vx)}_{\infty}\leq2mI(\vx,\vWeight).$
By the assumption, $\normFull{\left(\mPhi''(\vx)\right)^{1/2}(\vx^{\text{(fixed})}-\vx)}_{\infty}$
is very small and hence one can show that 
\[
\normFull{\ma^{T}\vx^{\text{(fixed})}-\vb}_{\left(\ma_{\vx^{\text{(fixed})}}^{T}\mWeight^{-1}\ma_{\vx^{\text{(fixed})}}\right)^{-1}}\leq2\normFull{\ma^{T}\vx^{\text{(fixed})}-\vb}_{\left(\ma_{x}^{T}\mWeight^{-1}\ma_{x}\right)^{-1}}.
\]

Now, we note that
\begin{eqnarray*}
 &  & \normFull{\ma^{T}\vx^{\text{(fixed})}-\vb}_{\left(\ma_{x}^{T}\mWeight^{-1}\ma_{x}\right)^{-1}}\\
 & = & \normFull{\left(\ma_{x}^{T}\mWeight^{-1}\ma_{x}\right)^{-1}\left(\ma^{T}\vx-\vb\right)-\mathtt{S}_{x,w}\left(\ma^{T}\vx-\vb\right)}_{\ma_{x}^{T}\mWeight^{-1}\ma_{x}}\\
 & \leq & \varepsilon_{\mathtt{S}}I(\vx,\vWeight).
\end{eqnarray*}

\end{proof}
Since $\epsilon_{S}$ is sufficiently small this lemma implies that
the given step improves feasibility by much more than it hurts centrality.
Therefore, by applying this step periodically throughout our algorithm
we can maintain the invariant the the infeasibility is small.

\subsection{An efficient $\protect\vWeight$ step\label{sub:efficient_w_step}}

There are two computations performed by our algorithm involving the
weights. The first is in the ``chasing 0'' game for centering we
are given approximate weights and then need to change the weights.
However, here there is no linear system that is solved. The second
place, is in the computing of these approximate weights. However,
here we just need to use approximate linear system solvers to approximate
leverage scores and we discussed how to do this in Part I \cite[ArXiv v3, Section D]{lsInteriorPoint}.

\subsection{The stable algorithm\label{sub:stable_algorithm}}

We summarize the section as follows:
\begin{thm}
\label{thm:LPSolve_detailed_stable}Suppose we have an interior point
$\vx\in\dInterior$ for the for the linear program \eqref{eq:lp}and
suppose that for any diagonal positive definite matrix $\md$ and
vector $\vq$, we can find $\vx$ in $\mathcal{T}_{w}$ work and $\mathcal{T}_{d}$
depth such that 
\[
\normFull{\vx-\left(\ma^{T}\md\ma\right)^{-1}\vq}_{\ma^{T}\md\ma}\leq\epsilon_{\mathtt{S}}\normFull{\left(\ma^{T}\md\ma\right)^{-1}\vq}_{\ma^{T}\md\ma}
\]
for $\epsilon_{\mathtt{S}}=1/m^{k}$ for some large constant $k$.
Then, using $\code{LPSolve}$ we can compute $\vx$ such that $\vc^{T}\vx\leq\text{OPT}+\epsilon$,
$\norm{\ma^{T}\vx-\vb}_{\ma^{T}\ms^{-2}\ma}\leq\epsilon$, and for
all $i\in[n]$ $l_{i}\leq x_{i}\leq u_{i}$ in $\tilde{O}\left(\sqrt{\rank(\ma)}\left(\mathcal{T}_{w}+\nnz(\ma)\right)\log\left(U/\epsilon\right)\right)$
work and $\tilde{O}\left(\sqrt{\rank(\ma)}\mathcal{T}_{d}\log\left(U/\epsilon\right)\right)$
depth where $U\defeq\max\left(\normFull{\frac{\vu-\vl}{\vu-\vx_{0}}}_{\infty},\normFull{\frac{\vu-\vl}{\vx_{0}-\vl}}_{\infty},\norm{\vu-\vl}_{\infty},\norm{\vc}_{\infty}\right)$
and $\ms$ is a diagonal matrix with $\ms_{ii}=\min(x_{i}-l_{i},u_{i}-x_{i}).$ \end{thm}
\begin{proof}
Lemma~\ref{lem:good_eta} shows that we can maintain $\veta$ which
is close to $\veta^{*}$ defined in Lemma~\ref{def:num:eta_star}.
Lemma~\ref{lem:x_stable_update} shows that using this $\veta$,
we can compute a more numerically stable step $\vx^{\text{(apx)}}$.
Hence, this gives us a way to implement $\vx$ step using $\mathtt{S}_{x,w}$.
In the previous subsection, we explained how to implement $\vWeight$
step using $\mathtt{S}_{x,w}$.

To deal with infeasibility, Lemma~\ref{lem:x_infeasiblity_changes}
shows that the stable step $\vx^{\text{(apx)}}$ does not hurt the
infeasibility $I(\vx,\vWeight)$ too much. It is also easy to show
the step for $\vWeight$ does not hurt the infeasibility too much.
Whenever $I(\vx,\vWeight)>1/m^{2}$, we improve the feasibility using
Lemma~\ref{lem:x_feasiblity_ensure}. This decreases the infeasibility
a lot while only taking a very small step as shown in Lemma~\ref{lem:x_feasiblity_ensure}
and consequently it does not hurt the progress $\delta_{t}$ and $\Phi_{\mu}$.

Therefore, Theorem~\ref{thm:LPSolve_detailed} can be implemented
using the necessary inexact linear algebra. To get the bound on $\norm{\ma^{T}\vx-\vb}_{\ma^{T}\ms^{-2}\ma}$,
we use Lemma~\ref{lem:gen:phi_properties_sim} to show that $\ms\preceq\Phi''(\vx)$,
therefore $\norm{\ma^{T}\vx-\vb}_{\ma^{T}\ms^{-2}\ma}\leq\norm{\ma^{T}\vx-\vb}_{\ma_{x}^{T}\ma_{x}}=I(\vx,\vWeight).$
\end{proof}
For some problems, we need a dual solution instead of the primal.
We prove how to do this in the following theorem. In the proof we
essentially show that the normal force we maintain for numerical stability
is essentially a dual solution.
\begin{thm}
\label{thm:LPSolve_detailed_stable_dual}Suppose we have an initial
$\vx_{0}$ such that $\ma^{T}\vx_{0}=\vb$ and $-1\leq[\vx_{0}]_{i}\leq1$
and suppose that for any diagonal positive definite matrices $\md$
and vectors $\vq$, we can find $\vx$ from such that 
\[
\normFull{\vx-\left(\ma^{T}\md\ma\right)^{-1}\vq}_{\ma^{T}\md\ma}\leq\varepsilon_{\mathtt{S}}\normFull{\left(\ma^{T}\md\ma\right)^{-1}\vq}_{\ma^{T}\md\ma}
\]
for $\varepsilon_{\mathtt{S}}=1/m^{k}$ for sufficiently large constant
$k$ in $\mathcal{T}_{w}$ work and $\mathcal{T}_{d}$ depth. Then,
there is an algorithm that compute $\vy$ such that 
\[
\vb^{T}\vy+\norm{\ma\vy+\vc}_{1}\leq\min_{\vy}\left(\vb^{T}\vy+\norm{\ma\vy+\vc}_{1}\right)+\epsilon.
\]
in $\tilde{O}(\sqrt{\rank(\ma)}\left(\mathcal{T}_{w}+\nnz(\ma)\right)\log\left(U/\epsilon\right))$
work and $\tilde{O}(\sqrt{\rank(\ma)}\mathcal{T}_{d}\log\left(U/\epsilon\right))$
depth where $U\defeq\max\left(\normFull{\frac{2}{1-\vx}}_{\infty},\normFull{\frac{2}{\vx+1}}_{\infty},\norm{\vc}_{\infty}\right)$~. \end{thm}
\begin{proof}
We can use our algorithm to solve the following linear program
\[
\min_{\ma^{T}\vx=\vb,-1\leq x_{i}\leq1}\vc^{T}\vx
\]
and find $(\vx,\vWeight,\veta)$ such that
\begin{equation}
\normFull{\frac{t\vc+\vWeight\vec{\phi}'(\vx)-\ma\veta}{\vWeight\sqrt{\vphi''(\vx)}}}_{\vWeight+\infty}\leq\delta\label{eq:optimality_condition}
\end{equation}
and
\[
I(\vx,\vb)=\normFull{\ma^{T}\vx-\vb}_{\left(\ma_{x}^{T}\mWeight^{-1}\ma_{x}\right)^{-1}}\leq\delta
\]
for some small $\delta$ and large $t$. To use this to derive a dual
solution it seems we need to be very close to the central path. Thus
we use our algorithm to compute a central path point for a particular
$t$ and then, we do an extra $\tilde{O}(\log(mU/\varepsilon))$ iterations
to ensure that the error $\delta$ is as small as $1/\poly(mU/\varepsilon))$.Let
$\vy=-\veta/t$, $\vec{\lambda}=\vc+\ma\vy$ and $\vec{\tau}=\vec{\lambda}+\frac{1}{t}\vWeight\vec{\phi}'(\vx)$
where $\phi(x)=-\log\cos(\frac{\pi x}{2})$ because the constraints
are all $-1<x_{i}<1$. By the definition of $\vec{\lambda}$, we have
\[
\left\langle \vec{\lambda},\vx\right\rangle =\left\langle \vc,\vx\right\rangle +\left\langle \ma\vy,\vx\right\rangle .
\]
We claim that :
\begin{enumerate}
\item $\left|\left\langle \ma\vy,\vx\right\rangle -\left\langle \vy,\vb\right\rangle \right|\leq\left(\frac{\delta+2m}{t}+2mU\right)\delta$.
\item $\left|\left\langle \vc,\vx\right\rangle +\min_{\vy}\left(\left\langle \vb,\vy\right\rangle +\norm{\vc+\ma\vy}_{1}\right)\right|\leq\left(\frac{1}{t}+\delta\right)\poly(mU).$
\item $\left|\left\langle \vec{\lambda},\vx\right\rangle +\norm{\vc+\ma\vy}_{1}\right|\leq\left(\delta+\frac{1}{t}\right)\poly(m).$
\end{enumerate}
Using these claims we can compute a very centered point for $t=\frac{1}{\delta}=(mU)^{k}/\varepsilon$
for sufficiently large $k$ and get the result
\[
\left\langle \vy,\vb\right\rangle +\norm{\vc+\ma\vy}_{1}\leq\min_{\vy}\left(\left\langle \vb,\vy\right\rangle +\norm{\vc+\ma\vy}_{1}\right)+\varepsilon.
\]

Claim (1): Note that
\begin{eqnarray*}
 &  & \left|\left\langle \ma\vy,\vx\right\rangle -\left\langle \vy,\vb\right\rangle \right|\\
 & \leq & \normFull{\vy}_{\left(\ma_{x}^{T}\mWeight^{-1}\ma_{x}\right)}\normFull{\ma^{T}\vx-\vb}_{\left(\ma_{x}^{T}\mWeight^{-1}\ma_{x}\right)^{-1}}\\
 & = & \frac{1}{t}\normFull{\frac{\ma\veta}{\vWeight\sqrt{\vphi''(\vx)}}}_{\mWeight}I(\vx,\vb)\\
 & \leq & \frac{1}{t}\left(\normFull{\frac{t\vc+\vWeight\vec{\phi}'(\vx)-\ma\veta}{\vWeight\sqrt{\vphi''(\vx)}}}_{\mWeight}+\normFull{\frac{t\vc}{\vWeight\sqrt{\vphi''(\vx)}}}_{\mWeight}+\normFull{\frac{\vWeight\vec{\phi}'(\vx)}{\vWeight\sqrt{\vphi''(\vx)}}}_{\mWeight}\right)\delta\\
 & \leq & \left(\frac{\delta}{t}+\normFull{\frac{\vc}{\vWeight\sqrt{\vphi''(\vx)}}}_{\mWeight}+\frac{\normFull{\onesVec}_{\mWeight}}{t}\right)\delta.
\end{eqnarray*}
Since $\phi(x)=-\log\cos(\frac{\pi x}{2})$, $\phi''(x)\geq\pi^{2}/4.$
Thus, we have
\[
\left|\left\langle \ma\vy,\vx\right\rangle -\left\langle \vy,\vb\right\rangle \right|\leq\left(\frac{\delta+2m}{t}+2mU\right)\delta.
\]

Claim (2): From the proof of Theorem~\ref{thm:LPSolve}, we see that
\[
\left|\vc^{T}\vx-\min_{\ma^{T}\vx=\vb,-1\leq x_{i}\leq1}\vc^{T}\vx\right|\leq\poly(mU)\left(\frac{1}{t}+\delta\right).
\]
Since there is an interior point for $\{\ma^{T}\vx=\vb,-1\leq x_{i}\leq1\}$
and the set is bounded, the strong duality shows that 
\begin{eqnarray*}
 &  & \min_{\ma^{T}\vx=\vb,-1\leq x_{i}\leq1}\vc^{T}\vx\\
 & = & \min_{\vx}\max_{\vec{\lambda}^{(1)}\geq0,\vec{\lambda}^{(2)}\geq0,\vy}\vc^{T}\vx+\left\langle \vy,\ma^{T}\vx-\vb\right\rangle +\left\langle \vec{\lambda}^{(1)},\vx-\onesVec\right\rangle +\left\langle \vec{\lambda}^{(2)},-\onesVec-\vx\right\rangle \\
 & = & \max_{\vy,\vec{\lambda}^{(1)}\geq0,\vec{\lambda}^{(2)}\geq0}\min_{\vx}\left\langle \vc+\ma\vy+\vec{\lambda}^{(1)}-\vec{\lambda}^{(2)},\vx\right\rangle -\left\langle \vb,\vy\right\rangle -\left\langle \vec{\lambda}^{(1)}+\vec{\lambda}^{(2)},\onesVec\right\rangle \\
 & = & -\min_{\vy}\left\langle \vb,\vy\right\rangle +\norm{\vc+\ma\vy}_{1}
\end{eqnarray*}
yielding the claim.

Claim (3): Recall that $\vec{\tau}=\vec{\lambda}+\frac{1}{t}\vWeight\vec{\phi}'(\vx)$.
Hence, we have
\[
\tau_{i}=\lambda_{i}+\frac{\pi}{2t}w_{i}\tan(\frac{\pi}{2}x_{i}).
\]
Therefore, we have
\[
x_{i}=\frac{2}{\pi}\tan^{-1}\left(\frac{2t}{\pi w_{i}}\left(\tau_{i}-\lambda_{i}\right)\right).
\]
Thus, we have
\begin{eqnarray*}
\left\langle \vec{\lambda},\vx\right\rangle  & = & \sum_{i}\lambda_{i}\frac{2}{\pi}\tan^{-1}\left(\frac{2t}{\pi w_{i}}\left(\tau_{i}-\lambda_{i}\right)\right)\\
 & = & -\frac{2}{\pi}\sum_{i}\lambda_{i}\tan^{-1}\left(\frac{2t}{\pi w_{i}}\left(\lambda_{i}-\tau_{i}\right)\right).
\end{eqnarray*}
Thus, Lemma~\ref{lem:x_atan} shows that
\[
-\norm{\vec{\lambda}}_{1}\leq\left\langle \vec{\lambda},\vx\right\rangle \leq-\norm{\vec{\lambda}}_{1}+2\norm{\vec{\tau}}_{1}+\frac{\norm{\vWeight}_{1}}{t}.
\]
Recall that $\norm{\vWeight}_{1}=O(m)$. Also, the bound \eqref{eq:optimality_condition}
is equivalent to
\[
\left|\tau_{i}\right|\leq\frac{\delta w_{i}\sqrt{\phi''(x_{i})}}{t}\leq\frac{\delta\sqrt{\phi''(x_{i})}}{t}.
\]
Lemma \ref{lem:slacks_bound} shows that the slack of central point
is larger than $\poly(m)/t.$ Therefore, Lemma~\ref{lem:gen:phi_properties_sim}
shows that $\frac{\delta\sqrt{\phi''(x_{i})}}{t}\leq\poly(m)\delta$.
Therefore $\norm{\vec{\tau}}_{1}=\poly(m)\delta$. Using $\vec{\lambda}=\vc+\ma\vy$,
we have $\left|\left\langle \vec{\lambda},\vx\right\rangle +\norm{\vc+\ma\vy}_{1}\right|\leq\poly(m)\left(\delta+\frac{1}{t}\right).$\end{proof}
\begin{rem}
Note how this algorithm uses the initial point to certify that $\min_{\vy}\vb^{T}\vy+\norm{\ma\vy+\vc}_{1}$
is bounded. As usual, one can use standard technique to avoid the
requirement on the initial point.
\end{rem}

\subsection{Well conditioned\label{sub:well_conditioned}}

For many problems, the running time of linear system solvers depend
on the condition number and/or how fast the linear systems change
from iteration to iteration. The following lemma shows that our interior
point method enjoys many properties frequently exploited in other
interior point methods and therefore is amenable to different techniques
for improving iteration costs. In particular, here we bound the condition
number of the matrices involved which in turn, allows us to use the
fast M matrix solver in next section.

There are two key lemmas we prove in this section. First, in Lemma~\ref{lem:slacks_bound}
we bound how close the weighted central path can go to the boundary
of the polytope. This allows us to reason about how ill-conditioned
the linear system we need to solver become over the course of the
algorithm. Weshows that if the slacks, i.e. distances to the boundary
of the polytope, of the initial point are polynomially bounded below
and if we only change the weight multiplicatively by a polynomial
factor, then the slacks of the new weighted central path point is
still polynomially bounded below. Second, in Lemma~\ref{lem:sequence_system_changes}
we bound how much the linear systems can change over the course of
our algorithm. 
\begin{lem}
\label{lem:slacks_bound} For all $\vWeight\in\dWeight$ and $t>0$
let $\vx_{t,\vWeight}=\arg\min f_{t}\left(\vx,\vWeight\right)$. For
all $\vx\in\dInterior$ and $i\in[m]$ let $s_{i}(\vx)$ denote the
slack of constraint $i$, i.e. $s_{i}(\vx)\defeq\min\{u_{i}-x_{i},x_{i}-l_{i}\}$.
For any $a,b>0$ and $\vWeight^{(1)},\vWeight^{(2)}\in\dWeight$ and
$i\in[m]$ we have 
\begin{equation}
s_{i}\left(\vx_{b,\vWeight^{(2)}}\right)\geq\min\left\{ \frac{\left(\min_{j\in[m]}w_{j}^{(2)}\right)\cdot\left(\min_{j\in[m]}s_{j}\left(\vx_{b,\vWeight^{(1)}}\right)\right)}{2\left(\frac{b}{a}\norm{\vWeight^{(1)}}_{1}+\norm{\vWeight^{(2)}}_{1}\right)},1\right\} s_{i}\left(\vx_{a,\vWeight^{(1)}}\right).\label{lem:master:slack_bound}
\end{equation}
\end{lem}
\begin{proof}
Fix an arbitrary $i\in[m]$ and consider the straight line from $\vx_{a,\vWeight^{(1)}}$
to $\vx_{b,\vWeight^{(2)}}$. If this line never reaches a point $\vy$
such that $s_{i}(\vy)=0$ then $s_{i}(\vx_{b,\vWeight^{(2)}})\geq s_{i}(\vx_{a,\vWeight^{(2)}})$
and clearly \eqref{lem:master:slack_bound}. Otherwise, we can parameterize
the the straight line by $\vp(t)$ such that $\vp(-1)=\vx_{a,\vWeight^{(1)}}$,
$s_{i}(\vp(0))=0$, and $\vp(-\theta)=\vx_{b,\vWeight^{(2)}}$ for
some $\theta\in[0,1]$. Since $\phi_{i}(p(t))\rightarrow\infty$ as
$t\rightarrow0$, Lemma~\ref{lem:gen:phi_properties_sim} shows that
\[
\left.\frac{d^{2}\phi_{i}}{dt^{2}}\right|_{t}\geq\frac{1}{t^{2}}.
\]
Integrating then yields that.
\begin{eqnarray*}
\left.\frac{d\phi_{i}}{dt}\right|_{t=-\theta} & \geq & \left.\frac{d\phi_{i}}{dt}\right|_{t=-1}+\int_{-1}^{-\theta}\frac{1}{t^{2}}dt\\
 & = & \left.\frac{d\phi_{i}}{dt}\right|_{t=-1}+\left(\frac{1}{\theta}-1\right).
\end{eqnarray*}
Since each of the $\phi_{j}$ is convex, we have
\begin{eqnarray*}
\sum_{j\in[m]}w_{j}^{(2)}\left.\frac{d\phi_{j}}{dt}\right|_{t=-\theta} & \geq & \sum_{j\in[m]}w_{j}^{(2)}\left.\frac{d\phi_{j}}{dt}\right|_{t=-1}+\left(\min_{j\in[m]}w_{j}^{(2)}\right)\cdot\left(\frac{1}{\theta}-1\right).
\end{eqnarray*}
Using the optimality condition of $\vx_{b,\vWeight^{(2)}}$ and the
optimality condition of $\vx_{a,\vWeight^{(1)}}$, we have
\begin{eqnarray*}
\frac{b}{a}\sum_{j\in[m]}w_{j}^{(1)}\left.\frac{d\phi_{i}}{dt}\right|_{t=-1} & \geq & \sum_{j\in[m]}w_{j}^{(2)}\left.\frac{d\phi_{i}}{dt}\right|_{t=-1}+\left(\min_{j\in[m]}w_{j}^{(2)}\right)\left(\frac{1}{\theta}-1\right).
\end{eqnarray*}
Hence,
\[
\left(\frac{b}{a}\norm{\vWeight^{(1)}}_{1}+\norm{\vWeight^{(2)}}_{1}\right)\max_{j\in[m]}\left|\left(\left.\frac{d\phi_{j}}{dt}\right|_{t=-1}\right)\right|\geq\left(\min_{j\in[m]}w_{j}^{(2)}\right)\left(\frac{1}{\theta}-1\right).
\]
Applying Lemma~\ref{lem:gen:phi_properties_force} again yields that
for all $j\in[m]$ 
\[
\left|\left(\left.\frac{d\phi_{j}}{dt}\right|_{t=-1}\right)\right|s_{j}\left(\vx_{b,\vWeight^{(1)}}\right)\leq1.
\]
Thus, we have
\[
\frac{b}{a}\norm{\vWeight^{(1)}}_{1}+\norm{\vWeight^{(2)}}_{1}\geq\left(\min_{j\in[m]}w_{j}^{(2)}\right)\left(\frac{1}{\theta}-1\right)\min_{j}s_{j}\left(\vx_{b,\vWeight^{(1)}}\right).
\]
Hence, 
\[
\theta\geq\frac{\left(\min_{j\in[m]}w_{j}^{(2)}\right)\cdot\left(\min_{j\in[m]}s_{j}\left(\vx_{b,\vWeight^{(1)}}\right)\right)}{2\left(\frac{b}{a}\norm{\vWeight^{(1)}}_{1}+\norm{\vWeight^{(2)}}_{1}\right)}.
\]
Since $i\in[m]$ was arbitrary we have the desired result.\end{proof}
\begin{lem}
\label{lem:sequence_system_changes} Using the notations and assumptions
in Theorem~\ref{thm:LPSolve_detailed_stable} or Theorem~\ref{thm:LPSolve_detailed_stable_dual}
let $\ma^{T}\md_{k}\ma$ be the $k^{th}$ linear system that is used
in the algorithm $\code{LPSolve}$. For all $k\geq1$, we have the
following:
\begin{enumerate}
\item The condition number of $\md_{k}$ is bounded by $\poly(mU/\varepsilon)$,
i.e., $\poly(\varepsilon/(mU))\ma^{T}\ma\preceq\ma^{T}\md_{k}\ma\preceq\poly(mU/\varepsilon)\ma^{T}\ma$
\item $\norm{\log(\md_{k+1})-\log(\md_{k})}_{\infty}\leq1/10$.
\item $\norm{\log(\md_{k+1})-\log(\md_{k})}_{\mSigma_{\ma}(\vd_{k})}\leq1/10$.
\end{enumerate}
\end{lem}
\begin{proof}
During the algorithm, the matrix we need to solve is of the form $\ma^{T}\md\ma$
where $\md=\mWeight^{-1}\mPhi''(\vx)^{-1}$. We know that $\frac{n}{2m}\leq\vWeight_{i}\leq3$.
In the proof of Theorem \ref{thm:LPSolve_detailed}, we showed that
$\vec{\phi}_{i}''(\vx)\geq\frac{1}{U^{2}}$. Also, Lemma \ref{lem:slacks_bound}
shows that the slacks is never too small and hence $\vec{\phi}_{i}''(\vx)$
is upper bounded by $\poly(mU/\varepsilon)$. Thus, the condition
number of $\md$ is bounded by $\poly(mU/\varepsilon)$.

Now, we bound the changes of $\md$ by bound the changes of $\mPhi''(\vx)$
and the changes of $\mWeight$ separately. For the changes of $\mPhi''(\vx)$,
\eqref{eq:centrality_equivalence} shows that $\mixedNorm{\sqrt{\vphi''(\vx)}\vh_{t}(\vx,\vWeight)}{\vWeight}\leq\mixedNorm{\mProj_{\vx,\vWeight}}{\vWeight}\delta_{t}.$
Since $\mixedNorm{\mProj_{\vx,\vWeight}}{\vWeight}\leq2$ and $\delta_{t}\leq1/80$,
we have 
\begin{eqnarray*}
\mixedNorm{\sqrt{\vphi''(\vx)}(\next{\vx}-\vx)}{\vWeight} & = & \mixedNorm{\sqrt{\vphi''(\vx)}\vh_{t}(\vx,\vWeight)}{\vWeight}\leq1/40.
\end{eqnarray*}
Using this on Lemma \ref{lem:gen:phi_properties_sim}, we have
\begin{eqnarray*}
\normFull{\log\left(\vphi''(\next{\vx})\right)-\log\left(\vphi''(\vx)\right)}_{\vWeight+\infty} & \leq & \left(1-\mixedNorm{\sqrt{\vphi''(\vx)}(\next{\vx}-\vx)}{\vWeight}\right)^{-1}-1\\
 & \leq & 1/36.
\end{eqnarray*}
Since $\vWeight_{i}\geq\frac{1}{2}\vsigma_{i}$ for all $i$, we have
\begin{equation}
\normFull{\log\left(\vphi''(\next{\vx})\right)-\log\left(\vphi''(\vx)\right)}_{\vsigma+\infty}\leq1/20.\label{eq:changes_phi}
\end{equation}

For the changes of $\mWeight$, we look at the description of $\code{centeringInexact}$.
The algorithm ensures the changes of $\log(\vWeight)$ is in $(1+\epsilon)U$
where $U=\{\vx\in\Rm~|~\mixedNorm{\vx}{\vWeight}\leq\left(1-\frac{7}{8c_{k}}\right)\delta_{t}\}$.
Since $\delta_{t}\leq1/80$ and $\vWeight_{i}\geq\frac{1}{2}\vsigma_{i}$
for all $i$, we get that
\begin{equation}
\normFull{\log\left(\next{\vWeight}\right)-\log\left(\vWeight\right)}_{\vsigma+\infty}\leq1/20.\label{eq:changes_of_w}
\end{equation}
The assertion (2) and (3) follows from \eqref{eq:changes_phi} and
\eqref{eq:changes_of_w}.\end{proof}

\section{Generalized Minimum Cost Flow\label{sec:Applications}}

In this section we show how to use the interior point method in Section~\ref{sec:master_thm}
to solve the maximum flow problem in time $\otilde(m\sqrt{n}\log^{O(1)}(U))$,
to solve the minimum cost flow problem in time, $\otilde(m\sqrt{n}\log^{O(1)}(U))$,
and to compute $\epsilon$-approximate solutions to the lossy generalized
minimum cost flow problem in time $\otilde(m\sqrt{n}\log^{O(1)}(U/\epsilon))$.
Our algorithm for the generalized minimum cost flow problem is essentially
the same as our algorithm for the simpler specific case of minimum
cost flow and maximum flow and therefore, we present the algorithm
for the generalized minimum cost flow problem directly. %
\footnote{Our algorithm could be simplified slightly for the simpler cases and
the dependence on polylogarithmic factors for these problems could
possibly be improved.%
}

The generalized minimum cost flow problem \cite{daitch2008faster}
is as follows. Let $G=(V,E)$ be a connected directed graph where
each edge $e$ has capacity $c_{e}>0$ and multiplier $1\geq\gamma_{e}>0$.
For each edge $e$, there can be only at most $c_{e}$ units of flow
on that edge and the flow on that edge must be non-negative. Also,
for each unit of flow entering edge $e$, there are only $\gamma_{e}$
units of flow going out. The generalized maximum flow problem is to
compute how much flow can be sent into $t$ given a unlimited source
$s$. The generalized minimum cost flow is to ask what is the minimum
cost of sending the maximum flow given the cost of each edge is $q_{e}$.
The maximum flow and the minimum cost flow are the case with $\gamma_{e}=1$
for all edges $e$.

Since the generalized minimum cost flow includes all of these cases,
we focus on this general formulation. The problem can be written as
the following linear program
\[
\min_{\vzero\leq\vx\leq\vc}\vq^{T}\vx\text{ such that }\ma\vx=F\onesVec_{t}
\]
where $F$ is the generalized maximum flow value, $\onesVec_{t}$
is a indicator vector of size $(n-1)$ that is non-zero at vertices
$t$ and $\ma$ is a $\left|V\backslash\{s\}\right|\times\left|E\right|$
matrix such that for each edge $e$, we have 
\begin{eqnarray*}
\ma(e_{head},e) & = & \gamma(e),\\
\ma(e_{tail},e) & = & -1.
\end{eqnarray*}
In order words, the constraint $\ma x=F\onesVec_{t}$ requires the
flow to satisfies the flow conversation at all vertices except $s$
and $t$ and requires it flows $F$ unit of flow to $t$. We assume
$c_{e}$ are integer and $\gamma_{e}$ is a rational number. Let $U$
be the maximum of $c_{e}$, $q_{e}$, the numerator of $\gamma_{e}$
and the denominator of $\gamma_{e}$. For the generalized flow problems,
getting an efficient exact algorithm is difficult and we aim for approximation
algorithms only.
\begin{defn}
\label{def:flow_def}We call a flow an $\epsilon-$approximate generalized
maximum flow if it is a flow satisfies the flow conservation and the
flow value is larger than maximum flow value minus $\epsilon$. We
call a flow is an $\epsilon-$approximate generalized minimum cost
maximum flow if it is an $\epsilon$-approximate maximum flow and
has cost not greater than the minimum cost maximum flow value.
\end{defn}
Note that $\rank\left(\ma\right)=n-1$ because the graph is connected
and hence our algorithm takes only $\tilde{O}(\sqrt{n}L)$ iterations.
Therefore, the problems remaining are to compute $L$ and bound how
much time is required to solve the linear systems involved. However,
$L$ is large in the most general setting and hence we cannot use
the standard theory to say how to get the initial point, how to round
to the vertex. Furthermore, the condition number of $\ma^{T}\ma$
can be very bad.

In \cite{daitch2008faster}, they used dual path following to solve
the generalized minimum cost flow problem with the caveats that the
dual polytope is not bounded, the problem of getting the initial flow,
the problem of rounding it to the a feasible flow. We use there analysis
to formulate the problem in a manner more amenable to our algorithms.
Since we are doing the primal path following, we will state a reformulation
of the LP slightly different. 
\begin{thm}
[\cite{daitch2008faster}]\label{thm:flowreduction} Given a directed
graph $G$. We can find a new directed graph $\tilde{G}$ with $O(m)$
edges and $O(n)$ vertices in $\tilde{O}(m)$ time such that the modified
linear program 
\[
\min_{0\leq x_{i}\leq c_{i},0\leq y_{i}\leq4mU^{2},0\leq z_{i}\leq4mU^{2}}\vq^{T}\vx+\frac{256m^{5}U^{5}}{\epsilon^{2}}\left(\onesVec^{T}\vy+\onesVec^{T}\vz\right)\text{ such that }\ma\vx+\vy-\vz=F\onesVec_{t}
\]
satisfies the following conditions:
\begin{enumerate}
\item $\vx=\frac{c}{2}\onesVec$, $\vy=2mU^{2}\onesVec-(\ma\frac{c}{2}\onesVec)^{-}+F\vec{1}_{t}$,
$z=2mU^{2}\onesVec+(\ma\frac{c}{2}\onesVec)^{+}$ is an interior point
of the linear program.
\item Given any $(\vx,\vy,\vz)$ such that $\norm{\ma\vx+\vy-\vz}_{2}\leq\frac{\epsilon^{2}}{128m^{2}n^{2}U^{3}}$
and with cost value within $\frac{\epsilon^{2}}{128m^{2}n^{2}U^{3}}$
of the optimum. Then, one can compute an $\epsilon$-approximate minimum
cost maximum flow for graph $G$ in time $\tilde{O}(m)$.
\item The linear system of the linear program is well-conditioned, i.e.,
the condition number of $\left[\begin{array}{ccc}
\ma & \iMatrix & -\iMatrix\end{array}\right]\left[\begin{array}{c}
\ma^{T}\\
\iMatrix\\
-\iMatrix
\end{array}\right]$ is $O(mU)$.
\item The linear system of the linear program can be solve in nearly linear
time, i.e. for any diagonal matrix $\ms$ with condition number $\kappa$
and vector $b$, it takes $\tilde{O}\left(m\log\left(\frac{\kappa U}{\delta}\right)\right)$
time to find $x$ such that
\[
\norm{x-\mvar L^{-1}b}_{\mvar L}\leq\delta\norm x_{\mvar L}
\]
where $\mvar L=\left[\begin{array}{ccc}
\ma & \iMatrix & -\iMatrix\end{array}\right]\ms\left[\begin{array}{c}
\ma^{T}\\
\iMatrix\\
-\iMatrix
\end{array}\right].$
\end{enumerate}
\end{thm}
The main difference between what stated in \cite{daitch2008faster}
and here is that 
\begin{enumerate}
\item Our linear program solver can support constraint $l_{i}\leq x_{i}\leq u_{i}$
and hence we do not need to split the flow variable to positive part
and negative part.
\item Our linear program solver is primal path following and hence we add
the constraint $y_{i}\leq4mU^{2}$ and $z_{i}\leq4mU^{2}$. Since
the maximum flow value is at most $mU^{2}$, it does not affect the
optimal solution of the linear program.
\item We remove the variable $\mathbf{x}_{3}$ in \cite{daitch2008faster}
because the purpose of that is to make the dual polytope is bounded
and we do not need it here.
\end{enumerate}
Using the reduction mentioned above, one can obtain the promised generalized
minimum cost flow algorithm.
\begin{thm}
There is a randomized algorithm to compute an $\varepsilon-$approximate
generalized minimum cost maximum flow in $\tilde{O}(\sqrt{n}\log^{O(1)}(U/\epsilon))$
depth $\otilde(m\sqrt{n}\log^{O(1)}\left(U/\epsilon\right))$ total
work (see Definition~\ref{def:flow_def}). Furthermore, there is
an algorithm to compute an exact standard minimum cost maximum flow
in $\tilde{O}(\sqrt{n}\log^{O(1)}(U))$ depth and $\otilde(m\sqrt{n}\log^{O(1)}\left(U\right))$
total work.\end{thm}
\begin{proof}
Using the reduction above and Theorem \ref{thm:LPSolve_detailed},
we get an algorithm of generalized minimum cost flow by solving $\otilde(\sqrt{n})$
linear systems to $\tilde{\mathcal{O}}\left(1\right)$ bit accuracy
and the condition number of those systems are $\poly(mU/\epsilon)$.
In \cite{daitch2008faster}, they showed that the linear system involved
can be reduced to $\tilde{O}(\log(U/\epsilon))$ many Laplacian systems
and hence we can use a recent nearly linear work polylogarithmic depth
Laplacian system solver of Spielman and Peng \cite{peng2013efficient}.
In total, it takes $\otilde(m\log^{O(1)}\left(\frac{U}{\epsilon}\right))$
time to solve each systems.

For the standard minimum cost maximum flow problem, it is known that
the solution set is a convex polytope with integer coordinates and
we can use Isolation lemma to make sure there is unique minimum. Hence,
we only need to take $\epsilon=\poly(1/mU)$ and round the solution
to the closest integer. See Section 3.5 in \cite{daitch2008faster}
for details. \end{proof}

\section{Acknowledgments}

We thank Yan Kit Chim, Andreea Gane, and Jonathan A. Kelner for many
helpful conversations. This work was partially supported by NSF awards
0843915 and 1111109, NSF Graduate Research Fellowship (grant no. 1122374)
and Hong Kong RGC grant 2150701. 

\bibliographystyle{plain}
\bibliography{main}

\appendix

\section{Glossary}

\label{sec:glossary}

Here we summarize various linear programming specific notation that
we use throughout the paper. For many quantities we included the typical
order of magnitude as they appear during our algorithms.
\begin{itemize}
\item Linear program related: constraint matrix $\ma\in\R^{m\times n}$
, cost vector $\vc\in\mathbb{R}^{m}$, constraint vector $\vb\in\mathbb{R}^{n}$,
solution $\vx\in\mathbb{R}^{m}$, weights of constraints $\vWeight\in\mathbb{R}^{m}$
where $m$ is the number of variables and $n$ is the number of constraints.
\item Matrix version of variables: $\ms$ is the diagonal matrix corresponds
to $\vs$, $\mWeight$ corresponds to $\vWeight$, $\mPhi$ corresponds
to $\phi$.
\item Penalized objective function (\ref{eq:penalized_objective}): $\penalizedObjective(\vx,\vWeight)=t\cdot\vc^{T}\vx+\sum_{i\in[m]}\vWeight_{i}\phi_{i}(\vx_{i}).$
\item Barrier functions (Sec \ref{sub:Preliminaries:The-Problem}): For
$[l,\infty)$, we use $\phi(x)=-\log(x-l)$. For $(-\infty,u]$, we
use $\phi(x)=-\log(u-x)$. For $[l,u]$, we use $\phi(x)=-\log(ax+b)$
where $a=\frac{\pi}{u-l}$ and $b=-\frac{\pi}{2}\frac{u+l}{u-l}.$
\item The projection matrix $\mProj_{\vx,\vWeight}$ (\ref{eq:def_Pxw}):
$\mProj_{\vx,\vWeight}=\iMatrix-\mWeight^{-1}\ma_{x}\left(\ma_{x}^{T}\mWeight^{-1}\ma_{x}\right)^{-1}\ma_{x}^{T}$
where $\ma_{x}\defeq\mPhi''(\vx)^{-1/2}\ma$. 
\item Newton step (\ref{eq:newton_step}): $\vNewtonStep_{t}(\vx,\vWeight)=-\mPhi''(\vx)^{-1/2}\Pxw\mWeight^{-1}\mPhi''(\vx)^{-1/2}\grad_{x}f_{t}(\vx,\vWeight).$
\item The mixed norm (\ref{eq:mixed_norm}): $\mixedNorm{\vy}{\vWeight}=\norm{\vy}_{\infty}+\cnorm\norm{\vy}_{\mWeight}$
where $\cnorm\approx\polylog(m).$
\item Centrality (\ref{eq:centrality_definition}): $\delta_{t}(\vx,\vWeight)=\min_{\veta\in\Rn}\normFull{\frac{\grad_{x}f_{t}(\vx,\vWeight)-\ma\veta}{\vWeight\sqrt{\vphi''(\vx)}}}_{\vWeight+\infty}\approx\frac{1}{\polylog(m)}.$
\item Properties of weight function (Def \ref{def:gen:weight_function}):
\textbf{size} $\cWeightSize(\fvWeight)=\normOne{\fvWeight(\vx)}\approx\rank\left(\ma\right)$,
\textbf{slack sensitivity} $\cWeightStab(\fvWeight)=\mixedNorm{\mProj_{\vx,\vWeight}}{\vWeight}\approx1+\frac{1}{\polylog(m)}$,
\textbf{step consistency} $\cWeightCons(\fvWeight)\approx1-\frac{1}{\polylog(m)}$.
\item Difference between $\vg$ and $\vWeight$ (\ref{eq:def_psi}): $\vWeightError(\vx,\vWeight)=\log(\fvWeight(\vx))-\log(\vWeight).$
\item Potential function for tracing $0$ (Thm \ref{thm:zero_game}): $\Phi_{\mu}(\vx)=e^{\mu x}+e^{-\mu x}\approx\poly(m)$.
\item The weight function proposed (\ref{eq:sec:weights:weight_function}):
\[
\vg(\vx)=\argmin_{\vWeight\in\rPos^{m}}\penalizedObjectiveWeight(\vx,\vWeight)\enspace\text{ where }\enspace\penalizedObjectiveWeight(\vx,\vWeight)=\onesVec^{T}\vWeight+\frac{1}{\alpha}\log\det\left(\ma_{x}^{T}\mWeight^{-\alpha}\ma_{x}\right)-\beta\sum_{i}\log\weight_{i}
\]
where $\ma_{x}=(\mPhi''(\vx))^{-1/2}\ma$, $\alpha\approx1+1/\log_{2}\left(\frac{m}{\rank(\ma)}\right)$,
$\beta\approx\rank(\ma)/m$.
\end{itemize}

\section{Appendix}

\subsection{Technical Lemmas}
\begin{lem}
\label{lem:max_flow:projection_lemma} For any norm $\norm{\cdot}$
and $\norm{\vy}_{Q}\defeq\min_{\veta\in\Rn}\normFull{\vy-\frac{\ma\veta}{\vWeight\sqrt{\vec{\phi}''(\vx)}}}$,
we have
\[
\norm{\vy}_{Q}\leq\normFull{\Pxw\vy}\leq\normFull{\Pxw}\cdot\normFull{\vy}_{Q}.
\]
\end{lem}
\begin{proof}
By definition $\Pxw\vy=\vy-\frac{\ma\veta_{y}}{\vWeight\sqrt{\vec{\phi}''(\vx)}}$
for some $\veta_{y}\in\Rn$. Consequently,
\begin{eqnarray*}
\norm{\vy}_{Q} & = & \min_{\veta\in\Rn}\normFull{\vy-\frac{\ma\eta}{\vWeight\sqrt{\vec{\phi}''(\vx)}}}\leq\normFull{\Pxw\vy}.
\end{eqnarray*}
On the other hand, let $\vec{\eta}_{q}$ by such that such that $\norm{\vy}_{Q}=\normFull{\vy-\frac{\ma\veta_{q}}{\vWeight\sqrt{\vec{\phi}''(\vx)}}}.$
Then, since $\Pxw\mWeight^{-1}(\mPhi'')^{-1/2}\ma=\mZero$, we have
\begin{eqnarray*}
\normFull{\Pxw\vy} & = & \normFull{\Pxw\left(\vy-\frac{\ma\veta_{q}}{\vWeight\sqrt{\vec{\phi}''}}\right)}\leq\normFull{\Pxw}\cdot\normFull{\vy-\frac{\ma\veta_{q}}{\vWeight\sqrt{\vec{\phi}''}}}=\normFull{\Pxw}\cdot\normFull{\vy}_{Q}.
\end{eqnarray*}
\end{proof}
\begin{lem}[{Log Notation \cite[Appendix]{lsInteriorPoint}}]
 \label{lem:appendix:log_helper} Suppose $\left|\log(a)-\log\left(b\right)\right|=\epsilon\leq1/2$
then $\left|\frac{a-b}{b}\right|\leq\epsilon+\epsilon^{2}$. If $\left|\frac{a-b}{b}\right|=\epsilon\leq1/2$,
then $\left|\log\left(a\right)-\log\left(b\right)\right|\leq\epsilon+\epsilon^{2}.$
\end{lem}

\begin{lem}[{\cite[Appendix]{lsInteriorPoint}}]
\label{lem:appendix:projection_matrices} For any projection matrix
$\mProj\in\Rmm$, $\mSigma=\mDiag(\mProj)$, $i,j\in[m]$, $\vx\in\Rm$,
and $\vWeight\in\rPos^{m}$ we have 
\begin{itemize}
\item $\mSigma_{ii}=\sum_{j\in[m]}\mProj_{ij}^{(2)},$ 
\item $\mZero\specLeq\mProj^{(2)}\specLeq\mSigma\specLeq\iMatrix$, 
\item $\mProj_{ij}^{(2)}\leq\mSigma_{ii}\mSigma_{jj}$, 
\item $|\indicVec i^{T}\mProj^{(2)}\vx|\leq\mSigma_{ii}\norm{\vx}_{\mSigma}$. 
\item $\grad_{\vWeight}\log\det(\ma^{T}\mWeight\ma)=\mLever_{\ma}(\vWeight)\vWeight^{-1}.$
\item \textup{$\jacobian_{\vWeight}(\vLever_{\ma}(\vWeight))=\mLapProj_{\ma}(\vWeight)\mWeight^{-1}$.}
\end{itemize}
\end{lem}

\begin{lem}
\label{lem:x_atan}For any $x,\varepsilon$ and $\lambda>0$, we have
\[
\frac{\pi}{2}\left|x\right|-\pi\varepsilon-\frac{1}{\lambda}\leq x\tan^{-1}(\lambda(x+\varepsilon))\leq\frac{\pi}{2}\left|x\right|.
\]
\end{lem}
\begin{proof}
We first consider the case $\varepsilon=0$. Note that
\[
x\tan^{-1}(\lambda x)\leq\frac{\pi}{2}\left|x\right|.
\]
Also, we note that
\[
x\tan^{-1}(\lambda x)\geq\left|x\right|\left(\frac{\pi}{2}-\frac{1}{\lambda\left|x\right|}\right)
\]
because
\[
\left|\tan(\frac{\pi}{2}-\frac{1}{\lambda\left|x\right|})\right|=\left|\frac{\cos(\frac{1}{\lambda\left|x\right|})}{\sin(\frac{1}{\lambda\left|x\right|})}\right|\leq\lambda\left|x\right|.
\]
Hence, we have
\[
\frac{\pi}{2}\left|x\right|-\frac{1}{\lambda}\leq x\tan^{-1}(\lambda x)\leq\frac{\pi}{2}\left|x\right|.
\]

For $\varepsilon\neq0$, we have
\[
\frac{\pi}{2}\left|x+\varepsilon\right|-\frac{1}{\lambda}\leq(x+\varepsilon)\tan^{-1}(\lambda(x+\varepsilon))\leq\frac{\pi}{2}\left|x+\varepsilon\right|.
\]
Thus, we have
\[
\frac{\pi}{2}\left|x\right|-\pi\varepsilon-\frac{1}{\lambda}\leq x\tan^{-1}(\lambda(x+\varepsilon))\leq\frac{\pi}{2}\left|x\right|.
\]

\end{proof}

\subsection{Projection on Mixed Norm Ball}

\label{sub: projection_mixed_ball}In the \cite{lsInteriorPoint},
we studied the following problem:
\begin{equation}
\max_{\norm{\vx}_{2}\leq1,-l_{i}\leq x_{i}\leq l_{i}}\left\langle \va,\vx\right\rangle \label{eq:projection_problem}
\end{equation}
for some given vector $\va$ and $\vl$ in $\Rm$. We proved that
the following algorithm outputs a solution of (\ref{eq:projection_problem})
in depth $\tilde{O}(1)$ and work $\tilde{O}(m)$. 

\begin{center}
\begin{tabular}{|l|}
\hline 
$\vx=\code{projectOntoBallBoxParallel}(\va,\vl)$\tabularnewline
\hline 
\hline 
1. Set $\va=\va/\norm{\va}_{2}$. \tabularnewline
\hline 
2. Sort the coordinate such that $\left|a_{i}\right|/l_{i}$ is in
descending order.\tabularnewline
\hline 
3. Precompute $\sum_{k=0}^{i}l_{k}^{2}$ and $\sum_{k=0}^{i}a_{k}^{2}$
for all $i$. \tabularnewline
\hline 
4. Find the first $i$ such that $\frac{1-\sum_{k=0}^{i}l_{k}^{2}}{1-\sum_{k=0}^{i}a_{k}^{2}}\leq\frac{l_{i}^{2}}{a_{i}^{2}}$.\tabularnewline
\hline 
5. Output $\vx_{j}=\begin{cases}
\text{sign}\left(a_{j}\right)l_{j} & \text{if }j\in\{1,2,\cdots,i\}\\
\sqrt{\frac{1-\sum_{k=0}^{i}l_{k}^{2}}{1-\sum_{k=0}^{i}a_{k}^{2}}}\va_{j} & \text{otherwise}
\end{cases}.$\tabularnewline
\hline 
\end{tabular}
\par\end{center}

In this section, we show that the algorithm above can be transformed
to solve the problem
\begin{equation}
\max_{\norm{\vx}_{\vWeight}+\norm{\vx}_{\infty}\leq1}\left\langle \va,\vx\right\rangle \label{eq:project_our_form}
\end{equation}
for some given vector $\va$ and $\vWeight>0$. To do this, let study
(\ref{eq:projection_problem}) more closely. Without loss of generality,
we can assume $\norm{\va}_{2}=1$ and $\left|a_{i}\right|/l_{i}$
is in descending order. The key consequence of $\code{projectOntoBallBoxParallel}$
is that the problem (\ref{eq:projection_problem}) always has a solution
of the form
\begin{equation}
\vx_{l,a}^{(i_{t})}=\begin{cases}
\text{sign}\left(a_{j}\right)l_{j} & \text{if }j\in\{1,2,\cdots,i_{t}\}\\
\sqrt{\frac{1-\sum_{k=0}^{i_{t}}l_{k}^{2}}{1-\sum_{k=0}^{i_{t}}a_{k}^{2}}}\va_{j} & \text{otherwise}
\end{cases}.\label{eq:project_x_form}
\end{equation}
where $i_{t}$ be the first coordinate such that
\[
\frac{1-t^{2}\sum_{k=0}^{i}l_{k}^{2}}{1-\sum_{k=0}^{i}a_{k}^{2}}\leq\frac{t^{2}l_{i}^{2}}{a_{i}^{2}}.
\]
Note that $i_{t}\geq i_{s}$ if $t\leq s$. Therefore, we have that
the set of $t$ such that $i_{t}=j$ is simply%
\footnote{There are some boundary cases we ignored for simplicity.%
}
\begin{equation}
\frac{\left|a_{j}\right|}{\sqrt{l_{j}^{2}\left(1-\sum_{k=0}^{j}a_{k}^{2}\right)+a_{j}^{2}\sum_{k=0}^{j}l_{k}^{2}}}\leq t<\frac{\left|a_{j-1}\right|}{\sqrt{l_{j-1}^{2}\left(1-\sum_{k=0}^{j-1}a_{k}^{2}\right)+a_{j-1}^{2}\sum_{k=0}^{j-1}l_{k}^{2}}}.\label{eq:proj_t_interval}
\end{equation}

Define the function $f$ by
\[
f(t)=\max_{\norm{\vx}_{2}\leq1,-tl_{i}\leq x_{i}\leq tl_{i}}\left\langle \va,\vx\right\rangle .
\]
We know that
\begin{eqnarray*}
f(t) & = & \left\langle \va,\vx_{tl,a}^{(i_{t})}\right\rangle \\
 & = & t\sum_{j=1}^{i_{t}}\left|a_{j}\right|\left|l_{j}\right|+\sqrt{1-t^{2}\sum_{k=0}^{i_{t}}l_{k}^{2}}\sqrt{1-\sum_{k=0}^{i_{t}}a_{k}^{2}}.
\end{eqnarray*}
Therefore, we have
\begin{eqnarray*}
\max_{\norm{\vx}_{2}+\norm{\vec{l}^{-1}\vx}_{\infty}\leq1}\left\langle \va,\vx\right\rangle  & = & \max_{0\leq t\leq1}\max_{\norm{\vx}_{2}\leq1-t\text{ and }-tl_{i}\leq x_{i}\leq tl_{i}}\left\langle \va,\vx\right\rangle \\
 & = & \max_{0\leq t\leq1}(1-t)\max_{\norm{\vx}_{2}\leq1\text{ and }-\frac{t}{1-t}l_{i}\leq x_{i}\leq\frac{t}{1-t}l_{i}}\left\langle \va,\vx\right\rangle \\
 & = & \max_{0\leq t\leq1}(1-t)f(\frac{t}{1-t})\\
 & = & \max_{0\leq t\leq1}t\sum_{j=1}^{i_{t}}\left|a_{j}\right|\left|l_{j}\right|+\sqrt{(1-t)^{2}-t^{2}\sum_{k=0}^{i_{t}}l_{k}^{2}}\sqrt{1-\sum_{k=0}^{i_{t}}a_{k}^{2}}.
\end{eqnarray*}
Note that the function $t\sum_{j=1}^{i}\left|a_{j}\right|\left|l_{j}\right|+\sqrt{(1-t)^{2}-t^{2}\sum_{k=0}^{i}l_{k}^{2}}\sqrt{1-\sum_{k=0}^{i}a_{k}^{2}}$
is concave and the solution has a close form. Therefore, one can compute
the maximum value for each interval of $t$ (\ref{eq:proj_t_interval})
and find which is the best. Hence, we get the following algorithm.

\begin{center}
\begin{tabular}{|l|}
\hline 
$\vx=\code{projectOntoMixedNormBallParallel}(\va,\vl)$\tabularnewline
\hline 
\hline 
1. Set $\va=\va/\norm{\va}_{2}$. \tabularnewline
\hline 
2. Sort the coordinate such that $\left|a_{i}\right|/l_{i}$ is in
descending order.\tabularnewline
\hline 
3. Precompute $\sum_{k=0}^{i}l_{k}^{2}$, $\sum_{k=0}^{i}a_{k}^{2}$
and $\sum_{j=1}^{i}\left|a_{j}\right|\left|l_{j}\right|$ for all
$i$. \tabularnewline
\hline 
4. Let $g_{i}(t)=t\sum_{j=1}^{i}\left|a_{j}\right|\left|l_{j}\right|+\sqrt{(1-t)^{2}-t^{2}\sum_{k=0}^{i}l_{k}^{2}}\sqrt{1-\sum_{k=0}^{i}a_{k}^{2}}.$\tabularnewline
\hline 
5. For each $j\in\{1,\cdots,n\}$, Find $t_{j}=\argmax_{i_{t}=j}g_{j}(t)$
using (\ref{eq:proj_t_interval})\tabularnewline
\hline 
6. Find $i=\argmax_{i}g_{i}(t_{i}).$\tabularnewline
\hline 
7. Output $(1-t_{i})\vx_{\frac{t_{i}}{1-t_{i}}l,a}^{(i)}$ defined
by (\ref{eq:project_x_form}).\tabularnewline
\hline 
\end{tabular}
\par\end{center}

The discussion above leads to the following theorem. The problem in
the from (\ref{eq:project_our_form}) can be solved by $\code{projectOntoMixedNormBallParallel}$
and a change of variables.
\begin{thm}
The algorithm \textup{$\code{projectOntoMixedNormBallParallel}$ outputs
a solution to 
\[
\max_{\norm{\vx}_{2}+\norm{\vec{l}^{-1}\vx}_{\infty}\leq1}\left\langle \va,\vx\right\rangle 
\]
in total work $\tilde{O}(m)$ and depth $\tilde{O}(1)$.}\end{thm}

\end{document}